\documentclass[11pt,table]{article}
\usepackage[a4paper, margin=1in]{geometry}

\usepackage{graphicx}
\graphicspath{ {./images/} }

\usepackage{tcolorbox}

\usepackage{amssymb,amsmath,amsthm}

\usepackage{pgfplots}
\usepackage{subcaption}
\pgfplotsset{compat=1.18}

\usepackage{setspace}

\usepackage{booktabs}
\usepackage{arydshln} 
\usepackage{makecell}

\usepackage{thmtools,thm-restate}
\usepackage{algorithm}
\usepackage{algpseudocode}
\usepackage{hyperref}
\usepackage{url}
\usepackage{multirow}

\usepackage[numbers]{natbib}

\newtcolorbox{problemBox}{
  colback=white,
  colframe=black,
  boxrule=0.5pt,
  arc=2pt,
  left=4pt,
  right=4pt,
  top=4pt,
  bottom=4pt
}

\usepackage{xcolor, soul}

\usepackage{caption}

\newtheorem{definition}{Definition}[section]
\newtheorem{theorem}[definition]{Theorem}
\newtheorem{corollary}[definition]{Corollary}

\newtheorem{prop}[definition]{Proposition}
\newtheorem{claim}{Claim}

\usepackage{enumitem}

\usepackage{authblk}
\usepackage{orcidlink}



\begin{document}

\title{A Minimax Perspective on Almost-Stable Matchings\thanks{The statements of Theorems 4.1, 4.2, 4.4, 5.1 and 5.2 appeared in Section 2 of ``Minimax and Preferential Almost-Stable Matchings'' in the Proceedings of AAMAS 2026 \cite{glitznermanloveasaamas26}.}}
\author[ ]{Frederik Glitzner \orcidlink{0009-0002-2815-6368} and David Manlove \orcidlink{0000-0001-6754-7308}}
\affil[ ]{School of Computing Science, University of Glasgow, Glasgow G12 8QQ, UK}
\affil[ ]{\normalfont \texttt{f.glitzner.1@research.gla.ac.uk, david.manlove@glasgow.ac.uk}}
\date{}

\maketitle

\begin{abstract}
    Stability is crucial in matching markets, yet in many real-world settings -- from hospital residency allocations to roommate assignments -- full stability is either impossible to achieve or can come at the cost of leaving many agents unmatched. When stability cannot be achieved, algorithmicists and market designers face a critical question: how should instability be measured and distributed among participants? Existing approaches to ``almost-stable'' matchings focus on aggregate measures, minimising either the number of blocking pairs or the count of agents involved in blocking pairs. However, such aggregate objectives can result in concentrated instability on a few individual agents, raising concerns about fairness and incentives to deviate. We introduce a fairness-oriented approach to approximate stability based on the minimax principle: we seek matchings that minimise the maximum number of blocking pairs any agent finds themselves in. Equivalently, we minimise the maximum number of agents that anyone has justified envy towards. This distributional objective protects the worst-off agents from bearing a disproportionate amount of instability. 
    
    We characterise the computational complexity of this notion across fundamental matching settings. Surprisingly, even very modest guarantees with respect to the distribution of instability prove computationally intractable: we show that it is {\sf NP-complete} to decide whether a matching exists in which no agent is in more than one blocking pair, even when preference lists have constant-bounded length. This hardness applies to both {\sc Stable Roommates} and maximum cardinality {\sc Stable Marriage} problems. On the positive side, we provide polynomial-time algorithms when agents rank at most two others, and present approximation algorithms and integer programs for general settings. Our results map the algorithmic landscape and reveal fundamental trade-offs between distributional guarantees on justified envy and computational feasibility in matching market design.
\end{abstract}


\section{Introduction}

\subsection{Motivation}

Stability in matching markets ensures that no pair of agents would prefer to deviate from their assigned matches by becoming assigned to one another instead. This property underpins trusted real-world mechanisms from medical residency allocations to school and campus housing assignments. Yet, in many real-world settings, stability is either impossible to achieve or comes at a significant trade-off with respect to the number of matched agents. Thus, a central question for economists and computer scientists alike is whether instability can be eliminated, and how it should be measured and distributed among participants if it cannot be avoided.

The two foundational stable matching problems {\sc Stable Marriage} and {\sc Stable Roommates} were first studied by Gale and Shapley \cite{gale_shapley}. In the former, there are two disjoint sets of agents $A_1$ and $A_2$, where each agent ranks a subset of the opposite set. The latter problem generalises this: there is a single set $A$ of agents, each ranking a subset of the remaining agents. We will denote these problems by {\sc smi} and {\sc sri}, respectively, to indicate that incomplete preference lists, i.e., preferences over a strict subset of agents (implying that all unranked agents are not acceptable), are permissible. In both settings, the goal is to find a \emph{matching} consisting of disjoint pairs of mutually acceptable agents, such that no two agents prefer each other to their partners (either or both of which could be none). If such a pair does exist, we refer to it as a \emph{blocking pair}. A \emph{stable matching} avoids such pairs, ensuring that no two agents have a mutual incentive to deviate and become matched to one another. We highlight that in both {\sc smi} and {\sc sri}, the notion of a blocking pair coincides with the notion of \emph{justified envy} \cite{BALINSKI199973}.

{\sc smi} and {\sc sri} underpin many real-world allocation tasks, such as matching doctors to hospitals \cite{irving1998matching}, students to dormitories \cite{abraham_grp}, or users to communication channels \cite{noma}. Much of the literature focuses on efficient algorithms for these centralised schemes \cite{matchup}. While a stable matching for an {\sc smi} instance can always be found in linear time using the Gale-Shapley algorithm \cite{gale_shapley}, not every {\sc sri} instance admits a stable matching (see \cite{gusfield89} for an extensive structural overview). If a stable matching exists, one can be found in linear time using Irving's algorithm \cite{irving_sr}. 

Two key challenges arise in practice. Incomplete preference lists can prevent stable matchings from maximising the number of matched agents, resulting in unfilled positions and unmatched agents despite possible improvements to the size \cite{biro_sm_10}. Secondly, in non-bipartite matching settings, stable matchings often do not exist at all \cite{glitzner2025empirics}. Gale and Sotomayor \cite{galesoto85} and Gusfield and Irving \cite{gusfield89} observed that, under strict preferences, the set of matched agents is identical in all stable matchings. Consequently, there is no flexibility to increase the size of the matching while maintaining stability. A large number of unmatched agents can, in practical settings such as school admissions or job markets, be considered a significant weakness, and experimental works have highlighted the practical importance of relaxing stability for a substantial improvement in this regard (e.g., see \cite{bertsimas2025relaxstabilitymatchingmarkets} for a recent perspective). 

When forced to accept some instability, existing work has largely focused on aggregate instability measures that, for example, minimise the number of blocking pairs or minimise the number of agents involved in at least one blocking pair \cite{matchup,chen17}. While Eriksson and Häggström \cite{Eriksson2008} argued that counting blocking pairs is the most indicative measure of instability, as each such pair has both the incentive and the opportunity to deviate, aggregate approaches can permit large concentrated instability. An agent who finds themselves in ten blocking pairs, i.e., who has justified envy towards ten other agents, may rightly feel that they have been treated unfairly and, furthermore, has an exceptionally large incentive to motivate other participants to reorganise among themselves and cause unravelling. Thus, both from the perspective of individual fairness, as well as from a robustness perspective, this appears problematic: if each blocking pair represents a legitimate injustice and causes a realistic potential for deviations and unravelling, why should we allow blocking pairs to accumulate in a few individuals? Are global measures of instability really more appropriate than local, individual-focused measures of instability?

Consider, for example, the bipartite instance with four agents on each side shown twice in Figure \ref{fig:ex1}. Each copy of this instance has a perfect matching indicated in blue, and each such matching admits two blocking pairs indicated in red. Edges that are neither in the matching nor blocking the matching are indicated in black. The left copy admits three blocking agents, one of which is in two blocking pairs and two of which are in one blocking pair. The right copy admits four blocking agents, all of which are in one blocking pair. Here, conventional wisdom aiming to minimise blocking agents would suggest choosing the left matching over the right one, but fairness and robustness, in our interpretation, would suggest choosing the right one instead.

\begin{figure}[!tbh]
        \centering
        \includegraphics[width=10cm]{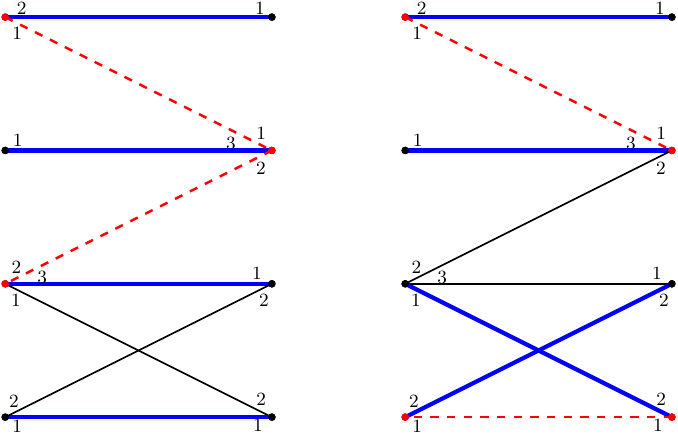}
        \caption{Two different perfect matchings with the same number of blocking pairs but different numbers of blocking agents and different levels of instability concentration.}
        \label{fig:ex1}
\end{figure}

Another reason to justify exploring individual rather than aggregate instability is \emph{exploration cost}. Suppose that agents rank a large number of other agents or positions and that rankings are private information. Then, following the release of the matching by the central matchmaker, agents may explore possible profitable deviations by contacting each agent they prefer to their match, which may come at a cost (e.g., the time it takes to contact other agents). However, when agents are guaranteed to be in only a few blocking pairs, or none, then the exploration cost to find a blocking pair can be high.\footnote{We purposefully keep this statement vague, as there are multiple factors at play here. For example, receiving a high-ranking partner leaves few better partners to block with. However, none of the discussed almost-stability criteria provide any inherent guarantees with regard to the average or maximum ranks of partners.}

In this paper, we continue research on algorithmic and complexity results related to decision and optimisation problems concerned with matchings and maximum-cardinality matchings that are ``as stable as possible'', which are also referred to as  ``almost-stable'' matchings \cite{abraham06,biro_sm_10}, and resolve open questions on the complexity of such problems \cite{chen17,chen2025fptapproximabilitystablematchingproblems,gupta2020parameterized}.

\subsection{Our Contributions}

Given the practical importance of matchings that minimise instability while maximising size, strong intractability results involving the computation of such matchings are discouraging. We therefore aim to understand the tractability frontier through a different lens. Rather than aggregate measures of instability, we adopt the classical principle of \emph{minimax fairness}: protecting the worst-off individual from excessive injustice, and simultaneously decreasing individuals' incentives to deviate and the risk of unravelling. Specifically, our interpretation of \emph{minimax} is to minimise the maximum number of blocking pairs that any agent finds themselves in, which is a new notion of almost-stability that we study in this paper. Equivalently, we minimise the maximum number of agents that any single agent has justified envy towards (under non-wastefulness). This distributional objective accepts that full stability may not be achievable, but ensures that no single agent bears a disproportionate burden of instability. Our approach aligns with fairness principles in social choice that prioritise the leximin allocation of bad outcomes, rather than maximising the utilitarian welfare. 

We characterise the computational complexity of this fairness-constrained notion of almost-stability with surprisingly strong intractability results. Specifically, we show that even deciding whether a matching (respectively maximum-cardinality matching) exists in which no agent is in more than one blocking pair is {\sf NP-complete} for {\sc sri} ({\sc smi}), even when preference lists are of short, constant-bounded length. We also give additional strong intractability results and extend these to strong inapproximability results for the corresponding optimisation problems of interest. On the positive side, we present polynomial-time algorithms for instances with short preference lists, an approximation algorithm, and exact linear programming models. Finally, we show through experiments that almost-stable matchings and maximum-cardinality matchings that distribute the burden of instability well among agents are very likely to exist.

\subsection{Structure of the Paper}

First, in Section \ref{sec:background}, we define crucial concepts of this paper and highlight relevant related works. Then, we start by introducing the minimax almost-stability notion in Section \ref{sec:minimax}. There, we also provide crucial observations about their relationship to other, previously studied, notions, and structural characterisations. In Section \ref{sec:intractability}, we highlight strong intractability results for the computation of minimax almost-stable matchings in various settings. On the positive side, we provide polynomial-time algorithms in Section \ref{sec:short} that compute optimal solutions for instances with preference lists of length at most 2. Furthermore, in Section \ref{sec:approx}, we investigate the approximability of our optimisation problems of interest. In Section \ref{sec:ilp}, we provide integer programming formulations to solve our problems to optimality, and present experimental results. Finally, we summarise new and known complexity results in Table \ref{table:results} and pose open questions in Section \ref{sec:conclusion}.

\section{Background}
\label{sec:background}

\subsection{Formal Definitions}

We start by rigorously defining the models that we already introduced informally.

\begin{definition}[{\sc sri} Instance]
    Let $I=(A,\succ)$ be a {\sc Stable Roommates with Incomplete Lists} ({\sc sri}) \emph{instance}, where $A=\{ a_1, a_2, \dots, a_n \}$ is a set of $n\in \mathbb{N}$ agents, and $\succ$ is a tuple of strict ordinal \emph{preference rankings} $\succ_i$ of agents $a_i$ over an arbitrary subset of $A\setminus\{a_i\}$. $a_i$ \emph{prefers} $a_j$ to $a_k$ if $a_j\succ_ia_k$. We write $a_j\succeq_ia_k$ if $a_j\succ_i a_k$ or $j=k$. $a_i$ and $a_j$ are \emph{acceptable} if $a_i$ appears in $\succ_j$ and vice versa. Throughout, we assume that acceptability is symmetric.

    We can equivalently treat an {\sc sri} instance $I=(A,\succ)$ as a graph $G=(A,E)$: let every agent $a_i\in A$ be a vertex in the graph and let all acceptable pairs of agents $\{a_i,a_j\}$ be in $E$. We refer to $G$ as the \emph{acceptability graph} and say that $I$ has \emph{complete preferences} if $G$ is isomorphic to the complete graph $K_n$, and \emph{incomplete preferences} if not. We also assume that every agent ranks themselves last.
\end{definition}

\begin{definition}[{\sc smi} Instance]
    Let $I=(A,\succ)$ be an {\sc sri} instance and let $G$ be the acceptability graph of $I$. If $G$ is bipartite, i.e., there exist disjoint sets $A_1$ and $A_2$ such that $A=A_1\dot{\cup}A_2$ and no agents within $A_1$ or $A_2$ are mutually acceptable, then we refer to $I$ as a {\sc Stable Marriage with Incomplete Lists} ({\sc smi}) instance. Furthermore, we say that {\sc smi} instance $I$ has \emph{complete preferences} if $G$ is isomorphic to $K_{\vert A_1\vert,\vert A_2\vert}$, and \emph{incomplete preferences} if not.
\end{definition}

We adopt the following notation and assume the following concepts for matchings.

\begin{definition}[Matchings]
    Let $I=(A,\succ)$ be an {\sc sri} (or {\sc smi}) instance. A \emph{matching} $M$ of $I$ is a set of unordered pairs of agents such that no agent $a_i$ is contained in more than one pair. If $\{a_i,a_j\}\in M$, then we refer to $a_i$ as \emph{matched} and to $a_j$ as the \emph{partner} of $a_i$, denoted by $M(a_i)$, and vice versa. If $a_i$ is not paired with any other agent, we say that $a_i$ is \emph{unmatched} and set $M(a_i)=a_i$. 
    
    We denote the set of all matchings of $I$ by $\mathcal M$ and refer to some $M\in \mathcal M$ as a \emph{maximum-cardinality matching} if, for any $M'\in \mathcal M$, $\vert M\vert\geq \vert M'\vert$. We denote the set of all maximum-cardinality matchings by $\mathcal{M}^+$. We refer to $M\in \mathcal{M}^+$ as \emph{perfect} if every agent is matched, and denote the set of such matchings by $\mathcal{M}^p$.
\end{definition}

We adopt the following classical stability notion.

\begin{definition}[Stability]
    Let $I=(A,\succ)$ be an {\sc sri} (or {\sc smi}) instance and let $M$ be a matching of $I$. A \emph{blocking pair} admitted by $M$ is a pair of distinct agents $a_i, a_j\in A$ such that $a_i\succ_j M(a_j)$ and $a_j\succ_i M(a_i)$. We denote the \emph{set of blocking pairs} by $bp(M)$, and the subset thereof that involves an agent $a_i$ by $bp_{a_i}(M)$. Furthermore, the set of \emph{blocking agents} $ba(M)$ consists of all agents involved in blocking pairs.
    
    If $bp(M)=\varnothing$, then $M$ is referred to as \emph{stable}, and as \emph{unstable} otherwise. We refer to $I$ as \emph{solvable} if it admits at least one stable matching and as \emph{unsolvable} otherwise.
\end{definition}

In the absence of stable matchings in {\sc sr}, we previously argued that one might want to find an ``almost-stable'' matching that minimises the number of blocking pairs \cite{abraham06}. A definition of this problem follows below.

\begin{problemBox}
{\sc MinBP-AlmostStable-sri} \\
\textbf{Input:} {\sc sri} instance $I$. \\
\textbf{Output:} A matching $M$ of $I$ such that $\vert bp(M)\vert=\min_{M'\in \mathcal M}\vert bp(M')\vert$.
\end{problemBox}

As noted before, even when considering only {\sc smi} instances, the set of stable matchings might not intersect with $\mathcal {M}^+$. Then, we may want to find a maximum-cardinality matching that is ``as stable as possible'', i.e., minimises the number of blocking pairs \cite{biro_sm_10}. This problem is defined as follows.

\begin{problemBox}
{\sc MinBP-AlmostStable-Max-smi} \\
\textbf{Input:} {\sc smi} instance $I$. \\
\textbf{Output:} A maximum-cardinality matching $M$ of $I$ such that $\vert bp(M)\vert=\min_{M'\in \mathcal M^+}\vert bp(M')\vert$.
\end{problemBox}

For our complexity reductions, we will take advantage of the fact that the problem {\sc (2,2)-e3-sat}, which we formally define below, is known to be {\sf NP-complete} \cite{Berman2003Max3SAT}.

\begin{problemBox}
{\sc (2,2)-e3-sat} \\[4pt]
\textbf{Input:} Boolean formula $B$ in conjunctive normal form, where every variable occurs exactly twice unnegated and twice negated, and every clause contains exactly three literals. \\[2pt]
\textbf{Question:} Does there exist a satisfying truth assignment $f$ of $B$?
\end{problemBox}

To illustrate this problem for 3 variables $V_i$, $B=(V_1\lor V_2\lor V_3)\land(\bar{V_1}\lor \bar{V_2}\lor \bar{V_3})\land(V_1\lor \bar{V_2}\lor V_3)\land(\bar{V_1}\lor V_2\lor \bar{V_3})$ is a yes-instance to {\sc (2,2)-e3-sat} because it admits satisfying truth assignments $f:V\rightarrow \{T,F\}$. For example, the assignment $f(V_1)=T,f(V_2)=F$ and $f(V_3)=F$ satisfies $B$.

\subsection{Related Work}
\label{sec:related}

{\sc MinBP-AlmostStable-sri} was shown to be {\sf NP-hard} regardless of whether preference lists are of length at most 3 or complete \cite{abraham06,biro12}, {\sf NP-hard} to approximate within $n^{1/2-\varepsilon}$ (for any $\varepsilon>0$, where $n$ is the number of agents) \cite{abraham06}, and {\sf W[1]-hard} with respect to $\min_{M\in \mathcal M}\vert bp(M)\vert$ even when preference lists are of length at most 5 \cite{chen17}.

Biró, Manlove and Mittal \cite{biro_sm_10} conducted a thorough study of the computational complexity of the trade-off between matching size and stability in {\sc smi}, where stability is measured in terms of the number of blocking pairs or blocking agents. They showed that {\sc MinBP-AlmostStable-Max-smi} (and similarly for blocking agents) is {\sf NP-hard} to approximate within $n^{1-\varepsilon}$ (for any $\varepsilon>0$) \cite{biro_sm_10}. Hamada et al. \cite{hamada09} strengthened this result to the case where all preference lists are of length at most 3. On the positive side, the problem is known to be solvable in polynomial time when one side of the bipartition has preference lists of length at most 2 \cite{biro_sm_10}. Gupta et al. \cite{gupta2020parameterized} showed that {\sc MinBP-AlmostStable-Max-smi} remains {\sf W[1]-hard} for the combined parameter $(\beta,t,d)$, where $\beta$ is the number of blocking pairs, $t$ is the increase in size to a stable matching, and $d$ is the maximum length of a preference list.

In a recent preprint, Chen, Roy and Simola \cite{chen2025fptapproximabilitystablematchingproblems} proved that it is {\sf W[1]-hard} with respect to the optimal value to even approximate {\sc MinBP-AlmostStable-sri} and {\sc MinBP-AlmostStable-Max-smi} within any computable function depending only on this optimal value.

We will denote the decision problems ``Does there exist a matching $M$ of $I$ such that $\vert bp(M)\vert\leq k$?'' and ``Does there exist a perfect matching $M$ of bipartite instance $I$ such that $\vert bp(M)\vert\leq k$?'' by {\sc $k$-BP-AlmostStable-sri} and {\sc $k$-BP-AlmostStable-Perfect-smi}, respectively. Both problems are known to be {\sf NP-complete} \cite{abraham06,biro12,biro_sm_10,chen2019computational}, but in {\sf XP} with respect to $k$ \cite{abraham06,biro_sm_10}. It has been shown that minimising the number of blocking agents instead also remains computationally intractable \cite{biro_sm_10,chen17,chen2019computational}.

Different definitions of almost-stability have also been considered from a communication-complexity perspective, e.g., by Floréen et al. \cite{Floreen2010} and Ostrovsky and Rosenbaum \cite{OstrovskyRosenbaum2015}. Work on almost-stable matchings for other problems has also been conducted (e.g., see \cite{gupta2020parameterized,minbp_hrc_17,couples24,glitznermanlovepref}). Other alternative solution concepts for {\sc sri} instances that are different from almost-stable matchings have also been explored (see, for example, \cite{gupta_popsr_21, faenza_pop,tan91_2,herings25,vandomme2025locally,glitzner2024structuralteac,glitznerteac2}). An alternative perspective on minimising instability in {\sc smi} was offered by Biermann \cite{Biermann2011Measure}, who argued that, instead of counting blocking pairs, one should count a subset thereof that comprises a possible and economically reasonable transformation of the market. Similar perspectives have been studied for {\sc sri} by Atay, Mauleon and Vannetelbosch \cite{ATAY2021102465}, for example.

\section{Introducing Minimax Almost-Stability}
\label{sec:minimax}

\subsection{Motivation and Problem Definitions}

In many real-world matching markets, agents contained in multiple blocking pairs not only have an increased opportunity to deviate and potentially cause unravelling, but they may also feel unfairly treated, knowing that they are in a position in which they could be better off after a decentralised reorganisation. To address this, we initiate the study of minimax almost-stable matchings, which optimise the solution for the worst-off agent by minimising the maximum number of blocking pairs any agent is involved in. This approach provides strong individual-level fairness guarantees and is particularly relevant in high-stakes multi-agent settings where agents may feel a critical need to be well-off. The two natural optimisation problems that we newly introduce and study in this paper are formally defined below. The first one, denoted by {\sc Minimax-AlmostStable-sri}, aims for a matching of an {\sc sri} instance with a minimal maximum number of blocking pairs for any agent.

\begin{problemBox}
{\sc Minimax-AlmostStable-sri} \\
\textbf{Input:} {\sc sri} instance $I=(A,\succ)$. \\
\textbf{Output:} A matching $M\in\mathcal M$ such that $\max_{a_i\in A}\vert bp_{a_i}(M)\vert=\min_{M' \in \mathcal M}\max_{a_i\in A}\vert bp_{a_i}(M')\vert$.
\end{problemBox}

The second problem, denoted by {\sc Minimax-AlmostStable-Max-smi}, is motivated naturally by the setting where the central requirement is to maximise the size of the matching and, subject to this, we aim to minimise the maximum number of blocking pairs that any agent is contained in. We study this in the restricted setting of {\sc smi}, as we will show that even in this setting the problem turns out to be highly intractable. However, the problem could also be defined more generally in the {\sc sri} context, which we will denote by {\sc Minimax-AlmostStable-Max-sri}.

\begin{problemBox}
{\sc Minimax-AlmostStable-Max-smi} \\
\textbf{Input:} {\sc smi} instance $I=(A,\succ)$. \\
\textbf{Output:} A matching $M\in \mathcal{M}^+$ such that $\max_{a_i\in A}\vert bp_{a_i}(M)\vert =\min_{M'\in \mathcal{M}^+} \max_{a_i\in A}\vert bp_{a_i}(M')\vert$.
\end{problemBox}

Now, in Section \ref{sec:comp}, we will first look at the compatibility of different notions of almost-stability and show that minimax almost-stability is fundamentally different from the others. We will then, in Section \ref{sec:structure}, provide further structural insights into the structure of these problems.

\subsection{Differentiating Almost-Stability Notions}
\label{sec:comp}

Minimax almost-stability differs fundamentally from aggregate stability metrics (such as minimising the number of blocking pairs \cite{abraham06,biro_sm_10,chen17}) by prioritising the worst-off agent (in terms of the number of favourable deviation opportunities). Furthermore, it contrasts the minimisation of blocking agents \cite{biro_sm_10,chen17}, which likely requires fewer agents with a stronger incentive to deviate, therefore causing a greater risk of unravelling.

We give the following examples to highlight the differences between these three notions. These are for the {\sc sri} setting, but similar examples apply for {\sc Max-smi}.

\begin{prop}
\label{prop:minbpminba}
    A matching that admits a minimum number of blocking agents might not admit a minimum number of blocking pairs, and vice versa.
\end{prop}
\begin{proof}
    Consider the {\sc sri} instance with preferences as follows:
    \begin{align*}
        a_1 &: a_2 \; a_3 \; a_4 \\
        a_2 &: a_3 \; a_1 \\
        a_3 &: a_1 \; a_2 \\
        a_4 &: a_5 \; a_6 \; a_1 \\
        a_5 &: a_6 \; a_4\\
        a_6 &: a_4 \; a_5 
    \end{align*}
    Then $M=\{\{a_2,a_3\},\{a_5,a_6\}\}$ admits four blocking agents ($a_1,a_3,a_4,a_6$), which is minimal (as agents $a_1,a_2,a_3$ form a directed cycle in the first and second preference list positions, and $a_4,a_5,a_6$ do too), but three blocking pairs ($\{a_1,a_2\},\{a_1,a_3\},\{a_4,a_6\}$). The matching $M'=\{\{a_1,a_4\},\{a_2,a_3\},\{a_5,a_6\}\}$, on the other hand, admits the same blocking agents, but only admits two blocking pairs (this is easy to see as $M\subset M'$, so $bp(M')\subseteq bp(M)$, and the blocking pair $\{a_1,a_4\}\in bp(M)$ is a match in $M'$).

    Now consider the {\sc sri} instance with preferences as follows:
    
    \begin{align*}
        a_1 &: a_2 \; a_3 \; a_4 \; a_5 \; a_6 \\
        a_2 &: a_3 \; a_1 \; a_4 \; a_5 \; a_6\\
        a_3 &: a_1 \; a_2 \; a_4 \; a_5 \; a_6\\
        a_4 &: a_5 \; a_1 \; a_4 \; a_5 \; a_6\\
        a_5 &: a_1 \; a_2 \; a_3 \; a_4 \; a_6 \\
        a_6 &: a_1 \; a_2 \; a_3 \; a_4 \; a_5 
    \end{align*}
    Then $M=\{\{a_1,a_4\},\{a_2,a_3\},\{a_5,a_6\}\}$ admits two blocking pairs ($\{a_1,a_3\},\{a_4,a_5\}$), which is minimal (this can be verified computationally), but four blocking agents ($a_1,a_3,a_4,a_5$). On the other hand, the matching $M'=\{\{a_1,a_6\},\{a_2,a_3\},\{a_4,a_5\}\}$ also admits two blocking pairs ($\{a_1,a_3\},\{a_1,a_5\}$), but only admits three blocking agents ($a_1,a_3,a_5$).
\end{proof}

\begin{prop}
    A matching that admits a minimum number of blocking agents might not admit a minimax number of blocking pairs per agent, and vice versa.
\end{prop}
\begin{proof}
    This follows immediately from the examples in the proof of Proposition \ref{prop:minbpminba}: the first example shows that matching $M$ with a minimum number of blocking agents may have an agent in two blocking pairs, while there exists a matching $M'$ in which no agent is in more than one blocking pair. Similarly, in the second example, $M$ is a matching where no agent is in more than one blocking pair, but $M'$ has fewer blocking agents.
\end{proof}

\begin{prop}
    A matching that admits a minimax number of blocking pairs per agent might not admit a minimum number of blocking pairs.
\end{prop}
\begin{proof}
    Consider the {\sc sri} instance with preferences as follows:
    \begin{align*}
        a_1 &: a_2 \; a_3 \\
        a_2 &: a_3 \; a_1 \\
        a_3 &: a_1 \; a_2 \\
        a_4 &: a_5 \\
        a_5 &: a_4 
    \end{align*}
    Clearly the collection of matchings $M_r=\{\{a_r,a_{r+1}\}, \{a_4,a_5\}\}$ where $1\leq r\leq 3$ and addition taken modulo 3 are the matchings with a minimum number of blocking pairs, but notice that any $M_r$, but also any $M_r'=\{\{a_r,a_{r+1}\}\}$, satisfies the minimax constraint, while no $M_r'$ matching satisfies minimality with respect to the number of blocking pairs. 
\end{proof}

We are not aware of an instance where a minimum number of blocking pairs does not have a minimax number of blocking pairs per agent.

\subsection{Further Structural Results}
\label{sec:structure}

While it is not challenging to construct a family of {\sc sri} instances where the minimum number of blocking pairs grows linearly with the number of agents, it is not straightforward to construct instances where the minimum maximum number of blocking pairs that any agent is contained in is higher than 1. Hence, one might falsely conclude that this number is constant or even just 1. However, we will now give a construction that shows that, with sufficiently many agents, this number can be arbitrarily large.

\begin{prop}
\label{prop:lbsri}
    For every positive integer $k$, there exists an {\sc sri} instance $I=(A,\succ)$ with $n=3^k$ agents and complete preference lists such that $\min_{M\in\mathcal M}\max_{a_i\in A}\vert bp_i(M)\vert= k$.
\end{prop}
\begin{proof}

    Suppose that $k\geq 1$. We will first describe the construction and then prove the desired invariant inductively. The main idea is to build a symmetrical preference system consisting of partitions of sizes that are multiples of 3 (up until a single partition consisting of all agents of size $3^k$), such that we have many nested \emph{preference cycles}, where a preference cycle $(f_{c_1},f_{c_2},\dots,f_{c_r})$ of length $r$ has the property that each agent $f_{c_k}$ prefers their successor $f_{c_{k+1}}$ to their predecessor $f_{c_{k-1}}$ (where addition and subtraction is taken modulo $r$) in the cycle. Notice that odd-length cycles where every agent ranks their successor and their predecessor in the first and second position of their preference list, respectively, make an instance unsolvable. Here, we draw inspiration from the concept of stable partitions, which characterise the existence of stable matchings in {\sc sri} instances (see \cite{matchup} for more details).

    We construct the instance, which we will denote by $I_k$, as follows: first, create $3^k$ agents $A=\{a_1,a_2,\dots,a_{3^k}\}$. Now partition the agents $a_i$ into sets of 3 in increasing order of $i$ (starting from 1) and initialise the preference lists such that $(a_{3r+1}, a_{3r+2}, a_{3r+3})$ form a preference cycle at the top of their preference list for all $0\leq r\leq \frac{n}{3}-1$. Next, partition the agents $a_i$ into sets of 9 in increasing order of $i$ (starting from 1) and append the preference lists such that 
    $$(\{a_{9r+1}, a_{9r+2}, a_{9r+3}\},
    \{a_{9r+4}, a_{9r+5}, a_{9r+6}\},
    \{a_{9r+7}, a_{9r+8}, a_{9r+9}\})$$ 
    form a preference cycle for all $0\leq r\leq \frac{n}{9}-1$, i.e., such that all agents $a_{9r+1}, a_{9r+2}, a_{9r+3}$ prefer all of $a_{9r+4}, a_{9r+5}, a_{9r+6}$ to all of $a_{9r+7}, a_{9r+8}, a_{9r+9}$, and so on (we may assume that agents within these sets are listed in increasing indicial order in the preference lists). Continue this process for increasing partition sizes $3^j$ up until $3^{k}$, in which case we append the preference lists such that 
    $$(\{a_{1}, a_{2}, \dots, a_{\frac{n}{3}}\},
    \{a_{\frac{n}{3}+1}, a_{\frac{n}{3}+2}, \dots, a_{\frac{2}{3}n}\},
    \{a_{\frac{2}{3}n+1}, a_{\frac{2}{3}n+2}, \dots, a_{n}\})$$
    form a preference cycle (notice that at each step $j$, each agent adds $2\cdot3^{j-1}$ other agents to the end of their preference list). Figure \ref{fig:lbsri} illustrates the construction of $I_3$ schematically for (i.e., $k=3$, so we have 27 agents).

    \begin{figure}[!tbh]
        \centering
        \includegraphics[width=12cm]{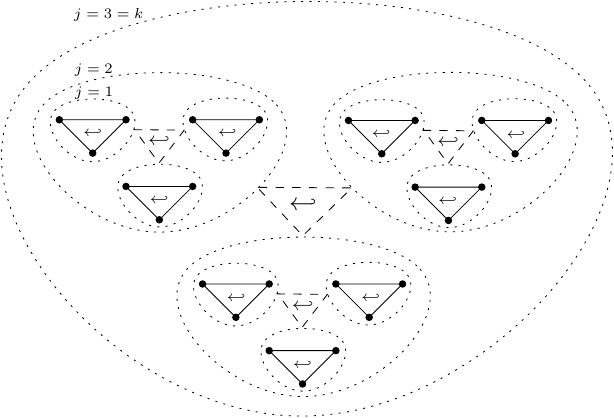}
        \caption{Schematic illustration of the {\sc sri} lower-bound construction for 27 agents (i.e,. $k=3$). Agents are partitioned into 3 nested sets containing $3^j$ (for $1\leq i\leq k$) agents each. Arrows indicate preference cycles within and between the nested levels.}
        \label{fig:lbsri}
    \end{figure}

    Notice that, at the end of this construction, all preference lists are complete. Furthermore, notice that, at every step $j$ with partitions of size $3^j$, we create nested odd-length cycles such that, from every set in the partition, we are forced to match at least one agent to someone of rank at least $3^j$ on their preference list (or leave the agent unmatched). However, due to the preference cycles, this incurs at least $j$ blocking pairs for this agent in every matching of $I$. 
    
    Formally, let $P(k)$ be the statement that, for the instance $I_k$ constructed as described above with $3^k$ agents, it is the case that $\min_{M\in\mathcal {M}(I_k)}\max_{a_i\in A}\vert bp_i(M)\vert\geq k$ (where $\mathcal{M}(I_{k})$ indicates the sets of all matchings of $I_{k}$). 
    
    Consider the base case where $k=1$: clearly, with three agents that form a preference cycle, any one of the three maximum-cardinality matchings of the instance leaves one agent unmatched and minimises the maximum number of blocking pairs that any agent is contained in. However, the unmatched agent must prefer both matched agents to being unmatched, and exactly one of the matched agents must prefer the unmatched agent to their partner. Thus, $\min_{M\in\mathcal {M}(I_k)}\max_{a_i\in A}\vert bp_i(M)\vert\geq 1=k$. 
    
    Now suppose that $P(k)$ holds for some integer $k\geq 1$ and consider $P(k+1)$. Let $I_k$ and $I_{k+1}$ be the respective instance constructions. By symmetry of the construction, $I_{k+1}$ consists of three components that mirror $I_k$, and then edges between these three components at the end of the preference lists. Notice that, because $I_k$ contains $3^k$ agents, any matching of $I_k$ must leave at least one agent unmatched. By the inductive hypothesis, we know that any matching of $I_k$ leaves at least one agent in at least $k$ blocking pairs. It is easy to see that this must apply to the unmatched agent, as being unmatched incurs at least one blocking pair from each of the $k$ nested cycles. Thus, for any union of three matchings of the distinct $I_k$ components occurring in $I_{k+1}$, at least three agents are left unmatched, each of whom is in at least $k$ blocking pairs (note that any blocking pair within an $I_k$ component is also blocking in the larger $I_{k+1}$ instance, as all agents within the smaller component prefer each other to all other agents outside their component, and the relative preference order for agents within the component is maintained). Furthermore, by construction, these three agents must form another preference cycle, introducing at least one more blocking pair among these agents. Therefore, at least one of the agents must remain unmatched in any matching of $I_{k+1}$, and they are contained in at least $k+1$ blocking pairs. Thus, it must hold that $\min_{M\in\mathcal {M}(I_{k+1})}\max_{a_i\in A}\vert bp_i(M)\vert\geq \min_{M\in\mathcal {M}(I_k)}\max_{a_i\in A}\vert bp_i(M)\vert +1 \geq k+1$ as required. Thus, $P(k)$ holds for any integer $k\geq 1$ by induction.

    Now, to see that this lower bound is tight, i.e., that $\min_{M\in\mathcal {M}(I_{k})}\max_{a_i\in A}\vert bp_i(M)\vert=k$, consider the following matching: $M_k=\{\{a_1,a_2\},\{a_3,a_4\},\dots,\{a_{3^k-2},a_{3^k-1}\}\}$. Let $Q(k)$ be the statement that, for the instance construction $I_k$ above and this associated matching $M_k$, it is the case that $\max_{a_i\in A}\vert bp_i(M_k)\vert\leq k$. 
    
    Consider the base case where $k=1$: we have a simple instance with three agents that rank each other in a cyclic fashion. Here, $M_k$ clearly admits the blocking pair $\{a_2,a_3\}$ and no others, so $\max_{a_i\in A}\vert bp_i(M_k)\vert=1=k$. Thus, $Q(1)$ holds. 
    
    Now, for the inductive step, suppose that $Q(k)$ holds for some $k\geq 1$ and consider $Q(k+1)$. We already know that $I_{k+1}$ consists of three symmetrical copies of $I_k$ and circular connecting preferences between these components appended at the end of the lists. Further, we can see by inspection that $M_{k+1}$ matches $3^k$ agents to other agents in their respective $I_k$ component in the same way that $M_k$ does (with the indexing possibly shifted but symmetrical), with exactly two agents from different $I_k$ components matched to each other and one agent from the third $I_k$ component remaining unmatched. Clearly, for the $3^k$ ``ordinarily matched'' agents who are matched within their own component, the inductive hypothesis applies, and all of these agents are in at most $k$ blocking pairs. Furthermore, by the same inductive hypothesis, the three ``unordinarily matched'' agents are blocking with at most $k$ agents from their own component. None of these three agents are blocking with ``ordinarily matched'' agents from components other than their own (as these all prefer their partner to the agents in question by construction of the preference lists). Thus, by matching two of these ``unordinarily matched`` agents to each other, and given that their preferences over each other are necessarily cyclic, these three agents can form at most one further blocking pair among each other. Thus, none of these agents are in more than $k+1$ blocking pairs, proving $Q(k+1)$. Correctness of $Q$ for all positive integers $k$ follows. Thus, putting the lower and upper bounds that we proved inductively together, we have that $\min_{M\in\mathcal M}\max_{a_i\in A}\vert bp_i(M)\vert=k$ as required. 
\end{proof}

With this construction, we can immediately observe the following worst case lower bound for solutions to {\sc Minimax-AlmostStable-sri}.

\begin{theorem}
    In the worst case, \emph{OPT}$ =\Omega(\log(n))$, where \emph{OPT} is the optimal solution value of {\sc Minimax-AlmostStable-sri} and $n$ is the number of agents, even when all preference lists are complete.
\end{theorem}
\begin{proof}
    This follows from the infinite family of instances constructed in the proof of Proposition \ref{prop:lbsri}. Notice that there, $n\leq3^{\text{OPT}}$, so, by logarithm rules, OPT $=\log_3(n)=\frac{\log_2(n)}{\log_2(3)}=\Omega(\log(n))$ as required.
\end{proof}

We will now present a simpler construction which demonstrates that, when prioritising size over stability, the minimum maximum number of blocking pairs that any agent is contained in increases arbitrarily, even in the {\sc smi} setting.

\begin{prop}
\label{prop:lbmaxsmi}
    For every non-negative integer $k$, there exists an {\sc smi} instance $I=(A,\succ)$ with $n=2(k+1)$ agents such that $\min_{M\in\mathcal {M}^+}\max_{a_i\in A}\vert bp_i(M)\vert=k$.
\end{prop}
\begin{proof}
    We give a construction with the following $2(k+1)$ agents $A=\mathcal A\cup\mathcal B=\{a_1,a_2,\dots, a_{k+1}\}\cup \{b_1,b_2,\dots, b_{k+1}\}$. For all $j$ ($1\leq j\leq k$), the agent $a_j$ ranks $b_{k+1}$ in first position and $b_j$ in second position, while $b_j$ only ranks $a_j$ (so, necessarily, in first position). Finally, $b_{k+1}$ ranks all $a_j$ (for $1\leq j\leq k+1$) agents in position $j$. Agent $a_{k+1}$ only ranks $b_{k+1}$. We illustrate the construction in Figure \ref{fig:lbmaxsmi}.
    
    \begin{figure}[!tbh]
        \centering
        \includegraphics[width=8cm]{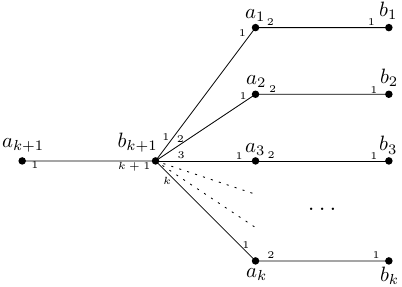}
        \caption{Illustration of the lower-bound construction for maximum-cardinality matchings in {\sc smi}.}
        \label{fig:lbmaxsmi}
    \end{figure}

    Clearly, the only maximum-cardinality matching is the perfect matching $M=\{\{a_j,b_j\} \;\vert\; 1\leq j\leq k+1\}$. However, it is easy to observe that $bp(M)=\{\{a_j,b_{k+1}\} \;\vert\; 1\leq j\leq k\}$. Therefore $bp_{b_{k+1}}(M)=\{a_j\;\vert\; 1\leq j\leq k \}$ and so $\min_{M'\in\mathcal {M}^+}\max_{a_i\in A}\vert bp_i(M')\vert=\vert bp_{b_{k+1}}(M)\vert=k$ as required.
\end{proof}

Using this construction, we can observe an even stronger lower bound and hence observe the following tight asymptotic bound for solutions to {\sc Minimax-AlmostStable-Max-smi} (and hence also for {\sc Minimax-AlmostStable-Max-sri}).

\begin{theorem}
    In the worst case, \emph{OPT} $=\Theta(n)$, where \emph{OPT} is the optimal solution value of {\sc Minimax-AlmostStable-Max-smi} and $n$ is the number of agents.
\end{theorem}
\begin{proof}
    This follows from the infinite family of instances constructed in the proof of Proposition \ref{prop:lbmaxsmi}. Notice that there, $n=2(\text{OPT}+1)$, so, OPT $=\frac{n}{2}-1=\Omega(n)$. Clearly, OPT $\leq n-1=O(n)$ as no agent can block with more agents than they have on their preference list. Thus, \emph{OPT}$=\Theta(n)$ as required.
\end{proof}

\section{Computational Intractability Results}
\label{sec:intractability}

\subsection{Decision Problems and Preliminary Results}

As noted in the previous sections, we now proceed to proving surprisingly strong intractability results for the computation of minimax almost-stable matchings. For our investigation, we first define decision problems associated with our central optimisation problems of interest. The following decision problem corresponds to {\sc Minimax-AlmostStable-sri}.

\begin{problemBox}
{\sc $k$-Max-AlmostStable-sri} \\
\textbf{Input:} {\sc sri} instance $I=(A,\succ)$ and a non-negative integer $k$. \\
\textbf{Question:} Does there exist a matching $M\in \mathcal M$ such that $\max_{a_i\in A}\vert bp_{a_i}(M)\vert\leq k$?
\end{problemBox}

It is easy to see that if {\sc $k$-Max-AlmostStable-sri} cannot be solved in polynomial time, then {\sc Minimax-AlmostStable-sri} cannot be solved in polynomial time. Notice that $I$ is a yes-instance to {\sc 0-Max-AlmostStable-sri} if and only if $I$ is solvable, so this case is tractable \cite{irving_sr}. Notice, furthermore, that {\sc $k$-Max-AlmostStable-sri} is trivial if $k$ exceeds the maximum preference list length. In Section \ref{sec:srihard}, we will highlight, though, that {\sc $k$-Max-AlmostStable-sri} is {\sf NP-complete} even in the very restricted case where $k=1$, and regardless of whether preference lists are bounded or complete. We will extend this result to any positive constant $k$ (but with weaker bounds on the preference list lengths) in Section \ref{sec:furtherhard}.

We will then turn our attention to {\sc Minimax-AlmostStable-Max-smi}. For our complexity analysis, we will consider a restricted decision version of this problem (where maximum-cardinality matchings are perfect), which we define as follows. 

\begin{problemBox}
{\sc $k$-Max-AlmostStable-Perfect-smi} \\
\textbf{Input:} {\sc smi} instance $I=(A,\succ)$ and a non-negative integer $k$. \\
\textbf{Question:} Does there exist a matching $M\in \mathcal{M}^p$ such that $\max_{a_i\in A}\vert bp_{a_i}(M)\vert\leq k$?
\end{problemBox}

It is easy to see that if {\sc $k$-Max-AlmostStable-Perfect-smi} cannot be solved in polynomial time, then {\sc Minimax-AlmostStable-Max-smi} cannot be solved in polynomial time. Notice that $I$ is a yes-instance to {\sc 0-Max-AlmostStable-Perfect-smi} if and only if there exists a perfect stable matching of $I$. This can be decided in polynomial time by computing an arbitrary stable matching $M$ of $I$ using the classical Gale-Shapley algorithm \cite{gale_shapley}, and checking whether $M$ is a perfect matching. It is well-known that all stable matchings of a given {\sc smi} instance have the same size \cite{galesoto85}, so if $M$ is not perfect, then $I$ does not admit any perfect stable matching. In Section \ref{sec:smihard}, we will highlight, though, that {\sc $k$-Max-AlmostStable-Perfect-smi} is {\sf NP-complete} even in the very restricted case where $k=1$ and preference lists are very short. We will extend this result to any positive constant $k$ and even to a linear dependence (within some arbitrarily small $\varepsilon$) on the number of agents in Section \ref{sec:furtherhard}.

\subsection{Hardness in General Bounded-Degree Graphs and Cliques}
\label{sec:srihard}

The result below states that it is impossible to distinguish in polynomial time even between instances that admit a matching in which every agent is in at most one blocking pair, and instances where every matching requires at least one agent to be in more than one blocking pair (unless {\sf P}={\sf NP}), even when preference lists are of bounded length.

\begin{theorem}
\label{thm:1-Maxbounded}
    {\sc $k$-Max-AlmostStable-sri} is {\sf NP-complete}, even if $k=1$ and all preference lists are of length at most 10.
\end{theorem}
\begin{proof}
    We first note that, given a matching $M$ of an {\sc sri} instance, we can verify in polynomial time whether $M$ is a solution to {\sc $k$-Max-AlmostStable-sri} by iterating through the preference lists of the agents and keeping track of the maximum number of blocking pairs that any agent is contained in. Clearly, blocking pairs can be detected efficiently by comparison of preference list entries to partners in $M$. Hence, {\sc $k$-Max-AlmostStable-sri} is in {\sf NP}.
    
    To show that {\sc $k$-Max-AlmostStable-sri} is {\sf NP-hard}, we will reduce from the previously defined {\sf NP-complete} problem {\sc (2,2)-e3-sat} to {\sc $1$-Max-AlmostStable-sri}. Suppose that $B$ is a boolean formula given as an instance of {\sc (2,2)-e3-sat}, where $V=\{V_1,V_2,\dots,V_n\}$ is the set of variables and $C=\{C_1,C_2,\dots,C_m\}$ is the set of clauses. We will denote a clause by $C_s=(C_s^1\lor C_s^2\lor C_s^3)$, where each $C_s^l$ (for $1\leq s\leq m$ and $1\leq l\leq 3$) is either an unnegated or a negated occurrence of some variable $V_i$. Now we construct an instance $I$ of {\sc $1$-Max-AlmostStable-sri} as follows. Let $A=A_F\cup A_V \cup A_C$ be the set of $20n+3m$ agents, where 
    \begin{align*}
        A_F&=\{f_i^{\alpha,z} \;;\; 1\leq i\leq n \land 1\leq \alpha\leq 8\land 1\leq z\leq 2\},\\
        A_V&=\{ v_i^T, v_i^F,v_i^1,v_i^2 \;;\; 1\leq i\leq n\}, \text{ and}\\
        A_C&=\{ x_j^s \;;\; 1\leq j\leq m \land 1\leq s\leq 3\}.
    \end{align*}
    These three disjoint sets of agents correspond to those from forcing gadgets, variable gadgets, and clause gadgets, respectively, all of which we will introduce now.

    Specifically, for every variable $V_i\in V$, we create a \textbf{variable gadget} $G_{V_i}$ consisting of the four vertices $v_i^T, v_i^F, v_i^1, v_i^2$ (and more inherited from the use of forcing gadgets). We construct the preferences of the variable gadget as follows:
    \begin{align*}
        v_i^T &: v_i^1 \; x_1(v_i^T) \; x_2(v_i^T) \; v_i^2 \\
        v_i^F &: v_i^1 \; x_1(v_i^F) \; x_2(v_i^F) \; v_i^2 \\
        v_i^1 &: v_i^T \; v_i^F \; G_{F(v_i^1)} \\
        v_i^2 &: v_i^T \; v_i^F \; G_{F(v_i^2)} \\
        (\text{+ agents inherited } &\text{from use of $G_{F(v_i^z)}$ gadgets})
    \end{align*}
    where $x_1(v_i^T)$ and $x_2(v_i^T)$ will be specified later. The construction is illustrated in Figure \ref{fig:gadgetV}.

    \begin{figure}[!tbh]
        \centering
        \includegraphics[width=8cm]{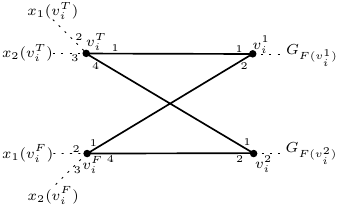}
        \caption{Illustration of the variable gadget construction}
        \label{fig:gadgetV}
    \end{figure}
    
    The \textbf{forcing gadget}, denoted by $G_{F(v_i^z)}$ for a given agent $v_i^z$ (where $z\in\{1,2\}$), attaches to each agent $v_i^z$ and contributes the eight agents $f_{i}^{1,z}, f_i^{2,z},\dots, f_i^{8,z}$. Let $v_i^z$'s preference list before attachment of the gadget be denoted by $P_{v_i^z}$ (in the variable gadget above, $P_{v_i^z}=(v_i^T \; v_i^F)$. Then, after attachment of the gadget, let $v_i^z$'s preference list consist of $P_{v_i^z}$ followed by all other agents $f_i^{\alpha,z}$ within the gadget in increasing order of $\alpha$:   
    $$ v_i^z : P_{v_i^z} \; f_i^{1,z} \; f_i^{2,z} \; f_i^{3,z} \; f_i^{4,z} \; f_i^{5,z} \; f_i^{6,z} \; f_i^{7,z} \; f_i^{8,z} $$
    The main idea of this gadget is to force agent $v_i^z$ to be matched to someone on $P_{v_i^z}$ in any matching where no agent is in more than one blocking pair. While $v_i^z$ has this special function and preference list, the remaining agents have symmetric preference lists. The construction is strongly inspired by the construction in Proposition \ref{prop:lbsri}. Specifically, the preference lists of the remaining agents in $G_{F(v_i^z)}$ are constructed such that they form nested preference cycles. We construct the remainder of the gadget such that $(v_i^z, f_i^{1,z}, f_i^{2,z})$ is a preference cycle at the top of the preference lists of the new $f_i^{\alpha,z}$ agents, and similarly for $(f_i^{3,z}, f_i^{4,z}, f_i^{5,z})$ and $(f_i^{6,z}, f_i^{7,z}, f_i^{8,z})$. Subsequently, the remainder of the preference lists are constructed such that we get the preference cycles $(\{v_i^z,f_i^{1,z},f_i^{2,z}\}, \{f_i^{3,z},f_i^{4,z},f_i^{5,z}\}, \{f_i^{6,z},f_i^{7,z},f_i^{8,z}\})$, i.e., all agents $v_i^z,f_i^{1,z}$ and $f_i^{2,z}$ prefer all of $f_i^{3,z},f_i^{4,z}$ and $f_i^{5,z}$ to all of $f_i^{6,z},f_i^{7,z}$ and $f_i^{8,z}$, and so on. The lists look as follows:

    $$ f_i^{1,z} : f_i^{2,z} \; v_i^z \; f_i^{3,z} \; f_i^{4,z} \; f_i^{5,z} \; f_i^{6,z} \; f_i^{7,z} \; f_i^{8,z} $$
    $$ f_i^{2,z} : v_i^z \; f_i^{1,z} \; f_i^{3,z} \; f_i^{4,z} \; f_i^{5,z} \; f_i^{6,z} \; f_i^{7,z} \; f_i^{8,z} $$
    and
    $$ f_i^{3,z} : f_i^{4,z} \; f_i^{5,z} \; f_i^{6,z} \; f_i^{7,z} \; f_i^{8,z} \; v_i^z \; f_i^{1,z} \; f_i^{2,z}$$
    $$ f_i^{4,z}: f_i^{5,z} \; f_i^{3,z} \; f_i^{6,z} \; f_i^{7,z} \; f_i^{8,z} \; v_i^z \; f_i^{1,z} \; f_i^{2,z}$$
    $$ f_i^{5,z}: f_i^{3,z} \; f_i^{4,z} \; f_i^{6,z} \; f_i^{7,z} \; f_i^{8,z} \; v_i^z \; f_i^{1,z} \; f_i^{2,z}$$
    and
    $$ f_i^{6,z} : f_i^{7,z} \; f_i^{8,z} \; v_i^z \; f_i^{1,z} \; f_i^{2,z} \; f_i^{3,z} \; f_i^{4,z} \; f_i^{5,z} $$
    $$ f_i^{7,z} : f_i^{8,z} \; f_i^{6,z} \; v_i^z \; f_i^{1,z} \; f_i^{2,z} \; f_i^{3,z} \; f_i^{4,z} \; f_i^{5,z}$$
    $$ f_i^{8,z} : f_i^{6,z} \; f_i^{7,z} \; v_i^z \; f_i^{1,z} \; f_i^{2,z} \; f_i^{3,z} \; f_i^{4,z} \; f_i^{5,z}$$
    A visual example of the forcing gadget (without edge labels) is shown in Figure \ref{fig:gadgetF}. This concludes the construction of the variable and forcing gadgets.

    \begin{figure}[!tbh]
        \centering
        \includegraphics[width=6cm]{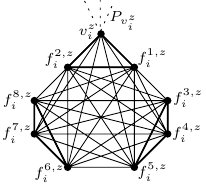}
        \caption{Illustration of the forcing gadget construction (edge labels omitted for clarity)}
        \label{fig:gadgetF}
    \end{figure}

    Finally, for each clause $C_j=(C_j^1\lor C_j^2\lor C_j^3)\in C$, we construct a \textbf{clause gadget} $G_{C_j}$ consisting of the three agents $x_j^1, x_j^2, x_j^3$, respectively. These agents are connected to variable gadgets, but we do not consider these occurrences of variable gadgets to be part of the clause gadget (in contrast to the forcing gadgets being part of the variable gadgets). The preferences are constructed as follows:
    \begin{align*}
        x_j^1 &: x_j^2 \; x_j^3 \; v(x_j^1) \\
        x_j^2 &: x_j^3 \; x_j^1 \; v(x_j^2) \\
        x_j^3 &: x_j^1 \; x_j^2 \; v(x_j^3)
    \end{align*}
    where $v(x_j^s)$ will be specified below. This concludes the construction of the clause gadgets. The clause gadget construction is illustrated in Figure \ref{fig:gadgetC}.
    
    \begin{figure}[!tbh]
        \centering
        \includegraphics[width=5cm]{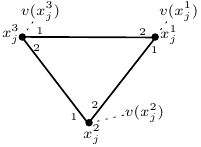}    
        \caption{Illustration of the clause gadget construction}
        \label{fig:gadgetC}
    \end{figure}

    Now we link up the variable and clause gadgets as follows: for each unnegated literal $C_j^d=V_i$, let $v(x_j^d)=v_i^T$ and let $x_1(v_i^T)=x_j^d$ if this is the first occurrence of this literal or $x_2(v_i^T)=x_j^d$ if this is the second occurrence of this literal. Similarly, for each negated literal $C_j^d=\bar{V_i}$, let $v(x_j^d)=v_i^F$ and let $x_1(v_i^F)=x_j^d$ if this is the first occurrence of this literal or $x_2(v_i^F)=x_j^d$ if this is the second occurrence of this literal. Recall that by definition of {\sc (2,2)-e3-sat}, there are exactly two unnegated and two negated occurrences of each literal in $B$. Notice that the choice of first or second occurrence in the preference list construction within the variable gadgets will be irrelevant, and that this is an arbitrary choice we make. This completes the construction of $I$. Notice that the construction of each gadget can be performed in constant time due to the bounded size of each gadget, so $I$ can be constructed from $B$ in polynomial time. 

    We will now show that the gadget $G_{F(v_i^z)}$ does indeed force $v_i^z$ to be matched to someone on $P_{v_i^z}$ for any matching in which no agent is in more than one blocking pair.

    \begin{claim}
        \label{claim:mmsri1}
        Let $G_{F(v_i^z)}$ be an occurrence of the forcing gadget, and let $M$ be a matching of $I$. If $M(v_i^z)\notin P_{v_i^z}$, then $\max_{a_r\in A[G_{F(v_i^z)}]}\vert bp_{a_r}(M)\cap E[G_{F(v_i^z)}]\vert> 1$.\footnote{We use the notation $A[G_{F(v_i^z)}]$ to refer to the agents contained in $G_{F(v_i^z)}$ and the notation $E[G_{F(v_i^z)}]$ to refer to the edges of the acceptability graph induced by $G_{F(v_i^z)}$.} 
    \end{claim}    
    \begin{proof}[Proof of Claim \ref{claim:mmsri1}]
    \renewcommand{\qedsymbol}{$\blacksquare$}
        Suppose that $M(v_i^z)\notin P_{v_i^z}$. Then at least one agent in each of $\{v_i^z,f_i^{1,z},f_i^{2,z}\}$, $\{f_i^{3,z},f_i^{4,z},f_i^{5,z}\}$ and $\{f_i^{6,z},f_i^{7,z},f_i^{8,z}\}$ must be matched to someone worse than both other members of their preference cycle. Without loss of generality (by symmetry in the preference construction), let these be $f_i^{2,z},f_i^{5,z}$ and $f_i^{8,z}$. Notice that immediately we must have $\{\{f_i^{1,z},f_i^{2,z}\}$, $\{f_i^{4,z},f_i^{5,z}\}$, $\{f_i^{7,z},f_i^{8,z}\}\}\subseteq bp(M)$. Notice, furthermore, that we can match at most two of the agents $f_i^{2,z},f_i^{5,z}$ and $f_i^{8,z}$ to each other. However, again, these three agents form a preference cycle of odd length. Without loss of generality, suppose that $f_i^{8,z}$ remains unmatched, then $\{f_i^{5,z},f_i^{8,z}\}\in bp(M)$. Hence, $\vert bp_{f_i^{5,z}}(M)\vert=\vert bp_{f_i^{8,z}}(M)\vert= 2>1$ and both of these blocking pairs are within $G_{F(v_i^z)}$, i.e., it must be the case that $\max_{a_r\in A[G_{F(v_i^z)}]}\vert bp_{a_r}(M)\cap E[G_{F(v_i^z)}]\vert> 1$.
    \end{proof}

    Notice that, by the contrapositive of Claim \ref{claim:mmsri1}, if $\max_{a_r\in A[G_{F(v_i^z)}]}\vert bp_{a_r}(M)\cap E[G_{F(v_i^z)}]\vert\leq 1$, then $M(v_i^z)\in P_{v_i^z}$. Conversely, if $\{\{f_i^{1,z},f_i^{2,z}\},\{f_i^{3,z},f_i^{4,z}\},\{f_i^{5,z},f_i^{6,z}\},\{f_i^{7,z},f_i^{8,z}\}\}\subseteq M$ and $M(v_i^z)\in P_{v_i^z}$, then, within the gadget $G_{F(v_i^z)}$, $M$ admits precisely the blocking pairs $\{f_i^{4,z},f_i^{5,z}\}$ and $\{f_i^{6,z},f_i^{8,z}\}$ (and possibly blocking pairs involving $v_i^z$ and agents outside of $G_{F(v_i^z)}$), so no agent within $G_{F(v_i^z)}$ is in more than one blocking pair with other agents within $G_{F(v_i^z)}$.

    Next, we will establish several properties about the full variable gadgets construction.

    \begin{claim}
    \label{claim:mmsri2}
        Let $G_{V_i}$ be an occurrence of the variable gadget, and let $M$ be a matching of $I$. Then
        \begin{itemize}
            \item if $\max_{a_r\in A[G_{V_i}]}\vert bp_{a_r}(M)\cap E[G_{V_i}]\vert\leq 1$ then $\{\{v_i^T,v_i^1\},\{v_i^F,v_i^2\}\}\subset M$ or $\{\{v_i^T,v_i^2\},\{v_i^F,v_i^1\}\}\subset M$;
            \item if $\{\{v_i^T,v_i^1\},\{v_i^F,v_i^2\}\}\not\subset M\land\{\{v_i^T,v_i^2\},\{v_i^F,v_i^1\}\}\not\subset M$, then $\max_{a_r\in A[G_{V_i}]}\vert bp_{a_r}(M)\cap E[G_{V_i}]\vert> 1$;
            \item if $\{\{v_i^T,v_i^1\},\{v_i^F,v_i^2\}\}\subseteq M$, then $bp(M)\vert_{G_{V_i}}=bp(M)\vert_{G_{F(v_i^1)}}\cup bp(M)\vert_{G_{F(v_i^2)}}$;\footnote{We use the notation $bp(M)\vert_{G_{F(v_i^s)}}$ to indicate that we restrict the set of blocking pairs admitted by $M$ to those that involve only pairs both contained in the part $G_{F(v_i^s)}$ of $I$.}
            \item if $\{\{v_i^T,v_i^2\},\{v_i^F,v_i^1\}\}\subseteq M$, then $bp(M)\vert_{G_{V_i}}=\{\{v_i^T,v_i^1\}\}\cup bp(M)\vert_{G_{F(v_i^1)}}\cup bp(M)\vert_{G_{F(v_i^2)}}$.
        \end{itemize}  
    \end{claim} 
    \begin{proof}[Proof of Claim \ref{claim:mmsri2}]
    \renewcommand{\qedsymbol}{$\blacksquare$}
        For the first part of the statement, suppose that $\max_{a_r\in A[G_{V_i}]}\vert bp_{a_r}(M)\cap E[G_{V_i}]\vert\leq 1$. By the contrapositive of Claim \ref{claim:mmsri1}, the use of forcing gadgets attached to $v_i^1$ and $v_i^2$ implies that we must have $M(v_i^1)\in \{v_i^T,v_i^F\}$ and $M(v_i^2)\in \{v_i^T,v_i^F\}$. Given that $M$ is a matching, it immediately follows that either $\{\{v_i^T,v_i^1\},\{v_i^F,v_i^2\}\}\subseteq M$ or $\{\{v_i^T,v_i^2\},\{v_i^F,v_i^1\}\}\subseteq M$. The strict subset notation of the claim follows from the fact that if equality held, the agents within each forcing gadget would all block with each other, therefore violating that $\max_{a_r\in A[G_{V_i}]}\vert bp_{a_r}(M)\cap E[G_{V_i}]\vert\leq 1$. 
        
        For the second part of the statement, if it is the case that neither $\{\{v_i^T,v_i^1\},\{v_i^F,v_i^2\}\}\subset M$ nor $\{\{v_i^T,v_i^2\}, \{v_i^F,v_i^1\}\}\subset M$, then it must be the case that at least one of $v_i^1,v_i^2$ is matched to some agent within their respective copy of the forcing gadget or is unmatched, in which case Claim \ref{claim:mmsri1} states that $\max_{a_r\in A[G_{V_i}]}\vert bp_{a_r}(M)\cap E[G_{V_i}]\vert> 1$.

        The third part of the statement follows trivially: $v_i^T$ and $v_i^1$ are each other's first choices, and they are matched to each other, and $v_i^2$ is matched to $v_i^F$ whom they prefer to any agent in their forcing gadget. Hence, the only blocking pairs within the $G_{V_i}$ gadget are those admitted by $M$ within the copies of the forcing gadgets. The fourth part of the statement follows from a similar argument, with the additional blocking pair $\{v_i^T,v_i^1\}$ incurred as these agents are each other's first choices.
    \end{proof}

    We will now establish how agents in the clause gadget must be matched for the matching to have desirable properties, and vice versa.

    \begin{claim}
    \label{claim:mmsri3}
        Let $G_{C_j}$ be an occurrence of the clause gadget, and let $M$ be a matching of $I$. If it is true that $\max_{a_r\in A[G_{C_j}]}\vert bp_{a_r}(M)\cap (E[G_{C_j}]\cup\{\{x_j^s,v(x_j^s)\}:1\leq s\leq 3\})\vert\leq 1$, then there exists at least one $x_j^d$ in $G_{C_j}$ such that $M(v(x_j^d))\succeq_{v(x_j^d)}x_j^d$. Conversely, if there exists at least one $x_j^d$ in $G_{C_j}$ such that $M(v(x_j^d))\succeq_{v(x_j^d)}x_j^d$ and, furthermore, $\{x_{j}^{d-1},x_{j}^{d+1}\}\in M$ (addition modulo 3), then $\max_{a_r\in A[G_{C_j}]}\vert bp_{a_r}(M)\cap (E[G_{C_j}]\cup\{\{x_j^s,v(x_j^s)\}:1\leq s\leq 3\})\vert\leq 1$.
    \end{claim}
    \begin{proof}[Proof of Claim \ref{claim:mmsri3}]
    \renewcommand{\qedsymbol}{$\blacksquare$}
        For the first part of the statement, notice that for at least one $x_j^d$ agent in $G_{C_j}$, it must be the case that $M(x_j^d)\in \{v(x_j^d),x_j^d\}$. Then due to the preference cycle, it must necessarily be the case that $\{x_{j}^d,x_{j}^{d-1}\}\in bp(M)$. Now suppose that $\max_{a_r\in A[G_{C_j}]}\vert bp_{a_r}(M)\cap (E[G_{C_j}]\cup\{\{x_j^s,v(x_j^s)\}:1\leq s\leq 3\})\vert\leq 1$. If $M(x_j^d)=v(x_j^d)$, then $v(x_j^d)$ is matched to someone exactly as good as $x_j^d$. If $M(x_j^d)=x_j^d$ (i.e., $x_j^d$ remains unmatched in $M$), then $v(x_j^d)$ must be matched to someone better than $x_j^d$, otherwise $\{x_j^d,v(x_j^d)\}\in bp(M)$ and so $\vert bp_{x_j^d}(M)\vert>1$.

        For the second part of the statement, notice from the preference construction within $G_{C_j}$ that if $\{x_{j}^{d-1},x_{j}^{d+1}\}\in M$, then $M$ admits the blocking pair $\{x_j^d,x_{j}^{d-1}\}$ but $\{x_j^d,x_{j}^{d+1}\}$ is not a blocking pair because $x_{j}^{d+1}$ prefers $x_{j}^{d-1}$ to $x_{j}^d$. Similarly, both of $x_{j}^{d-1}$ and $x_{j}^{d+1}$ prefer each other to $v(x_j^{d-1})$ and $v(x_j^{d+1})$, respectively. Lastly, if $v(x_j^d)$ is matched to someone at least as good as $x_j^d$, then $\{v(x_j^d),x_j^d\}$ is not a blocking pair. Hence, no agent within $G_{C_j}$ is contained in more than one blocking pair.
    \end{proof}

    We now establish one direction of the relationship between satisfiability and minimax almost-stable matchings.

    \begin{claim}
        \label{claim:mmsri4}
        If $B$ is satisfiable then there exists a matching $M$ of $I$ such that $\max_{a_r\in A}\vert bp_{a_r}(M)\vert\leq 1$.
    \end{claim}
    \begin{proof}[Proof of Claim \ref{claim:mmsri4}]
    \renewcommand{\qedsymbol}{$\blacksquare$}
        To establish this claim, we will construct such a matching $M$ explicitly, starting from the empty set $\varnothing$. Consider a satisfying truth assignment $f: V\rightarrow \{T,F\}$. 
        
        For every variable $V_i\in V$, if $f(V_i)=T$ then, in the corresponding variable gadget $G_{V_i}$, add $\{v_i^T,v_i^1\}$ and $\{v_i^F,v_i^2\}$ to $M$. If $f(V_i)=F$ instead, then add $\{v_i^T,v_i^2\}$ and $\{v_i^F,v_i^1\}$ to $M$.

        For every forcing gadget $G_{F(v_i^z)}$ that was used in the construction of the variable gadgets, add $\{f_i^{1,z},f_i^{2,z}\}$, $\{f_i^{3,z},f_i^{4,z}\}$, $\{f_i^{5,z},f_i^{6,z}\}$ and $\{f_i^{7,z},f_i^{8,z}\}$ to $M$.

        For every clause $C_j=(C_j^1\lor C_j^2 \lor C_j^3)$ in $B$, pick one literal $C_j^s$ such that either $C_j^s$ is an unnegated variable $V_s$ and $f(V_s)=T$, or $C_j^s$ is a negated variable $\bar{V_s}$ and $f(V_s)=F$ (i.e., $C_j^s$ satisfies the clause $C_j$; by the assumption that $B$ is satisfiable at least one such literal must exist in each clause). In the corresponding clause gadget $G_{C_j}$, add $\{x_j^{s-1},x_j^{s+1}\}$ (addition and subtraction taken modulo 3) to $M$. This completes the construction of $M$.

        By Claim \ref{claim:mmsri1} and its subsequent discussion, within all forcing gadgets (attached to some agent $v_i^z$), $M$ admits the blocking pairs $\{f_i^{4,z},f_i^{5,z}\}$, $\{f_i^{6,z},f_i^{8,z}\}$ and no others, i.e., no agent is contained in more than one blocking pair unless $v_i^z$ blocks with agents in $P_{v_i^z}$. However, by construction of $I$, each $v_i^z$ is either $v_i^1$ or $v_i^2$, and by construction of $M$, all such agents are matched either to $v_i^T$ or to $v_i^F$.

        By Claim \ref{claim:mmsri2}, within all variable gadgets $G_{V_i}$, where $\{\{v_i^T,v_i^1\},\{v_i^F,v_i^2\}\}\subseteq M$, $M$ admits only the blocking pairs coming from the forcing gadgets. For those variable gadgets where, instead, it is the case that $\{\{v_i^T,v_i^2\},\{v_i^F,v_i^1\}\}\subseteq M$, $M$ admits the blocking pair $\{v_i^T,v_i^1\}$ in addition to the blocking pairs coming from the forcing gadgets. Either way, no agent is contained in more than one blocking pair within any variable gadget. It remains to show that we cannot have additional blocking pairs between agents $v_i^T$ and $v_i^F$ with agents outside of the gadget (i.e., agents in the clause gadgets). 
        
        As shown in the proof of Claim \ref{claim:mmsri3}, within each clause gadget $G_{C_j}$, $M$ admits the blocking pair $\{x_j^d, x_j^{d-1}\}$ (where $x_j^d$ is the agent not matched within the clause gadget and subtraction is taken modulo 3) and no others, i.e., no agent is contained in more than one blocking pair. Similarly, we constructed $M$ such that $M(v(x_j^d))\succ_{v(x_j^d)}x_j^d$, and such that $M(x_j^{d+1})=x_j^{d-1}\succ_{x_j^{d+1}}v(x_j^{d+1})$ and $M(x_j^{d-1})=x_j^{d+1}\succ_{x_j^{d-1}}v(x_j^{d-1})$, so none of the links $\{\{x_j^s,v(x_j^s)\}:1\leq s\leq 3\}$ between variable and clause gadgets is blocking.

        Hence, $M$ is a matching such that no agent from $I$ is contained in more than one blocking pair, i.e., $\max_{a_r\in A}\vert bp_{a_r}(M)\vert\leq 1$ as required.
    \end{proof}

    We will now establish the converse of Claim \ref{claim:mmsri4}.

    \begin{claim}
        \label{claim:mmsri5}
        If there exists a matching $M$ of $I$ such that $\max_{a_r\in A}\vert bp_{a_r}(M)\vert\leq 1$, then $B$ is satisfiable.
    \end{claim}
    \begin{proof}[Proof of Claim \ref{claim:mmsri5}]
    \renewcommand{\qedsymbol}{$\blacksquare$}
        Suppose that there exists a matching $M$ such that $\max_{a_r\in A}\vert bp_{a_r}(M)\vert\leq 1$. We will construct a satisfying truth assignment $f: V\rightarrow \{T,F\}$. By Claim \ref{claim:mmsri2}, for every variable $V_i\in V$ of $B$, it must be the case for the corresponding variable gadget $G_{V_i}$ that either $\{\{v_i^T,v_i^1\},\{v_i^F,v_i^2\}\}\subseteq M$ or $\{\{v_i^T,v_i^2\},\{v_i^F,v_i^1\}\}\subseteq M$ (and clearly not both because $M$ is a matching). In the former case, set $f(V_i)=T$ and, in the latter case, set $f(V_i)=F$.

        Clearly, $f$ is a truth assignment. Now suppose, for the sake of contradiction, that $f$ does not satisfy $B$. Then there exists a clause $C_j=(C_j^1\lor C_j^2 \lor C_j^3)$ in $B$ such that all literals $C_j^d$ are false with respect to the assignment $f$. However, because $\max_{a_r\in A}\vert bp_{a_r}(M)\vert\leq 1$, by Claim \ref{claim:mmsri3}, it must be the case that for at least one $C_j^s$, agent $v(x_j^s)$ must be matched to someone at least as good as $x_j^s$. By Claim \ref{claim:mmsri2}, though, $\{x_j^s,v(x_j^s)\}\notin M$ because each $v(x_j^s)$ is matched to either $v_i^1$ (where $V_i$ is the variable such that either $C_j^s=V_i$ or $C_j^s=\bar{V_i}$) or $v_i^2$, since $\max_{a_r\in A}\vert bp_{a_r}(M)\vert\leq 1$. Because $v(x_j^s)$ prefers $x_j^s$ to $v_i^2$ (again, where $V_i$ is the variable such that either $C_j^s=V_i$ or $C_j^s=\bar{V_i}$), the only better options than $x_j^s$ on $\succ_{v(x_j^s)}$ are $v_i^1$ and possibly the first occurrence of the literal at $x_1(v(x_j^s))$ (if $C_j^s$ corresponds to the second occurrence of the variable, i.e., $x_j^s=x_2(v(x_j^s))$). However, due to the forcing gadget, $v(x_j^s)$ and $x_1(v(x_j^s))$ cannot be matched unless at least one agent is in more than one blocking pair, so the only possible match for $v(x_j^s)$ better than $x_j^s$ is $v_i^1$. Thus, $\{v(x_j^s),v_i^1\}\in M$, in which case literal $C_j^s$ satisfies the clause, contradicting our assumptions about $C_j$. Thus, there cannot exist any unsatisfied clause $C_j$ with respect to truth assignment $f$, so $f$ is a satisfying truth assignment of $B$, i.e., $B$ is satisfiable as required.
    \end{proof}

    Together, Claims \ref{claim:mmsri4}-\ref{claim:mmsri5} establish that {\sc 1-Max-AlmostStable-sri} is {\sf NP-hard} because {\sc (2,2)-e3-sat} is {\sf NP-hard}. This finishes the proof that {\sc $k$-Max-AlmostStable-sri} is {\sf NP-complete}. Finally, observe that, in our reduction, the agents $v_i^z$ (where $z\in\{1,2\}$) connecting forcing and variable gadgets have the longest preference lists. Specifically, for each such $v_i^z$, length$(\succ_{v_i^z})=\vert (v_{i}^T\;v_i^F)\vert+\vert(f_i^{1,z}\;f_i^{2,z}\dots f_i^{8,z})\vert=10$. This finishes the proof of Theorem \ref{thm:1-Maxbounded}.
\end{proof}

We can also extend this intractability result to complete preference lists.

\begin{theorem}
\label{thm:1-Maxcomplete}
    {\sc $k$-Max-AlmostStable-sri} is {\sf NP-complete}, even if $k=1$ and all preference lists are complete.
\end{theorem}
\begin{proof}        
    To show this, we will modify the clause gadget in the reduction given in the proof of Theorem \ref{thm:1-Maxbounded} slightly. Let $I$ be the {\sc sri} instance constructed in the previous reduction, and let $I'=(A',\succ')$ be a copy of $I$ that we will modify. For each clause gadget $G_{C_j}$ in $I$, we introduce one additional agent $x_j^4$ to the respective clause gadget $G_{C_j}'$ in $I'$ and append it to the end of the preference lists of each agent $x_j^1,x_j^2$ and $x_j^3$ (i.e., in the fourth position of their preference lists). Agent $x_j^4$ shall rank the other agents $x_j^1,x_j^2$ and $x_j^3$ of $G_{C_j}'$ in arbitrary order, followed by the attachment of a forcing gadget $G_{F(x_j^4)}$ contributing eight additional agents $\{f_{j}^{1},f_{j}^{2},\dots,f_{j}^{8}\}$. Specifically,
    $$x_j^4 : x_j^1 \; x_j^2 \; x_j^3 \; G_{F(x_j^4)}$$
    A visual example of the gadget construction is shown in Figure \ref{fig:gadgetC2}.
    
    \begin{figure}[!tbh]
        \centering
        \includegraphics[width=7cm]{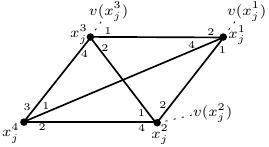}
        \caption{Illustration of the modified clause gadget construction}
        \label{fig:gadgetC2}
    \end{figure}
    
    We can now complete all preference lists of all agents in the instance by appending all previously unranked agents in arbitrary order to the end of the preference lists. The resulting instance $I'$, constructed from an original Boolean formula $B$ with $n$ variables and $m$ clauses, consists of $\vert A'\vert=20n+12m$ agents. We now extend the main reduction argument as follows.

    \begin{claim}
        \label{claim:mmsric}
        There exists a matching $M$ of $I$ such that $\max_{a_r\in A}\vert bp_{a_r}(M)\vert\leq 1$ if and only if there exists a matching $M'$ of $I'$ such that $\max_{a_r\in A}\vert bp_{a_r}(M')\vert\leq 1$.
    \end{claim}
    \begin{proof}[Proof of Claim \ref{claim:mmsric}]
    \renewcommand{\qedsymbol}{$\blacksquare$}
        First, suppose that there exists a matching $M$ of $I$ with $\max_{a_r\in A}\vert bp_{a_r}(M)\vert\leq 1$. To start, let $M'=M$. We have shown in the proof of Theorem \ref{thm:1-Maxbounded} that, except for agents in clause gadgets, every other agent is matched in $M$. Furthermore, we showed that for every clause gadget $G_{C_j}$, there exists at least one agent $x_j^s$ such that $M(x_j^s)=x_j^s$. If there exists exactly one such agent in $G_{C_j}$, then we can update $M'$ such that $M'(x_j^s)=x_j^4$. If all agents in the clause gadget are unmatched in $M$ instead, then we can match $M'(x_j^1)=x_j^2$ and $M'(x_j^3)=x_j^4$, for example. In either case, for every agent $x_j^l$ in $G_{C_j}$, it is the case that $\vert bp_{x_j^l}(M')\vert\leq \vert bp_{x_j^l}(M)\vert$. Furthermore, for the forcing gadget $G_{F(x_j^4)}$ attached to $x_j^4$ in $I'$, we can add $\{f_j^1,f_j^2\}$, $\{f_j^3,f_j^4\}$, $\{f_j^5,f_j^6\}$ and $\{f_j^7,f_j^8\}$ to $M'$. Then $\vert bp_{x_j^4}(M')\vert = 0$. Doing this for every clause gadget establishes that $\max_{a_r\in A}\vert bp_{a_r}(M')\vert\leq 1$.

        Conversely, suppose that there exists a matching $M'$ of $I'$ where $\max_{a_r\in A}\vert bp_{a_r}(M')\vert\leq 1$. Let $M=M'\cap E$, where $E$ contains the edges in the acceptability graph of $I$, and let $M(a_i)=a_i$ for all agents of $I$ that are left unmatched in $M'\cap E$. We claim that $\max_{a_r\in A}\vert bp_{a_r}(M)\vert\leq 1$. Note that $M$ is clearly a matching of $I$. Furthermore, for every agent $a_i$ in $I$, we have that $M(a_i)=M'(a_i)$ if $M'(a_i)$ is an agent in $I$ (i.e., $M'(a_i)\neq x_j^4$ for some $j$) and $\{a_i,M'(a_i)\}$ is acceptable in $I$, and $M(a_i)=a_i$ otherwise. Then $bp_{a_i}(M)\subseteq bp_{a_i}(M')$, so $\max_{a_r\in A}\vert bp_{a_r}(M)\vert\leq\max_{a_r\in A}\vert bp_{a_r}(M')\vert\leq 1$ as required.
    \end{proof}
    
    This finishes the proof of \ref{thm:1-Maxcomplete} as the correspondence established in the claim above shows that {\sc $k$-Max-AlmostStable-sri} is {\sf NP-hard} even if $k=1$ and all preference lists are complete, and {\sc $k$-Max-AlmostStable-sri} clearly remains in {\sf NP} under these restrictions too.
\end{proof}

This strong intractability highlights the computational limitations of enforcing individual-level stability guarantees in these multi-agent settings. Intractability of {\sc Minimax-AlmostStable-sri} follows immediately, but we will state it again formally below. Note that a problem is {\sf para-NP-hard} with respect to a parameter $\kappa$ if it is {\sf NP-hard} already for a constant value of $\kappa$, and a problem is in {\sf XP} with respect to $\kappa$ if there exists an $O(n^{f(\kappa)})$ algorithm. If a problem is {\sf para-NP-hard} with respect to a parameter $\kappa$, then it is not in {\sf XP} with respect to $\kappa$ unless {\sf P}$=${\sf NP} (we refer to \cite{flumgrohe} for an introduction to parametrised complexity theory).  

\begin{corollary}
\label{cor:srihard}
    The problem {\sc Minimax-AlmostStable-sri} is {\sf para-NP-hard} with respect to parameter $\kappa=\min_{M\in \mathcal M}\max_{a_i\in A}\vert bp_{a_i}(M)\vert$, regardless of whether preference lists are of bounded length at most 10 or whether preferences are complete.
\end{corollary}
\begin{proof}
    This follows immediately from Theorems \ref{thm:1-Maxbounded}-\ref{thm:1-Maxcomplete}.
\end{proof}

\subsection{Hardness of Maximum Matchings in Bipartite Bounded-Degree Graphs}
\label{sec:smihard}

In a similar fashion to the previous subsection, the result below states that it is impossible to distinguish in polynomial time even between instances that admit a perfect matching in which every agent is in at most one blocking pair and instances where every perfect matching requires at least one agent to be in more than one blocking pair (unless {\sf P}={\sf NP}).

\begin{theorem}
\label{thm:1-Maxsmi}
    {\sc $k$-Max-AlmostStable-Perfect-smi} is {\sf NP-complete}, even if $k=1$ and all preference lists are of length at most 3.
\end{theorem}
\begin{proof}
    First, note that membership of {\sc $k$-Max-AlmostStable-Perfect-smi} in {\sf NP} follows from the fact that, given a matching $M$, we can efficiently verify whether $M$ is perfect by checking that the size of $M$ is half the number of agents, and whether the maximum number of blocking pairs that any agent is a part of is at most $k$ by iterating through the preference lists.

    We will show that {\sc $k$-Max-AlmostStable-Perfect-smi} is {\sf NP-hard} by reducing from the {\sf NP-complete} problem {\sc (2,2)-e3-sat} to {\sc $1$-Max-AlmostStable-Perfect-smi}. For the reduction itself, we use the same gadgets and leverage a central claim of \citet{biro_sm_10} in their reduction from {\sc (2,2)-e3-sat} to {\sc $k$-BP-AlmostStable-Perfect-smi}. For completeness, we will outline their reduction below.

    Let $B$ be a {\sc (2,2)-e3-sat} instance with the variables $V=\{V_1,V_2,\dots,V_n\}$ and the clauses $C=\{C_1,C_2,\dots,C_m\}$. We will construct an {\sc smi} instance $I=(A,\succ)$ from $B$, where the set of agents $A$ consists of the disjoint union of agents $A_v=\{x_i^r,y_i^r\;;\;1\leq i\leq n\land 1\leq r\leq 4\}$ from the variable gadgets and $A_c=\{c_j^s,p_j^s\;;\;1\leq j\leq m\land 1\leq s\leq 3\}\cup\{q_j,z_j\;;\;1\leq j\leq m\}$ from the clause gadgets. Altogether, $\vert A\vert=8(n+m)$.

    Now, for every variable $V_i\in V$, we construct a \textbf{variable gadget} $G_{V_i}$ consisting of the eight agents $x_i^1,x_i^2,x_i^3,x_i^4,y_i^1,y_i^2,y_i^3,y_i^4$, with preferences as follows: 
    \[
    \begin{minipage}{0.45\textwidth}
    \begin{align*}
    x_i^1 &: y_i^1 \; c(x_i^1) \; y_i^2 \\
    x_i^2 &: y_i^2 \; c(x_i^2) \; y_i^3 \\
    x_i^3 &: y_i^4 \; c(x_i^3) \; y_i^3 \\
    x_i^4 &: y_i^1 \; c(x_i^4) \; y_i^4
    \end{align*}
    \end{minipage}
    \hfill
    \begin{minipage}{0.45\textwidth}
    \begin{align*}
    y_i^1 &: x_i^1 \; x_i^4 \\
    y_i^2 &: x_i^1 \; x_i^2 \\
    y_i^3 &: x_i^2 \; x_i^3 \\
    y_i^4 &: x_i^3 \; x_i^4
    \end{align*}
    \end{minipage}
    \]
    where the $c(x_i^r)$ entries will be specified later. The construction is illustrated in Figure \ref{fig:gadgetvsmi}. Notice that each $G_{V_i}$ gadget only admits two perfect matchings:
    \begin{align*}
        M^1_i&=\{\{x_i^1,y_i^1\},\{x_i^2,y_i^2\},\{x_i^3,y_i^3\},\{x_i^4,y_i^4\}\}\text{ and}\\
        M^2_i&=\{\{x_i^1,y_i^2\},\{x_i^2,y_i^3\},\{x_i^3,y_i^4\},\{x_i^4,y_i^1\}\},
    \end{align*}
    which admit the blocking pairs $\{x_i^3,y_i^4\}$ and $\{x_i^1,y_i^1\}$, respectively, within the gadget (not counting blocking pairs that may also involve $c(x_i^r)$ entries), respectively. Intuitively, these matchings will later correspond to $V_i=T$ and $V_i=F$, respectively.

    \begin{figure}[!tbh]
        \centering
        \includegraphics[width=6cm]{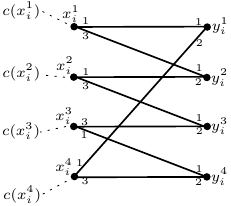}
        \caption{Illustration of the variable gadget construction}
        \label{fig:gadgetvsmi}
    \end{figure}

    For every clause $C_j=(C_j^1\lor C_j^2 \lor C_j^3)\in C$, we construct a \textbf{clause gadget} $G_{C_j}$ consisting of the eight agents $c_j^1, c_j^2, c_j^3, p_j^1, p_j^2, p_j^3, q_j, z_j$, with preferences constructed as follows: 
    \[
    \begin{minipage}{0.45\textwidth}
    \begin{align*}
    c_{j}^1 &: p_{j}^1 \; x(c_{j}^1) \; q_{j} \\
    c_{j}^2 &: p_{j}^2 \; x(c_{j}^2) \; q_{j} \\
    c_{j}^3 &: p_{j}^3 \; x(c_{j}^3) \; q_{j} \\
    z_{j} &: p_{j}^1 \; p_{j}^2 \; p_{j}^3
    \end{align*}
    \end{minipage}
    \hfill
    \begin{minipage}{0.45\textwidth}
    \begin{align*}
    p_{j}^1 &: c_{j}^1 \; z_{j} \\
    p_{j}^2 &: c_{j}^2 \; z_{j} \\
    p_{j}^3 &: c_{j}^3 \; z_{j} \\
    q_{j} &: c_{j}^1 \; c_{j}^2 \; c_{j}^3
    \end{align*}
    \end{minipage}
    \]
    where the $x(c_j^s)$ entries will be specified below. The construction is illustrated in Figure \ref{fig:gadgetcsmi}. Notice that each $G_{C_j}$ gadget only admits three different perfect matchings: 
    \begin{align*}
        M^1_j&=\{\{c_{j}^1,q_{j}\},\{c_{j}^2,p_{j}^2\},\{c_{j}^3,p_{j}^3\},\{p_{j}^1,z_{j}\}\},\\
        M^2_j&=\{\{c_{j}^1,p_{j}^1\},\{c_{j}^2,q_{j}\},\{c_{j}^3,p_{j}^3\},\{p_{j}^2,z_{j}\}\}, \text{ and}\\
        M^3_j&=\{\{c_{j}^1,p_{j}^1\},\{c_{j}^2,p_{j}^2\},\{c_{j}^3,q_{j}\},\{p_{j}^3,z_{j}\}\},
    \end{align*}
    which admit the blocking pairs $\{c_{j}^1,p_{j}^1\}$, $\{c_{j}^2,p_{j}^2\}$ and $\{c_{j}^3,p_{j}^3\}$ within the gadget (not counting blocking pairs that may also involve $x(c_j^s)$ entries), respectively. Intuitively, these matchings correspond to the first, second, and third literals of $C_j$ being true, respectively.

    \begin{figure}[!tbh]
        \centering
        \includegraphics[width=7.3cm]{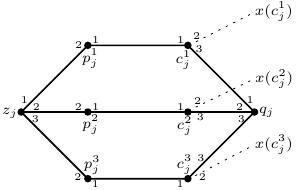}
        \caption{Illustration of the clause gadget construction}
        \label{fig:gadgetcsmi}
    \end{figure}

    The preference list entries that we previously left to be specified are completed as follows: let $C_j^s$ correspond to the $s$th literal in $C_j$ ($1\leq j\leq m, 1\leq s\leq 3$), which is either an unnegated or a negated occurrence of some variable $V_i$. If, processing $B$ from left to right, $c_j^s$ is the
    \begin{itemize}
        \item first unnegated occurrence of $V_i$, then let $x(c_j^s)=x_i^1$ and let $c(x_i^1)=c_j^s$;
        \item second unnegated occurrence of $V_i$, then let $x(c_j^s)=x_i^2$ and let $c(x_i^2)=c_j^s$;
        \item first negated occurrence of $V_i$, then let $x(c_j^s)=x_i^3$ and let $c(x_i^3)=c_j^s$;
        \item second negated occurrence of $V_i$, then let $x(c_j^s)=x_i^4$ and let $c(x_i^4)=c_j^s$.
    \end{itemize} 
    
    We also refer to these connections as the communication edges between variable and clause gadgets.

    Recall that this construction is the same as given by \citet{biro_sm_10}. The authors proved the following correspondence between satisfiability and blocking pairs.

    \begin{claim}[\cite{biro_sm_10}]
        \label{claim:birosm}
        $B$ is satisfiable if and only if $I$ admits a perfect matching $M$ such that $\vert bp(M)\vert\leq n+m$.
    \end{claim}
    
    Now recall that we already pointed out the structure of perfect matchings in the clause and variable gadgets, none of which can involve a communication edge. Furthermore, we noted the necessary blocking pairs incurred from each of the possible matchings. Specifically, we pointed out that every perfect matching must incur at least one blocking pair within each clause and within each variable gadget. Hence, for any perfect matching $M$ of $I$, $\vert bp(M)\vert\geq n+m$ by construction. In particular, notice that if $I$ admits a perfect matching $M$ with at most $n+m$ blocking pairs, then no communication edge is blocking $M$. Observe that, in this case, $\max_{a_r\in A}\vert bp_{a_r}(M)\vert\leq 1$. Hence, if $B$ is satisfiable, then $I$ admits a perfect matching $M$ with $\max_{a_r\in A}\vert bp_{a_r}(M)\vert\leq 1$.

    Now, conversely, suppose that $I$ admits a perfect matching $M$ with $\max_{a_r\in A}\vert bp_{a_r}(M)\vert\leq 1$. Then, because $M$ is a perfect matching, within every $G_{C_j}$ gadget there exists exactly one $c_j^s$ agent who is matched to agent $q_j$, and every other agent $c_j^{s'}$ (for $s\neq s'$) is matched to $p_j^{s'}$. Clearly no $q_j$ and no $z_j$ agent is part of a blocking pair of $M$. Thus, notice that $\vert bp(M)\vert>n+m$ only if some communication edge is blocking. However, $p_j^{s'}$ is the first choice of $c_j^{s'}$, so $c_j^{s'}$ does not block with their communication edge. Furthermore, $c_j^s$ does block with $p_j^s$ by construction, but we assume that $\max_{a_r\in A}\vert bp_{a_r}(M)\vert\leq 1$, so $c_j^s$ cannot block with their communication edge. Hence, no communication edge is blocking, in which case $\vert bp(M)\vert \leq n+m$, in which case $B$ is satisfiable as \citet{biro_sm_10} proved.
    
    Therefore, we have also established the following claim.

    \begin{claim}
        \label{claim:nocommblocking}
        $I$ admits a perfect matching $M$ such that $\vert bp(M)\vert\leq n+m$ if and only if $I$ admits a perfect matching in which no communication edge is blocking if and only if $I$ admits a perfect matching $M$ with $\max_{a_r\in A}\vert bp_{a_r}(M)\vert\leq 1$.
    \end{claim}
    
    Together with Claim \ref{claim:birosm}, we can conclude that $B$ is satisfiable if and only if $I$ admits a perfect matching $M$ such that $\max_{a_r\in A}\vert bp_{a_r}(M)\vert\leq 1$. This establishes that {\sc 1-Max-AlmostStable-Perfect-smi} is {\sf NP-hard}, even if all preference lists are of length at most 3, and thus completes the proof that {\sc $k$-Max-AlmostStable-Perfect-smi} is {\sf NP-complete}.
\end{proof}

\begin{corollary}
\label{cor:maxsmihard}
    The problem {\sc Minimax-AlmostStable-Max-smi} is {\sf para-NP-hard} with respect to parameter $\kappa=\min_{M\in \mathcal{M}^+}\max_{a_i\in A}\vert bp_{a_i}(M)\vert$ even if all preference lists are of length at most 3. By generalisation, this also applies to {\sc Minimax-AlmostStable-Max-sri}.
\end{corollary}
\begin{proof}
    This follows immediately from Theorem \ref{thm:1-Maxsmi}. 
\end{proof}

Notice that this result cannot be extended to {\sc smi} instances with complete preference lists, as all stable matchings of such instances are maximum-cardinality matchings (notice: any two unmatched agents on opposite sides are acceptable and thus would be blocking).

\subsection{Further Intractability Results}
\label{sec:furtherhard}

In Theorem \ref{thm:1-Maxbounded}, we proved that, even in the restricted case where $k=1$, {\sc $k$-Max-AlmostStable-sri} is {\sf NP-complete}, and we could therefore conclude immediately in Corollary \ref{cor:srihard} that {\sc Minimax-AlmostStable-sri} is {\sf para-NP-hard} with respect to the optimal solution value. Here, this problem differs from other almost-stable matching problems such as {\sc MinBP-AlmostStable-sri} and {\sc MinBA-AlmostStable-sri}, which are in {\sf XP} with respect to the optimal solution value \cite{abraham06,chen17}. Given this strong intractability frontier, it is interesting to know whether {\sc Minimax-AlmostStable-sri} becomes efficiently solvable for larger optimal solution values. However, we will show below that our original reduction can be extended to show that {\sc $k$-Max-AlmostStable-sri} remains {\sf NP-complete} for any positive integer $k$ (although with preference list lengths increasing proportionally to $k$).

\begin{theorem}
\label{thm:k-Maxhardsri}
    {\sc $k$-Max-AlmostStable-sri} is {\sf NP-complete} for any fixed positive integer $k$.
\end{theorem}
\begin{proof}
    In the proof of Theorem \ref{thm:1-Maxbounded}, we already noted that {\sc $k$-Max-AlmostStable-sri} is in {\sf NP}. To prove this new intractability result, we modify the reduction in the remainder of that proof slightly. Specifically, we need to generalise the forcing gadget $G_{F(a_r)}$. Recall that the forcing gadget attaches to an agent $a_r$ with original preference list $P_{a_r}$, contributes eight new agents, and enforces that in any matching $M$ where no agent is in more than one blocking pair, there exists some agent $a_s\in P_{a_r}$ such that $\{a_r,a_s\}\in M$. Conversely, whenever no such agent $a_s$ exists on $P_{a_r}$, it must be the case that, for any matching, at least one agent is in more than one blocking pair. 

    We introduce the following parametrised version of the forcing gadget: let $G_{F_\omega(a_r)}$ (for some $\omega\in\mathbb N$) denote the generalisation of $G_{F(a_r)}$ which introduces $3^w-1$ additional agents (rather than 8) and has the nested preference cycles constructed as in Proposition \ref{prop:lbsri}. Specifically, $a_r$ retains $P_{a_r}$ at the top of their preference list, followed by the nested preference list construction. We highlight that $G_{F_2(a_r)}=G_{F(a_r)}$. 
    
    Now, for an arbitrary $G_{F_\omega(a_r)}$, if $a_r$ is not matched to some agent on $P_{a_r}$, then the result from Proposition \ref{prop:lbsri} applies, in which case $\min_{M\in\mathcal M}\max_{a_r\in A}\vert bp_r(M)\vert\geq \omega$. The contrapositive of this immediately implies that if $\max_{a_r\in A}\vert bp_r(M)\vert< \omega$ for some matching $M$, then $\{a_r,a_s\}\in M$ for some $a_s\in P_{a_r}$. Notice that this naturally extends Claim \ref{claim:mmsri1}.

    Now, for the reduction, let $k\geq 1$ be fixed and let $B$ be a boolean formula given as an instance of the {\sf NP-complete} problem {\sc (2,2)-e3-sat}, where $V=\{V_1,V_2,\dots,V_n\}$ is the set of variables and $C=\{C_1,C_2,\dots,C_m\}$ is the set of clauses. We construct an instance $I=(A,\succ)$ of {\sc $k$-Max-AlmostStable-sri} as follows. Let $A=A_{F_1}\cup A_{F_2} \cup A_U \cup A_V \cup A_C$ be a set of $2n(3^{k+1}+1)+m((k-1)3^{k+2}+3k)$ agents, where
    
    \begin{align*}
        A_{F_1}&=\{f_{1,i}^{\alpha,z} \;;\; 1\leq i\leq n \land 1\leq \alpha\leq 3^{k+1}-1\land 1\leq z\leq 2\},\\
        A_{F_2}&=\{f_{2,j}^{\alpha,z,\beta} \;;\; 1\leq j\leq m \land 1\leq \alpha\leq 3^{k+1}-1\land 1\leq z\leq 3\land 1\leq\beta\leq k-1\},\\
        A_U &=\{u_\beta^1(x_j^s),\;u_\beta^2(x_j^s)\;;\;1\leq j\leq m \land 1\leq s\leq 3\land 1\leq\beta\leq k-1\},\\
        A_V&=\{ v_i^T, v_i^F,v_i^1,v_i^2 \;;\; 1\leq i\leq n\}, \text{ and}\\
        A_C&=\{ x_j^s \;;\; 1\leq j\leq m \land 1\leq s\leq 3\}.
    \end{align*}
    These disjoint sets of agents correspond to those from forcing gadgets attached to agents in the variable gadgets, forcing gadgets attached to agents in the clause gadgets, blocking-pair introducing gadgets in the clause gadgets, the variable gadgets, and the clause gadgets, respectively. Notice that, although the size of $I$ is exponential in $k$, it is polynomial in $n$ and $m$ (the size of $B$), and we treat $k$ as a fixed constant. Therefore, this is a polynomial-time reduction (in the size of $B$).

    We start by constructing the variable and clause gadgets as in the proof of Theorem \ref{thm:1-Maxbounded}, but replacing the use of forcing gadgets $G_{F(a_r)}$ in the variable gadgets with copies of $G_{F_{k+1}(a_r)}$ (the agents from these copies make up $A_{F_1}$). We refer back to Figures \ref{fig:gadgetV}-\ref{fig:gadgetC} for the original constructions. Next, for every $x_j^s$ agent in each clause gadget, we connect the $2(k-1)$ agents $\{u_\beta^1(x_j^s),\;u_\beta^2(x_j^s)\;\vert\;1\leq\beta\leq k-1\}$ (which collectively make up $A_U$) and transform $x_j^s$'s preference list from the original construction 
    $$ x_j^s: x_j^{s+1}\; x_j^{s-1} \; v(x_j^s) $$
    (where addition and subtraction is taken modulo 3) to a new preference list
    $$ x_j^s: x_j^{s+1}\; x_j^{s-1} \; v(x_j^s)\; u_1^1(x_j^s)\;\;u_2^1(x_j^s)\;\dots\; u_{k-1}^1(x_j^s). $$ 
    We define the preferences of the $u_\beta^1(x_j^s)$ and $u_\beta^2(x_j^s)$ agents as follows:
    \begin{align*}
        u_\beta^1(x_j^s)&: x_j^s\; u_\beta^2(x_j^s)\; \text{ , and }\\
        u_\beta^2(x_j^s)&: u_\beta^1(x_j^s).
    \end{align*}
    Finally, we attach one copy of the $G_{F_{k+1}(u_\beta^2(x_j^s))}$ gadget to each $u_\beta^2(x_j^s)$ agent (the agents of these gadgets make up $A_{F_2}$). The construction is illustrated in Figure \ref{fig:fixedKhardSRI}.

    \begin{figure}[!tbh]
        \centering
        \includegraphics[width=8.5cm]{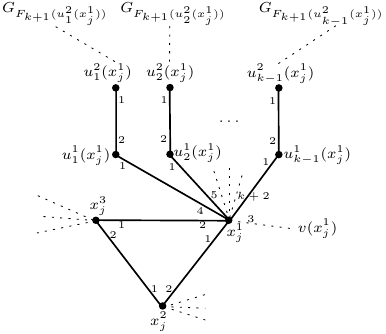}
        \caption{Schematic illustration of the modified clause gadget construction on $x_j^1$}
        \label{fig:fixedKhardSRI}
    \end{figure}

    Notice that Claim \ref{claim:mmsri3}, which argued how agents in the clause gadget must be matched for the matching to have desirable properties, can be extended to this new construction as follows.

    \begin{claim}
    \label{claim:mmKsri3}
        Let $G_{C_j}^*$ be an occurrence of the clause gadget together with its attached agents in $A_U$ and $A_{F_2}$\footnote{We use the $^*$ notation to indicate that this is a generalised clause gadget compared to the construction in the proof of Theorem \ref{thm:1-Maxbounded}.}, and let $M$ be a matching of $I$. If it is true that $\max_{a_r\in A[G_{C_j}^*]}\vert bp_{a_r}(M)\vert\leq k$, then there exists at least one $x_j^d$ in $G_{C_j}^*$ such that $M(v(x_j^d))\succeq_{v(x_j^d)}x_j^d$. Conversely, if there exists at least one $x_j^d$ in $G_{C_j}^*$ such that $M(v(x_j^d))\succeq_{v(x_j^d)}x_j^d$ and, furthermore, $\{x_{j}^{d-1},x_{j}^{d+1}\}\in M$ (addition modulo 3), then the remaining agents in $G_{C_j}^*$ can be matched among each other such that $\max_{a_r\in A[G_{C_j}^*]}\vert bp_{a_r}(M)\vert\leq k$.
    \end{claim}
    \begin{proof}[Proof of Claim \ref{claim:mmKsri3}]
    \renewcommand{\qedsymbol}{$\blacksquare$}
        For the first part of the statement, suppose that $\max_{a_r\in A[G_{C_j}^*]}\vert bp_{a_r}(M)\vert\leq k$ and notice that, for at least one $x_j^d$ agent in $G_{C_j^*}$, it must be the case that $v(x_j^d)\succeq_{x_j^d}M(x_j^d)$ (i.e., $x_j^d$ must either be matched to their corresponding agent in the variable gadget, one of the attached agents in $A_U$, or remain unmatched). Then, due to the preference cycle $(x_j^1,x_j^2,x_j^3)$ in the first and second preference list positions of all three involved agents, it must necessarily be the case that $\{x_{j}^d,x_{j}^{d-1}\}\in bp(M)$. If $M(x_j^d)=v(x_j^d)$, then $v(x_j^d)$ is matched to someone exactly as good as $x_j^d$, satisfying our requirement. Otherwise, it must be the case that $v(x_j^d)\succ_{x_j^d}M(x_j^d)$, in which case $x_j^d$ must either be matched to a $u_\beta^1(x_j^d)$ agent or remain unmatched. However, by having a forcing gadget $G_{F_{k+1}(u_\beta^2(x_j^d))}$ attached to every $u_\beta^2(x_j^d)$ agent, every $u_\beta^2(x_j^d)$ is either matched to someone on $P_{u_\beta^2(x_j^d)}$ (which consists of only agent $u_\beta^1(x_j^d)$ by construction), or some agent must be in at least $k+1$ blocking pairs otherwise. Thus, by assumption, we must have $\{u_\beta^1(x_j^d),u_\beta^2(x_j^d)\}\in M$ for every $1\leq \beta\leq k-1$. Therefore, $x_j^d$ must remain unmatched. Because every $u_\beta^1(x_j^d)$ agent prefers $x_j^d$ to $u_\beta^2(x_j^d)$, though, it is the case that $\{\{x_{j}^d,x_{j}^{d-1}\},\{x_j^d,u_1^1(x_j^d)\},\dots,\{x_j^d,u_{k-1}^1(x_j^d)\}\}\subseteq bp_{x_{j}^d}(M)$, so $\vert bp_{x_{j}^d}(M)\vert=k$ unless $x_j^d\succ_{v(x_j^d)}M(v(x_j^d))$, in which case also $\{x_j^d,v(x_j^d)\}\in bp_{x_{j}^d}(M)$ and so $\vert bp_{x_{j}^d}(M)\vert=k+1$, a contradiction. Thus, $M(v(x_j^d))\succ_{v(x_j^d)}x_j^d$ in this case, proving our statement.

        For the second part of the statement, notice from the preference construction within $G_{C_j}^*$ that if $\{x_{j}^{d-1},x_{j}^{d+1}\}\in M$, then $M$ admits the blocking pair $\{x_j^d,x_{j}^{d-1}\}$, but $\{x_j^d,x_{j}^{d+1}\}$ is not a blocking pair because $x_{j}^{d+1}$ prefers $x_{j}^{d-1}$ to $x_{j}^d$. Similarly, both of $x_{j}^{d-1}$ and $x_{j}^{d+1}$ prefer each other to $v(x_j^{d-1})$ and $v(x_j^{d+1})$, respectively. Lastly, if $v(x_j^d)$ is matched to someone at least as good as $x_j^d$, then either $\{v(x_j^d),x_j^d\}$ is a match or $M(v(x_j^d))\succ_{v(x_j^d)}x_j^d$ -- either way, $\{v(x_j^d),x_j^d\}$ is not a blocking pair. Now, if we match $\{u_\beta^1(x_j^z),u_\beta^2(x_j^z)\}\in M$ for every $1\leq \beta\leq k-1$ and $1\leq z\leq 3$, and match the agents in the forcing gadgets as described in Proposition \ref{prop:lbsri}, then if $\{v(x_j^d),x_j^d\}$ is a match in $M$, $x_j^d$ does not block with any $u_\beta^1(x_j^d)$ agent. Otherwise, we incur the blocking pairs $\{x_j^d,u_1^1(x_j^d)\},\dots,\{x_j^d,u_{k-1}^1(x_j^d)\}$. $x_j^{d+1}$ and $x_j^{d-1}$ do not block with any agents in $A_U$ as they are matched to each other and prefer each other to any agent in $A_U$. We also incur blocking pairs within each forcing gadget, but no more than $k$ by construction, as every $u_\beta^2(x_j^z)$ agent is matched to agent $u_\beta^1(x_j^z)$ on $P_{u_\beta^2(x_j^z)}$. Thus, indeed, no agent is in more than $k$ blocking pairs as required.
    \end{proof}

    We now establish the relationship between satisfiability and minimax almost-stable matchings.

    \begin{claim}
        \label{claim:satmaxksri}
        If $B$ is satisfiable then there exists a matching $M$ of $I$ such that $\max_{a_r\in A}\vert bp_{a_r}(M)\vert \leq k$.
    \end{claim}    
    \begin{proof}[Proof of Claim \ref{claim:satmaxksri}]
    \renewcommand{\qedsymbol}{$\blacksquare$}
        Let $f: V \rightarrow \{T,F\}$ be a satisfying truth assignment of $B$ and consider the matching $M$ constructed as follows (starting from the empty set). 

        For every variable $V_i\in V$, if $f(V_i)=T$ then, in the corresponding variable gadget $G_{V_i}$, add $\{v_i^T,v_i^1\}$ and $\{v_i^F,v_i^2\}$ to $M$. If $f(V_i)=F$ instead, then add $\{v_i^T,v_i^2\}$ and $\{v_i^F,v_i^1\}$ to $M$.

        For every forcing gadget $G_{F_{k+1}(v_i^z)}$ that was used in the construction of the variable gadgets, add $\{f_{1,i}^{1,z},f_{1,i}^{2,z}\}$, $\{f_{1,i}^{3,z},f_{1,i}^{4,z}\},\dots, \{f_{1,i}^{3^{k+1}-2,z},f_{1,i}^{3^{k+1}-1,z}\}$ to $M$.

        For every clause $C_j=(C_j^1\lor C_j^2 \lor C_j^3)$ in $B$, pick one literal $C_j^d$ such that either $C_j^d$ is an unnegated variable $V_s$ and $f(V_s)=T$, or $C_j^d$ is a negated variable $\bar{V_s}$ and $f(V_s)=F$ (i.e., $C_j^d$ satisfies the clause $C_j$; by the assumption that $B$ is satisfiable at least one such literal must exist in each clause). In the corresponding clause gadget $G_{C_j}$, add $\{x_j^{d-1},x_j^{d+1}\}$ (addition and subtraction taken modulo 3) to $M$ and leave $x_j^d$ unmatched. 

        Match all agents $\{u_\beta^1(x_j^z),u_\beta^2(x_j^z)\}$ for every $1\leq \beta\leq k-1$, $1\leq j\leq m$ and $1\leq z\leq 3$. Furthermore, for every forcing gadget $G_{F_{k+1}(u_\beta^2(x_j^s))}$ that is attached to an agent $u_\beta^2(x_j^s)$, add $\{f_{2,j}^{1,s,\beta},f_{2,j}^{2,s,\beta}\}$, $\{f_{2,j}^{3,s,\beta},f_{2,j}^{4,s,\beta}\},$ $\dots,$ $\{f_{2,j}^{3^{k+1}-2,s,\beta},f_{2,j}^{3^{k+1}-1,s,\beta}\}$ to $M$. This completes the construction of $M$.

        By our earlier discussions regarding the generalised forcing gadget, because each $v_i^1$ and each $v_i^2$ is matched to someone on $P_{v_i^1}$ and $P_{v_i^2}$ in $M$, respectively, no agent in any forcing gadget $G_{F_{k+1}(v_i^z)}$ attached to one of these two agents is in more than $k$ blocking pairs.

        By our previous Claim \ref{claim:mmsri2}, within all variable gadgets where $\{\{v_i^T,v_i^1\},\{v_i^F,v_i^2\}\}\subseteq M$, $M$ admits only the blocking pairs coming from the forcing gadgets. For those variable gadgets where, instead, it is the case that $\{\{v_i^T,v_i^2\},\{v_i^F,v_i^1\}\}\subseteq M$, $M$ admits the blocking pair $\{v_i^T,v_i^1\}$ in addition to the blocking pairs coming from the forcing gadgets (which cannot involve $v_i^1$ and $v_i^2$). Either way, no agent is contained in more than $k$ blocking pairs within any variable gadget. 
        
        As shown in the proof of Claim \ref{claim:mmKsri3}, within each clause gadget (including the respective $A_U$ and $A_{F_2}$ agents), $M$ is constructed such that no agent is contained in more than $k$ blocking pairs. 

        Hence, $M$ is a matching such that no agent from $I$ is contained in more than $k$ blocking pairs, i.e., $\max_{a_r\in A}\vert bp_{a_r}(M)\vert\leq k$ as required.
    \end{proof}

    Now, let us prove the converse direction.
    
    \begin{claim}
        \label{claim:maxksatsri}
        If there exists a matching $M$ of $I$ such that $\max_{a_r\in A}\vert bp_{a_r}(M)\vert \leq k$, then $B$ is satisfiable.
    \end{claim}    
    \begin{proof}[Proof of Claim \ref{claim:maxksatsri}]
    \renewcommand{\qedsymbol}{$\blacksquare$}
        Suppose that there exists a matching $M$ of $I$ such that $\max_{a_r\in A}\vert bp_{a_r}(M)\vert \leq k$. We will construct a satisfying truth assignment $f: V \rightarrow \{T,F\}$. By a simple extension of Claim \ref{claim:mmsri2}, for every variable $V_i\in V$ of $B$, it must be the case for the corresponding variable gadget that either $\{\{v_i^T,v_i^1\},\{v_i^F,v_i^2\}\}\subseteq M$ or $\{\{v_i^T,v_i^2\},\{v_i^F,v_i^1\}\}\subseteq M$ (and clearly not both because $M$ is a matching), since $\max_{a_r\in A}\vert bp_{a_r}(M)\vert\leq k$, and otherwise some agent in either $G_{F_{k+1}(v_i^1)}$ or $G_{F_{k+1}(v_i^1)}$ (for any $1\leq i\leq n$) must be in at least $k+1$ blocking pairs by construction. In the former case, set $f(V_i)=T$ and, in the latter case, set $f(V_i)=F$.

        Clearly, $f$ is a truth assignment. Now to show that $f$ satisfies $B$, consider any clause $C_j=(C_j^1\lor C_j^2 \lor C_j^3)$ in $B$. Because $\max_{a_r\in A}\vert bp_{a_r}(M)\vert\leq k$, by Claim \ref{claim:mmKsri3}, it must be the case that, for at least one literal $C_j^d$, agent $v(x_j^d)$ must be matched to someone at least as good as $x_j^d$; otherwise, no matter how the agents are matched, at least one agent must be in more than $k$ blocking pairs. However, we know that $\{x_j^d,v(x_j^d)\}\notin M$, because we already argued above that each $v(x_j^d)$ is matched to either $v_i^1$ (where $V_i$ is the variable such that either $C_j^d=V_i$ or $C_j^d=\bar{V_i}$) or $v_i^2$. Because $v(x_j^d)$ prefers $x_j^d$ to $v_i^2$ (again, where $V_i$ is the variable such that either $C_j^d=V_i$ or $C_j^d=\bar{V_i}$), the only better options than $x_j^d$ on $\succ_{v(x_j^d)}$ are $v_i^1$ and possibly the first occurrence of the literal at $x_1(v(x_j^d))$ (if $C_j^d$ corresponds to the second occurrence of the variable, i.e., $x_j^d=x_2(v(x_j^d))$). However, due to the forcing gadgets on the variable agents, $\{v(x_j^d),x_1(v(x_j^d))\}$ cannot be in the matching without causing at least one agent to be in at least $k+1$ blocking pairs, so the only possible match for $v(x_j^d)$ better than $x_j^d$ is $v_i^1$. Thus, if $v(x_j^d)=v_i^T$ (for some $1\leq i\leq n$), then $\{v_i^T,v_i^1\}\in M$, in which case we must also have $\{v_i^F,v_i^2\}\in M$ (as otherwise $v_i^2$ is not matched to an agent on $P_{v_i^2}$, which results in too many blocking pairs as argued above), and if $v(x_j^d)=v_i^F$ (for some $1\leq i\leq n$), then $\{v_i^F,v_i^1\}\in M$, in which case we must also have $\{v_i^T,v_i^2\}\in M$. In the former case, we defined $f(V_i)=T$, and in the latter case we defined $f(V_i)=F$, so either way the literal $C_j^d$ satisfies the clause $C_j$. Thus, there cannot exist any unsatisfied clause $C_j$ with respect to truth assignment $f$, so $f$ is a satisfying truth assignment of $B$, i.e., $B$ is satisfiable as required. 
    \end{proof}

    Together, Claims \ref{claim:satmaxksri}-\ref{claim:maxksatsri} establish that {\sc $k$-Max-AlmostStable-sri} is {\sf NP-hard}. This finishes the proof of Theorem \ref{thm:k-Maxhardsri}.
\end{proof}

We can also extend our reduction from Theorem \ref{thm:1-Maxsmi} for {\sc $k$-Max-AlmostStable-Max-smi} to any positive integer $k$.

\begin{theorem}
\label{thm:k-Maxhardsmi}
    {\sc $k$-Max-AlmostStable-Max-smi} is {\sf NP-complete} for any fixed positive integer $k$.
\end{theorem}
\begin{proof}
    Let $k$ be fixed. In the proof of Theorem \ref{thm:1-Maxsmi}, we already noted that {\sc $k$-Max-AlmostStable-sri} is in {\sf NP}. To show that the problem is {\sf NP-hard}, we start with the construction of {\sc smi} instance $I=(A,\succ)$ from the proof of Theorem \ref{thm:1-Maxsmi}. Now, we will modify this construction to an {\sc smi} instance $J=(A^J,\succ^J)$. We will use a similar approach as in the proof of Theorem \ref{thm:k-Maxhardsri} to force agents in clause gadgets to have a higher number of blocking pairs when they are not matched to their first choice. For every $c_j^s$ agent in each clause gadget, we introduce $2(k-1)$ additional agents $U=\{u_\alpha^1(c_j^s),\;u_\alpha^2(c_j^s)\;\vert\;1\leq\alpha\leq k-1\}$ and transform $c_j^s$'s preference list from the original construction 
    $$ c_j^s: p_j^s\; x(c_j^s)\; q_j $$
    to a new construction
    $$ c_j^s: p_j^s\;\; x(c_j^s)\;\; u_1^1(c_j^s)\;\;u_2^1(c_j^s)\;\dots\; u_{k-1}^1(c_j^s)\;\; q_j $$ 
    We also define the preferences of the $u_\alpha^1(c_j^s)$ and $u_\alpha^2(c_j^s)$ agents as follows:
    \begin{align*}
        u_\alpha^1(c_j^s)&: c_j^s\; u_\alpha^2(c_j^s)\; \text{ , and }\\
        u_\alpha^2(c_j^s)&: u_\alpha^1(c_j^s).
    \end{align*}
    The construction is illustrated in Figure \ref{fig:fixedKhardSMI}.

    \begin{figure}[!tbh]
        \centering
        \includegraphics[width=7.3cm]{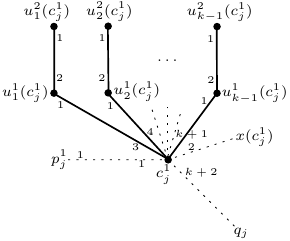}
        \caption{The modified clause gadget construction illustrated on a $c_j^1$ agent}
        \label{fig:fixedKhardSMI}
    \end{figure}

    Notice that now the size of $A^J$ increases from the original $8(n+m)$ agents (where $n$ is the number of variables and $m$ is the number of clauses) to $8n+2m(1+3k)$ agents, which is still polynomial in the size of the {\sc (2,2)-e3-sat} instance $B$ that we originally reduce from (in the construction of $I$).

    Furthermore, we know from our previous reduction that $I$ admits perfect matchings. Hence, by inspection of our modification above, the modified instance $J$ admits perfect matchings, and, for every such matching $M\in\mathcal{M}^p(J)$, we must have $\{u_\alpha^1(c_j^s),u_\alpha^2(c_j^s)\}\in M$ for all $c_j^s$ and for all $\alpha$. Hence, our modification introduces no new flexibility with regard to the perfect matchings. Also, note that, for every $M\in\mathcal{M}^p(J)$, we must have exactly one $c_j^s$ agent matched to $q_j$ in every clause, in which case $\{c_j^s,u_\alpha^1(c_j^s)\}\in bp^J(M)$. Hence, $\vert bp_{c_j^s}(M)\vert\geq k-1$, where each of the $k-1$ blocking pairs come from the $u_\alpha^1(c_j^s)$ agents.

    From Claims \ref{claim:birosm}-\ref{claim:nocommblocking} we know that $B$ is satisfiable $\Leftrightarrow$ $I$ admits a perfect matching with at most $n+m$ blocking pairs $\Leftrightarrow$ $I$ admits a perfect matching in which no communication edge between clause and variable gadgets is blocking $\Leftrightarrow$ $I$ admits a perfect matching $M$ with $\max_{a_r\in A}\vert bp_{a_r}^I(M)\vert\leq 1$. 

    It remains to show that $I$ admits a perfect matching $M$ with $\max_{a_r\in A}\vert bp_{a_r}^I(M)\vert\leq 1$ if and only if $J$ admits a perfect matching $M$ with $\max_{a_r\in A^J}\vert bp_{a_r}^J(M)\vert\leq k$.

    First, suppose that $I$ admits a perfect matching $M$ with $\max_{a_r\in A}\vert bp_{a_r}^I(M)\vert\leq 1$. Consider the matching $M'=M\cup\{\{u_\alpha^1(c_j^s),u_\alpha^2(c_j^s)\}\;\vert\; \forall c_j^s\;\forall\alpha \}$. Clearly, $M'$ is a perfect matching of $J$. Furthermore, because we add no more than $k-1$ entries to any agent's preference list and maintain the relative preference order of $\succ$ in $\succ^J$, we have that, for all $a_r\in A$, $\vert bp_{a_r}^J(M')\vert\leq \vert bp_{a_r}^I(M)\vert+k-1$. Thus, $\vert bp_{a_r}^J(M')\vert\leq k$ as required.

    For the converse direction, suppose that $J$ admits a perfect matching $M$ that satisfies $\max_{a_r\in A^J}\vert bp_{a_r}^J(M)\vert\leq k$. Notice that no communication edge can be blocking in $M$: suppose, for the sake of contradiction, that there exists some $c_j^s$ agent that blocks with $x(c_j^s)$. Then because $M$ is perfect, $c_j^s$ must be matched to $q_j$ and, by construction of the preference lists and our observation above, we must have that $bp_{c_j^s}^J(M)=\{\{c_j^s,p_j^s\}, \{c_j^s,x(c_j^s)\},\{c_j^s,u_1^1(c_j^s)\},\{c_j^s,u_2^1(c_j^s)\},\dots,$ $\{c_j^s,u_{k-1}^1(c_j^s)\}\}$, in which case, $\vert bp_{c_j^s}^J(M)\vert=k+1>k$, a contradiction. Thus, it is easy to see that $M'=M\vert_I$ (i.e., $M$ restricted to pairs that involve agents in $A$) is a perfect matching of $I$ in which no communication edge is blocking, in which case $\max_{a_r\in A}\vert bp_{a_r}^I(M')\vert\leq 1$.
\end{proof}

In fact, we can also extend this new reduction above even further to establish a remarkably tight intractability bound as follows.

\begin{theorem}
\label{thm:k-MaxhardNsmi}
    For any $\varepsilon\in(0,1)$, the problem {\sc $k$-Max-AlmostStable-Max-smi} is {\sf NP-complete} for $k=\left\lfloor N^{1-\varepsilon}\right\rfloor$, where $N$ is the number of agents.
\end{theorem}
\begin{proof}
    The same reduction from a {\sc (2,2)-e3-sat} instance $B$ as in the proof of Theorem \ref{thm:k-Maxhardsmi} applies -- all that remains to show is that an instance with the correct number of agents exists and can be constructed in polynomial time (in the size of $B$).

    We argued previously that our construction produces an {\sc smi} instance with $N=8n+2m(1+3k)$ agents. Now substitute $k$ to see that $N=8n+2m\left(1+3\left\lfloor N^{1-\varepsilon}\right\rfloor\right)$ and notice that $N$ is an integer value (we will later show that $N$ is strictly positive under our assumptions on $n$ and $m$). We already noted that {\sc (2,2)-e3-sat} instances satisfy $m=\frac{4}{3}n$ and that $n\in3\mathbb N$. Let 
    \begin{align*}
        q&=1-\varepsilon,\\
        k&=\left\lfloor N^{1-\varepsilon}\right\rfloor=\left\lfloor N^{q}\right\rfloor,\\
        r&=n/3,\\
        X&=8n+2m=32r \text{, and}\\
        Y&=6m=24r.
    \end{align*}
    Then 
    \begin{align*}
        N&=8n+2m\left(1+3\left\lfloor N^{1-\varepsilon}\right\rfloor\right)\\
        &=X+Y\cdot k\\
        &=r(32+24k),
    \end{align*}
    where the non-negative integer $k=\left\lfloor N^{q}\right\rfloor$ must satisfy $k^{1/q}\leq X+Y\cdot k<(k+1)^{1/q}$. A suitable $k$ must exist -- consider the continuous function $\phi(k)=k^{1/q}-Y\cdot k$ for $k\in[0,\infty)$. Clearly, by $q\in(0,1)$ and $X\in\mathbb N$, $\phi(0)=0$ and $\lim_{k\rightarrow\infty}\phi(k)=\infty$. Notice, furthermore, that $\phi$ is convex, has a unique minimum and two (not necessarily distinct) $x$-axis intercepts at $k_0=0$ and $k_1\in[0,\infty)$ such that $\phi$ is non-positive on $[k_0=0,k_1]$, and strictly positive and monotonically increasing on $(k_1,\infty)$. By the intermediate value theorem, there must exist a point $t_*>0$ such that $\phi(t_*)=X$. By $X>0$, clearly $t_*>k_1$. Let $k_*=\lfloor t_*\rfloor$. By definition, $k_*\leq t_*<k_*+1$, so by being monotonically increasing on $(k_1,\infty)$, we must have that $\phi(k_*)\leq \phi(t_*)=X<\phi(k_*+1)$, which is equivalent to $k_*^{1/q}-Y\cdot k_*\leq X<(k_*+1)^{1/q}-Y\cdot (k_*+1)$. Adding $Y\cdot k_*$ to all parts of the inequality, and keeping in mind that $Y=6m>0$, yields $k_*^{1/q}\leq X+Y\cdot k_*<(k_*+1)^{1/q}-Y< (k_*+1)^{1/q}$ as above. This establishes the existence of an appropriate construction with $N$ agents. 

    We set out to argue that the reduction can be performed in polynomial time (in the size of $B$), so let us show how to choose such a suitable value of $k$ that ensures that $k=\lfloor N^{1-\varepsilon}\rfloor$ efficiently. Note that if $k\geq \frac{X}{Y}=\frac{32r}{24r}=\frac{4}{3}$, then $X+Y\cdot k\leq 2Y\cdot k$, so from $k^{1/q}\leq X+Y\cdot k$ we get that $k\leq (X+Y\cdot k)^q\leq (2Y\cdot k)^q$. Rearranging, multiplying both sides by $1/k^q$ and keeping in mind that $k>0$, we get $k^{1-q}\leq (2Y)^q$ and so $k\leq (2Y)^{\frac{q}{1-q}}$. Thus, finally, we get that $0\leq k\leq k_{\max}$, where $k_{\max}=\left\lceil \max\left\{\frac{X}{Y},(2Y)^{\frac{q}{1-q}}\right\} \right\rceil$. Thus, the search for $k$ (and hence for $N$) requires at most $O(k_{\max})=O(r^{\frac{q}{1-q}})=O(n^{\frac{1-\varepsilon}{\varepsilon}})$ time, by considering at all integers between 0 and $k_{\max}$, which is polynomial in $n$ (as $\varepsilon$ is fixed) and thus also in the size of $B$ as required.    
\end{proof}

\section{Efficient Algorithms for Very Short Preference Lists}
\label{sec:short}

While we have shown strong intractability results even in cases where preference lists are of bounded length, we will now consider the case where preference lists are very short. As has been shown for similar optimisation problems such as {\sc MinBP-AlmostStable-sri} and {\sc MinBP-AlmostStable-Max-smi}, preference lists of length at most 2 guarantee a sufficiently simple instance structure to allow for optimal algorithms that can be executed in polynomial time. In Section \ref{sec:srishort}, we will give such an algorithm for minimax matchings in {\sc sri}. Then, in Section \ref{sec:smishort}, we will extend our algorithm to the {\sc Max-smi} case. Although the setting where preference lists are of length at most 2 is very constrained, we highlight that these algorithms apply whenever choices are very limited. Furthermore, the algorithms can also be useful in the following scenario: agents may indicate their top 2 choices, then we compute a minimax almost-stable maximum-cardinality matching, and then we find a maximum-cardinality matching among the remaining unmatched agents. This ensures desirable stability guarantees with respect to the top 2 choices of agents, and, subject to this, minimises the number of unmatched agents.

\subsection{Minimax Matchings in General Graphs with Max-Degree 2}
\label{sec:srishort}

We now give an algorithm for {\sc Minimax-AlmostStable-sri} that solves the problem to optimality in linear time whenever preference lists are of length at most 2. The procedure is given in Algorithm \ref{alg:sriexact2} and works as follows: determine whether the instance $I$ admits a stable matching $M_S$ using Irving's algorithm \cite{irving_sr}. If yes, immediately return $M_S$. If not, compute a maximum-cardinality matching $M_C$ as follows. Notice that because preference lists are of length at most 2, the acceptability graph of $I$ has maximum degree at most 2, i.e., it consists only of paths and cycles. Thus, to compute $M_C$, we can simply include every second edge of every path, starting from one of the endpoints, and every second edge of every cycle (such that no included edges have a vertex in common). It is easy to see that this results in a maximum-cardinality matching for graphs with maximum degree at most 2. Next, we perform a rematching procedure on $M_C$ that ensures that no agent is contained in more than one blocking pair, and return this modified matching.

\begin{algorithm}[!htb]
\renewcommand{\algorithmicrequire}{\textbf{Input:}}
\renewcommand{\algorithmicensure}{\textbf{Output:}}

\begin{algorithmic}[1]
\Require{$I$ : an {\sc sri} instance}
\Ensure{$M$ : a matching}

\State $M_S \gets $ {\sf Irving}$(I)$ \Comment{Compute a stable matching if one exists}

\If{$M_S$ exists}
    \State\Return{$M_S$}
\EndIf

\State $M_C \gets $ {\sf MaxCard}$(I)$ \Comment{Compute a maximum-cardinality matching}

\While{$\exists \;a_i$ such that $\vert bp_{a_i}(M_C)\vert=2$} \Comment{Rematching procedure}
    \State pick one agent $a_r$ that is in a blocking pair with $a_i$
    \State $a_k\gets M_C(a_r)$
    \State $M_C(a_i)\gets a_r \;;\; M_C(a_r)\gets a_i \;;\; M_C(a_k)\gets a_k$
\EndWhile
\State\Return{$M_C$}

\end{algorithmic}
\caption{Exact algorithm for {\sc Minimax-AlmostStable-sri} with lists of length $\leq2$}
\label{alg:sriexact2}
\end{algorithm}

\begin{theorem}
\label{thm:sriexact2}
    Let $I$ be an {\sc sri} instance with $n$ agents. If all preference lists of $I$ are of length at most 2, then $\min_{M\in \mathcal{M}}\max_{a_i\in A}\vert bp_{a_i}(M)\vert\leq 1$ and Algorithm \ref{alg:sriexact2} computes an optimal solution to {\sc Minimax-AlmostStable-sri} in $O(n)$ time.
\end{theorem}
\begin{proof}
    Clearly, if $M_S$ exists, then $\max_{a_i\in A}\vert bp_{a_i}(M_S)\vert\leq 0$, so $M_S$ is optimal. Furthermore, if $M_S$ does not exist, then $\min_{M\in \mathcal{M}}\max_{a_i\in A}\vert bp_{a_i}(M)\vert> 0$.
   
    Now suppose that $\max_{a_i\in A}\vert bp_{a_i}(M_C)\vert> 1$, then, because all preference lists are of length at most 2, there exists an agent $a_i$ who is unmatched in $M_C$ and who is contained in exactly two blocking pairs, say $bp_{a_i}(M_C)=\{\{a_i,a_r\},\{a_i,a_s\}\}$ (and suppose that $a_r\succ_i a_s$). Clearly, both $a_r$ and $a_s$ must be matched in $M_C$, otherwise $M_C$ would not be a maximum-cardinality matching as $M_C\cup\{\{a_i,a_r\}\}$ or  $M_C\cup\{\{a_i,a_s\}\}$ would be a matching of larger size. Also, without loss of generality, let $a_k$ denote the current partner of $a_r$ in $M_C$, then $a_i\succ_r a_k$, otherwise $a_i,a_r$ would not block. Now we update $M_C$ such that $a_i$ and $a_r$ are matched to each other (instead of $a_r$ to $a_k$). Due to the bound on the preference list lengths, $a_k$'s preference list either contains just $a_r$, or two agents $a_r$ and $a_l$ (for some agent $a_l$). Note, though, that $a_l$ must be matched in $M_C$, otherwise $M_C'= (M_C\setminus\{\{a_r,a_k\}\})\cup\{\{a_i,a_r\},\{a_k,a_l\}\}$ would exceed $M_C$ in size and thus $M_C$ would not be a maximum-cardinality matching.

    After this rematching procedure completes, the size of $M_C$ remains the same as previously, but $\vert bp_{a_i}(M_C)\vert= 0$ (because now $a_i$ is matched to their first choice $a_r$). Furthermore, because $a_i\succ_r a_k$, $\vert bp_{a_r}(M_C)\vert= 0$. Finally, $\vert bp_{a_k}(M_C)\vert\leq 1$ because while $a_k$ can generally block with at most two agents due to its preference list of length at most two, it cannot block with $a_r$ because $a_i\succ_r a_k$. Thus, this rematching procedure strictly decreases the number of agents who are contained in more than one blocking pair at each iteration of the main while loop. While at least one such agent remains, we repeat this rematching procedure with this agent taking the role of $a_i$ in the argument above.

    Asymptotically, Irving's algorithm requires at most linear time in the sum of the preference list lengths of $I$ \cite{irving_sr}, which, due to bounded preference list lengths, is linear in the number of agents of $I$. The maximum-cardinality matching algorithm we described simply walks over paths and cycles, so it requires $O(m)$ time for a graph with $m$ edges. We execute this on the acceptability graph of $I$, which has $n$ vertices and, again due to bounded preference list lengths, $m=O(n)$ edges. Hence, computing $M_C$ requires $O(n)$ time. The rematching procedure, if invoked, is executed at most $O(n)$ times because the number of agents in more than one blocking pair strictly decreases with each iteration (as argued above), and each execution requires constant time. We can implement the while loop efficiently by building a stack of unmatched agents (every agent involved in two blocking pairs must clearly be unmatched) and, at every iteration of the loop, popping an agent from this stack and verifying whether this agent is indeed blocking with both agents on its preference list (this can clearly be performed in constant time). Because we do the computation of the initial $M_C$ and the rematching sequentially, and assuming that matching size comparisons can be implemented in constant time, we arrive at an overall time complexity of $O(n)$.
\end{proof}

This result identifies a tractable frontier. Even though minimax almost-stability is generally intractable, instances with very short preference lists admit an efficient computation. Such instances can arise in practical scenarios where agents may only choose between very limited options.

\subsection{Minimax Maximum Matchings in Bipartite Graphs with Max-Degree 2}
\label{sec:smishort}

We now show how the high-level technique of Algorithm \ref{alg:sriexact2} for low-degree acceptability graphs in {\sc sri} carries over to {\sc Minimax-AlmostStable-Max-smi}. Specifically, the procedure we give in Algorithm \ref{alg:smiexact} solves {\sc Minimax-AlmostStable-Max-smi} to optimality when preference lists are of length at most 2. It works as follows: compute a stable matching $M_S$ using the Gale-Shapley algorithm \cite{gale_shapley} and a maximum-cardinality matching $M_C$ using the same procedure as the one that we described for Algorithm \ref{alg:sriexact2} (as the acceptability graph still only consists of paths and cycles). If $\vert M_S\vert = \vert M_C\vert$ then return $M_S$, otherwise perform the same rematching procedure on $M_C$ as in Algorithm \ref{alg:sriexact2} and return this modified matching.

\begin{algorithm}[!htb]
\renewcommand{\algorithmicrequire}{\textbf{Input:}}
\renewcommand{\algorithmicensure}{\textbf{Output:}}

\begin{algorithmic}[1]
\Require{$I$ : an {\sc smi} instance}
\Ensure{$M$ : a matching}

\State $M_S \gets $ {\sf GaleShapley}$(I)$ \Comment{Compute a stable matching}
\State $M_C \gets $ {\sf MaxCard}$(I)$ \Comment{Compute a maximum-cardinality matching}

\If{$\vert M_S\vert=\vert M_C\vert$}
    \State\Return{$M_S$}
\EndIf

\While{$\exists \;a_i$ such that $\vert bp_{a_i}(M_C)\vert=2$} \Comment{Rematching procedure}
    \State pick one agent $a_r$ that is in a blocking pair with $a_i$
    \State $a_k\gets M_C(a_r)$
    \State $M_C(a_i)\gets a_r \;;\; M_C(a_r)\gets a_i \;;\; M_C(a_k)\gets a_k$
\EndWhile
\State\Return{$M_C$}

\end{algorithmic}
\caption{Exact algorithm for {\sc Minimax-AlmostStable-Max-smi} with lists of length $\leq 2$}
\label{alg:smiexact}
\end{algorithm}

\begin{theorem}
\label{thm:smiexact}
    Let $I$ be an {\sc smi} instance with $n$ agents. If all preference lists of $I$ are of length at most 2, then $\min_{M\in \mathcal{M}^+}\max_{a_i\in A}\vert bp_{a_i}(M)\vert\leq 1$ and Algorithm \ref{alg:smiexact} computes an optimal solution to {\sc Minimax-AlmostStable-Max-smi} in $O(n)$ time.
\end{theorem}
\begin{proof}
    Notice that all stable matchings of $I$ have the same size \cite{galesoto85}. Clearly, $\vert M_C\vert \geq \vert M_S\vert$ and $\max_{a_i\in A}\vert bp_{a_i}(M_S)\vert= 0$, so if $\vert M_S\vert = \vert M_C\vert$, then $M_S$ is simultaneously a maximum-cardinality matching and a stable matching of $I$, so it is also a solution to {\sc Minimax-AlmostStable-Max-smi}. However, if $\vert M_S\vert \neq \vert M_C\vert$, then necessarily $\vert M_S\vert < \vert M_C\vert$, so no solution to {\sc Minimax-AlmostStable-Max-smi} is stable, i.e., for all maximum cardinality matchings $M'$ of $I$, $\max_{a_i\in A}\vert bp_{a_i}(M')\vert> 0$. 
    
    Now suppose that $\max_{a_i\in A}\vert bp_{a_i}(M_C)\vert> 1$, then, because all preference lists are of length at most 2, there exists an agent $a_i$ who is unmatched in $M_C$ and who is contained in exactly two blocking pairs, say $bp_{a_i}(M_C)=\{\{a_i,a_r\},\{a_i,a_s\}\}$. Clearly, both $a_r$ and $a_s$ must be matched in $M_C$, otherwise $M_C$ would not be a maximum-cardinality matching as $M_C\cup\{\{a_i,a_r\}\}$ or  $M_C\cup\{\{a_i,a_s\}\}$ would be a matching of larger size. Also, without loss of generality, let $a_k$ denote the current partner of $a_r$ in $M_C$, then $a_i\succ_r a_k$, otherwise $a_i,a_r$ would not block in $M_C$. Now we update $M_C$ such that $a_i$ and $a_r$ are matched to each other (instead of $a_r$ and $a_k$). Furthermore, due to the bound on the lengths of preference lists, $a_k$'s preference list either contains just $a_r$, or two agents $a_r$ and $a_l$ (for some agent $a_l$). Agent $a_l$ must be matched in $M_C$, otherwise $M_C'= (M_C\setminus\{\{a_r,a_k\}\})\cup\{\{a_i,a_r\},\{a_k,a_l\}\}$ would exceed $M_C$ in size and thus $M_C$ would not be a maximum-cardinality matching.

    After this rematching procedure completes, the size of $M_C$ remains the same as previously, but $\vert bp_{a_i}(M_C)\vert\leq 1$. Furthermore, because $a_i\succ_r a_k$, $\vert bp_{a_r}(M_C)\vert= 0$. Finally, $\vert bp_{a_k}(M_C)\vert\leq 1$ because while $a_k$ can generally block with at most two agents due to its preference list of length at most two, it cannot block with $a_r$ because $a_i\succ_r a_k$. Thus, this rematching procedure strictly decreases the number of agents who are contained in more than one blocking pair. While at least one such agent remains, we repeat this rematching procedure with this agent taking the role of $a_i$ in the argument above.

    Asymptotically, the Gale-Shapley algorithm requires at most linear time in the sum of the preference list lengths of $I$ \cite{gale_shapley}, which, due to bounded preference list lengths, is linear in the number of agents of $I$. The maximum-cardinality matching we described simply walks over all paths and cycles, so due to the maximum degree 2 assumption on the acceptability graph, computing $M_C$ requires $O(n)$ time. The rematching procedure, if invoked, is executed at most $O(n)$ times because the number of agents in more than one blocking pair strictly decreases with each iteration (as argued above), and each execution requires constant time. We can implement the while loop efficiently by building a stack of unmatched agents (every agent involved in two blocking pairs must clearly be unmatched) and, at every iteration of the loop, popping an agent from this stack and verifying whether this agent is indeed blocking with both agents on its preference list (this can clearly be performed in constant time). Because we do the computation of the initial $M_S$, $M_C$ and the rematching sequentially, and assuming that matching size comparisons can be implemented in constant time, we arrive at an overall time complexity of $O(n)$.
\end{proof}

This positive result implies a tight tractability frontier: {\sc Minimax-AlmostStable-Max-smi} can be solved very efficiently when preference lists are of length at most 2. However, the problem is {\sf NP-hard} when preference lists are of length at most 3, as previously established in Corollary \ref{cor:maxsmihard}.

\section{Approximating Minimax Almost-Stable Matchings}
\label{sec:approx}

At this point, we have established the practical relevance of our central optimisation problems, but provided strong intractability results for their computation (except in very restricted cases). Apart from considering restricted cases, another standard technique of dealing with {\sf NP-hard} problems of practical relevance and theoretical interest is to design efficient approximation algorithms \cite{Vazirani2001}. However, we can immediately state the following negative observation based on our intractability results.

\begin{corollary}
\label{cor:inapprox2}
    The problems {\sc Minimax-AlmostStable-sri} and {\sc Minimax-AlmostStable-Max-smi} admit no polynomial-time approximation algorithms with a guarantee better than 2, unless {\sf P}={\sf NP}.
\end{corollary}
\begin{proof}
    Notice that an efficient algorithm with a guarantee $c<2$ would contradict that the problems {\sc $1$-Max-AlmostStable-sri} and {\sc $1$-Max-AlmostStable-Perfect-smi} are {\sf NP-hard} (as we showed in Theorems \ref{thm:1-Maxbounded} and \ref{thm:1-Maxsmi}, respectively), as we could distinguish in polynomial time between instances that admit a matching (respectively perfect matching) in which no agent is in more than $c\cdot\text{OPT}<2$ blocking pairs, i.e., every agent is in at most one blocking pair, and instances where this is not true.
\end{proof}

A trivial positive observation is the following: both {\sc Minimax-AlmostStable-sri} and {\sc Minimax-AlmostStable-Max-smi} can be efficiently approximated within the maximum preference list length $d_{\max}$, as we can efficiently decide whether the optimal (minimum) maximum number OPT of blocking pairs that any agent is contained in is 0 or at least 1 (see our discussion at the beginning of Section \ref{sec:intractability}). Hence, we can either return a stable matching (which is optimal) or the empty matching (which admits at most $d_{\max}\cdot$ OPT blocking pairs per agent). Can we do better than this? 

In Section \ref{sec:approxDsri}, we will give a natural algorithm for {\sc Minimax-AlmostStable-sri} that gives a more desirable guarantee of $\left\lceil\frac{d_{\max}}{2}\right\rceil$ (in fact, we provide a stronger guarantee: no agent blocks with more than half of their preference list, regardless of the optimal value). Then, in Section \ref{sec:inappproxsmi}, we will strengthen the inapproximability result from Corollary \ref{cor:inapprox2} to rule out (subject to {\sf P$\neq$NP}) many classes of desirable approximation algorithms for {\sc Minimax-AlmostStable-Max-smi}. 


\subsection{Approximating Minimax Matchings in General Graphs}
\label{sec:approxDsri}

Beyond the trivial upper bound OPT $\leq d_{\max}\leq n-1$ that we stated at the beginning of the section, we can make the following observation. Let $I=(A,\succ)$ be an {\sc sri} instance with complete preference lists and $n$ agents. Then, for any cut of $A$ (i.e., a partition of $A$ into two sets) consisting of (roughly) equally sized sets $A_1$ and $A_2$ (if $n$ is odd, let $\vert A_1\vert=\vert A_2\vert+1$), the {\sc smi} instance $I'=(A,\succ')$, where $\succ'$ is constructed from $\succ$ by making all agents within $A_1$ and all agents within $A_2$ mutually unacceptable, admits a stable matching $M$. Clearly, $\max_{a_i\in A} \vert bp_i^{I'}(M)\vert=0\leq\max_{a_i\in A} \vert bp_i^{I}(M)\vert\leq \vert A_1\vert-1=\left\lceil\frac{n}{2} \right\rceil-1$.

It turns out that we can generalise this observation to arbitrary {\sc sri} instances with bounded preference lists of length $d_{\max}$ by connecting this problem to a restricted local-search version of the classical {\sc Max-Cut} problem \cite{garey_johnson}. Our approximation algorithm below for {\sc Minimax-AlmostStable-sri} is inspired by folklore heuristics and procedures using similar flip-based operations, such as \cite{kl70,JOHNSON198879}. In particular, we highlight in advance that our algorithm does not really depend on the preference rankings at all, other than an initial step that checks whether the instance is solvable, but rather exploits the acceptability graph in a purely graph-theoretical way, which allows us to make desirable approximation guarantee deductions.

\begin{algorithm}[!htb]
\renewcommand{\algorithmicrequire}{\textbf{Input:}}
\renewcommand{\algorithmicensure}{\textbf{Output:}}

\begin{algorithmic}[1]
\Require{$I=(A,\succ)$ : an {\sc sri} instance}
\Ensure{$M$ : a matching}

\State $M \gets $ {\sf Irving}$(I)$ \Comment{Compute a stable matching or determine that none exists}
\If{$M$ is not None}
    \State\Return $M$
\EndIf

\State side$(a_i)\gets 0\;\;\forall a_i\in A$ 
\State cross$(a_i)\gets 0\;\;\forall a_i\in A$ 
\State $Q\gets \{a_i\in A\;\vert\; d_i > 2\;\cdot$ cross$(a_i)\}$ \Comment{As $d_i>0$, this yields $Q\gets A$, but we maintain this for clarity}

\While{$Q\neq \varnothing$}
    \State $a_i\gets Q$.pop()
    \State oldside $\gets$ side$(a_i)$
    \State side$(a_i)\gets$ $1\;-$ oldside
    \For{$a_j$ on $\succ_i$}
        \If{side$(a_j)$ is oldside}
            \State cross$(a_j)\gets$ cross$(a_j)+1$
        \Else
            \State cross$(a_j)\gets$ cross$(a_j)-1$
        \EndIf
        \If{$a_j\in Q$ and $d_j \leq 2\;\cdot$ cross$(a_j)$}
            \State $Q$.remove$(a_j)$
        \ElsIf{$a_j\notin Q$ and $d_j > 2\;\cdot$ cross$(a_j)$}
            \State $Q$.add$(a_j)$
        \EndIf
    \EndFor
    \State cross$(a_i)\gets$ $d_i\;-$ cross$(a_i)$  \Comment{We have flipped crossing and non-crossing entries}
\EndWhile

\State $I'\gets(A,\succ')$ \Comment{Where initially $\succ'_i$ is empty for all $a_i\in A$}

\For{$a_i\in A$}
    \For{$a_j$ on $\succ_i$ (in decreasing order of preference)}
        \If{side$(a_i)\neq$ side$(a_j)$}
            \State $\succ_i'$.append$(a_j)$
        \EndIf
    \EndFor
\EndFor

\State $M \gets $ {\sf GaleShapley}$(I')$ \Comment{Compute a stable matching of this new bipartite instance}

\State\Return{$M$}

\end{algorithmic}
\caption{Approximation algorithm for {\sc Minimax-AlmostStable-sri}}
\label{alg:localapprox}
\end{algorithm}

The full procedure is given in Algorithm \ref{alg:localapprox} and we provide the following intuition. Let $I=(A,\succ)$ be an {\sc sri} instance with preference lists of length at most $d_{\max}$ and $n$ agents. If $I$ is unsolvable, then we aim to find a cut $(A_1,A_2)$ of the acceptability graph $G$ associated with $I$ such that, informally, no agent has ``too many'' edges removed. More specifically, let cross$(a_i)$ denote the number of mutually acceptable agents of $a_i\in A$ in the opposite set (i.e., if $a_i\in A_1$, then cross$(a_i)$ is the number of acceptable agents in $A_2$, and vice versa). The local search algorithm starts with an arbitrary cut (we choose $A_1=A$ and $A_2=\varnothing$, but any partition is valid) represented by an array ``side'' (that assigns either a 0 or 1 to every agent), and initialises cross$(a_i)=0$ for all agents. We then keep track of all agents whose preference list length strictly exceeds twice their crossing number throughout the execution (which is all agents to start with). While this set is non-empty, we remove one such agent from the set and flip its set membership in the cut. Consequently, we make necessary adjustments to the crossing numbers and memberships of its neighbours. Once no such agent remains in the set, we have a desirable cut in which every agent ranks at least half as many agents from the opposite set as they rank in total (in $\succ$). Thus, we can consider the {\sc smi} instance $I'$ induced by the spanning bipartite subgraph associated with this final cut $(A_1,A_2)$, and compute a stable matching $M$ of $I'$ using the Gale-Shapley algorithm \cite{gale_shapley}. Then, by construction, $M$ has the property that, for all $a_i\in A$, $\vert bp_i^{I'}(M)\vert=0\leq \vert bp_i^{I}(M)\vert\leq d_i-\left\lfloor\frac{d_i}{2}\right\rfloor=\left\lceil\frac{d_i}{2}\right\rceil$ (where $d_i$ is the length of $a_i$'s preference list $\succ_i$ in $\succ$). Thus, necessarily, $\max_{a_i\in A} \vert bp_i^{I'}(M)\vert=0\leq\max_{a_i\in A} \vert bp_i^{I}(M)\vert\leq \left\lceil\frac{d_{\max}}{2}\right\rceil$. The guarantees of this algorithm and its time complexity are summarised in the following statement. 

\begin{theorem}
\label{thm:approxsri}
    Let $I=(A,\succ)$ be an {\sc sri} instance with preference lists of length at most $d_{\max}$ and $n$ agents. Algorithm \ref{alg:localapprox} returns a matching $M$ of $I$ such that $\vert bp_i(M)\vert\leq \left\lceil\frac{d_i}{2}\right\rceil$ and thus $\max_{a_i\in A} \vert bp_i(M)\vert\leq \left\lceil\frac{d_{\max}}{2}\right\rceil$. Furthermore, the algorithm returns a $\left\lceil\frac{d_{\max}}{2}\right\rceil$-approximation for {\sc Minimax-AlmostStable-sri} and terminates in at most $O(n\cdot d_{\max}^2)$ time.
\end{theorem}
\begin{proof}
    We start by noting that, if $I$ is solvable, then we return a stable matching right away, which we compute using Irving's algorithm \cite{irving_sr}. Otherwise, we know that the optimal solution value of the problem is at least 1. Notice, also, that when we pop an agent $a_i$ from $Q$ in the while loop, we never need to add it back into $Q$ in the same iteration. To see this, notice that, at the start of the loop iteration, $d_i>2\;\cdot$ cross$(a_i)$ and at the end of the iteration, we update $a_i$'s crossing number to be $d_i\;-$ cross$(a_i)$. Hence, the new crossing number is at least $d_i-\frac{d_i}{2}=\frac{d_i}{2}$, in which case $a_i$ does not satisfy the requirement to be part of $Q$.

    Next, we argue why the algorithm must terminate. Let $\Pi$ be a potential function of the cut $(A_1,A_2)$ (represented by the ``side'' array in the pseudocode) defined as $\Phi(A_1,A_2)=\frac{1}{2}\sum_{a_i\in A}\text{cross}(a_i)$, and, for every $a_i\in A$, let int$(a_i)=d_i\;-$ cross$(a_i)$ denote the number of mutually acceptable agents of $a_i$ in the same set. Notice that $0\leq \Phi(A_1,A_2)\leq \frac{1}{2}\sum_{a_i\in A}d_i\leq \frac{1}{2}n\cdot d_{\max}$. Now, if some $a_i$ satisfies int$(a_i)>$ cross$(a_i)$, i.e., whenever $d_i>2\;\cdot$ cross$(a_i)$, then flipping the membership of $a_i$ from one set to the opposite set converts $a_i$'s int$(a_i)$ many internal incident edges into crossing edges, and its cross$(a_i)$ many crossing edges into internal edges. Hence, this causes a net increase of at least 1 for $\Phi$, as we will now show. Let $(A_1',A_2')$ be the cut after the flip, and let cross$'$ be the crossing numbers after the flip. Then
    {\allowdisplaybreaks
    \begin{align*}
        \Phi(A_1',A_2')&=\frac{1}{2}\left(\text{cross}'(a_i)+\sum_{a_j\in A\;:\; a_j\text{ on }\succ_i}\text{cross}'(a_j)+\sum_{a_k\in A\setminus\{a_i\}\;:\; a_k\text{ not on }\succ_i}\text{cross}'(a_k)\right)\\
        &=\frac{1}{2}\left(\text{cross}'(a_i)+\sum_{a_j\in A\;:\; a_j\text{ on }\succ_i}\text{cross}'(a_j)+\sum_{a_k\in A\setminus\{a_i\}\;:\; a_k\text{ not on }\succ_i}\text{cross}(a_k)\right).
    \end{align*}}
    
    Therefore, notice that
    {\allowdisplaybreaks
    \begin{align*}
        &\Phi(A_1',A_2')-\Phi(A_1,A_2)\\
        &=\frac{1}{2}\left(\text{cross}'(a_i)+\sum_{a_j\in A\;:\; a_j\text{ on }\succ_i}\text{cross}'(a_j)+\sum_{a_k\in A\setminus\{a_i\}\;:\; a_k\text{ not on }\succ_i}\text{cross}(a_k)\right)\\
        &\;\;\; -\frac{1}{2}\left(\text{cross}(a_i)+\sum_{a_j\in A\;:\; a_j\text{ on }\succ_i}\text{cross}(a_j)+\sum_{a_k\in A\setminus\{a_i\}\;:\; a_k\text{ not on }\succ_i}\text{cross}(a_k)\right)\\
        &=\frac{1}{2}\left( \text{cross}'(a_i)-\text{cross}(a_i) +\sum_{a_j\in A\;:\; a_j\text{ on }\succ_i}(\text{cross}'(a_j)-\text{cross}(a_j)) \right).
    \end{align*}}
    
    By assumption, int$(a_i)>$ cross$(a_i)$, and both of these values are integers. Hence, int$(a_i)\geq$ cross$(a_i)+1$. Furthermore, at the end of the while loop, we update the crossing numbers such that $\text{cross}'(a_i)=$ int$(a_i)$. Therefore, $\text{cross}'(a_i)-\text{cross}(a_i)\geq 1$. Furthermore, flipping set membership of $a_i$ has the following effects on the crossing numbers of the agents on $\succ_i$: every agent $a_j$ that was in the original set of $a_i$ now has their crossing number increased by 1 (because $a_i$-$a_j$ are now crossing), and every agent $a_j$ that was in the opposite set of $a_i$ and is now in the same set of $a_i$ has their crossing number decreased by 1. Clearly, there are no side effects between other agents, as only $a_i$'s set membership is changed. Hence, $\sum_{a_j\in A\;:\; a_j\text{ on }\succ_i}(\text{cross}'(a_j)-\text{cross}(a_j))=\text{int}(a_i)-\text{cross}(a_i)\geq 1$. Thus
    \begin{align*}
        &\Phi(A_1',A_2')-\Phi(A_1,A_2)\\
        &\geq\frac{1}{2}\left( 1 +1 \right)=1.
    \end{align*}
    Therefore, any flip strictly increases the value of $\Phi$ by an integer amount. Because $\Phi$ is bounded from above by $\frac{1}{2}n\cdot d_{\max}$, the while loop must terminate.

    The stated bounds follow easily: when the algorithm terminates, we know that, for every agent, their crossing number is at least half their original preference list length. By definition of the crossing number, this means that at least half of the agents on their preference list are part of the opposite set. Thus, by construction of $\succ'$, every agent in $I'$ has a preference list of length at least half the length of their preference list in $\succ$. Now, $M$ is a stable matching of $I'$ and, thus, does not admit any blocking pairs with respect to $I'$. Therefore, any blocking pair admitted by $M$ with respect to $I$ must be a pair that was made unacceptable in the construction of $\succ'$ from $\succ$. The bound $\vert bp_i(M)\vert\leq \left\lceil\frac{d_i}{2}\right\rceil$ for any $a_i\in A$ follows immediately. Therefore, $\max_{a_i\in A} \vert bp_i(M)\vert\leq \left\lceil\frac{d_{\max}}{2}\right\rceil$ also follows immediately. Because the optimal solution value of {\sc Minimax-AlmostStable-sri} is known to be at least 1, the approximation result follows.

    Finally, we justify the stated time complexity of the algorithm. Irving's algorithm can be implemented to run in $O(n \cdot d_{\max})$ time \cite{irving_sr}. We argued already that the while loop terminates after at most $\frac{1}{2}n\cdot d_{\max}=O(n\cdot d_{\max})$ iterations. In each iteration of the while loop, we execute an inner for loop. By the bounded preference list length $d_{\max}$, the for loop is executed at most $O(d_{\max})$ many times. Each other operation within the two loop layers can be implemented to run in constant time. $Q$ can be implemented as a stack for efficient addition of new agents and removal of arbitrary agents, together with a Boolean array (indexed by agent IDs) to check whether an agent is, indeed, still in $Q$ (for efficient removal of agents from $Q$). The crossing number and side functions can be implemented through integer arrays indexed by agent IDs. The construction of $I'$ can clearly be implemented in $O(n\cdot d_{\max})$ time. Finally, it is well-known that the Gale-Shapley algorithm can be implemented in $O(n\cdot d_{\max}')$ time (where $d_{\max}'$ is the maximum preference list within $\succ'$). By construction, $d_{\max}'\leq d_{\max}$. Thus, we conclude that Algorithm \ref{alg:localapprox} runs in $O(n\cdot d_{\max}^2)$ time.
\end{proof}

\subsection{Inapproximability of Minimax Maximum Matchings in Bipartite Graphs}
\label{sec:inappproxsmi}

For {\sc Minimax-AlmostStable-Max-smi}, we will prove a gap-introducing reduction that is a modification of the reduction from the proof of Theorem \ref{thm:1-Maxsmi} to show that the problem cannot be approximated efficiently within a multiplicative square-root factor of the agents. Specifically, we replace the communication edges with sufficiently many copies of a new gadget that we introduce below and show that approximating {\sc Minimax-AlmostStable-Max-smi} well enough would imply that {\sf P=NP}. Note that this result also holds for {\sc Minimax-AlmostStable-Max-sri} by generalisation.

\begin{theorem}
\label{thm:inapproxsmi}
    {\sc Minimax-AlmostStable-Max-smi} is not approximable within $N^{1/2}$, where $N$ is the number of agents, unless {\sf P=NP}.
\end{theorem}
\begin{proof}
    Similar to the proof of Theorem \ref{thm:1-Maxsmi}, we will start with the reduction from {\sc (2,2)-e3-sat} to {\sc Minimax-AlmostStable-Max-smi}. While we use the same vertex and clause gadgets that we outlined in the proof of Theorem \ref{thm:1-Maxsmi}, we will also require an additional new \emph{connector gadget} that will replace the communication edges between vertex and clause gadgets.

    Formally, let $B$ be a {\sc (2,2)-e3-sat} instance with variables $V=\{V_1,V_2,\dots,V_n\}$ and clauses $C=\{C_1,C_2,\dots,C_m\}$. Without loss of generality, we can assume that $n\geq 3$. We will construct an {\sc smi} instance $J=(A,\succ)$ from $B$, where the set of agents $A$ has size
    \begin{equation}
        \label{eq:numagentssmi}
        \vert A\vert=\left\lceil\frac{8}{3}\left(1728n^2+43n+24\sqrt{6}\sqrt{864n^4+43n^3}\right)\right\rceil.
    \end{equation}
    The set $A$ consists of the disjoint union of agents 
    \begin{align*}
        A_v&=\{x_i^r,y_i^r\;;\;1\leq i\leq n\land 1\leq r\leq 4\},\\
        A_c&=\{c_j^s,p_j^s\;;\;1\leq j\leq m\land 1\leq s\leq 3\}\cup\{q_j,z_j\;;\;1\leq j\leq m\}, \text{ and}\\
        A_t&=\{t_{i,\gamma}^{\beta,k}\;;\; 1\leq i\leq n \land 1\leq \gamma\leq \gamma_* \land 1\leq \beta\leq 4\land 1\leq k\leq 12\}
    \end{align*}
    from the variable, clause, and connector gadgets (where $\gamma_*=2(\left\lfloor\vert A\vert^{1/2}\right\rfloor+1)$. We also have some remaining agents $A_r$ due to rounding in the floor and ceiling operations. We remark that, although the instance $J$ contains a large number of agents $\vert A\vert$, the construction of $J$ remains polynomial in the size of $B$.

    The following intuition may help understand the number of agents in Equation \ref{eq:numagentssmi}: every one of the $n$ variable gadgets contributes $8$ agents, every one of the $m$ clause gadgets contributes $8$ agents, and every variable gadget will be incident to four sets of $\gamma_*=2(\left\lfloor\vert A\vert^{1/2}\right\rfloor+1)$ connector gadgets, each contributing $12$ agents. We determined $\gamma_*$ to achieve the desired inapproximability result in the construction, as will become clear at the end of the proof. We already noted in a previous reduction that $B$ has a lot of structure and that $m=\frac{4}{3}n$. $A_r$ acts as a garbage collector for unnecessary agents created due to rounding issues, as will become clear below. Putting this information together, we have 
    \begin{align*}
        \vert A\vert&=\vert A_v\vert+\vert A_c\vert+\vert A_t\vert+\vert A_r\vert\\
        &=8n+8m+4\cdot 12\cdot n\cdot\gamma_*+\vert A_r\vert \\
        & = 8\left(n+\frac{4}{3}n\right)+48\cdot n\cdot 2\left(\left\lfloor\vert A\vert^{1/2}\right\rfloor+1\right)+\vert A_r\vert \\
        & = \frac{56}{3}n+96n\left(\left\lfloor\vert A\vert^{1/2}\right\rfloor+1\right)+\vert A_r\vert
    \end{align*}
    Then, as an intermediary step, we solve the equation $x=\frac{56}{3}n+96n(x^{1/2}+1)$ for $x$, which yields the function $x(n)=\frac{8}{3}(1728n^2+43n+24\sqrt{6}\sqrt{864n^4+43n^3})$. Here, even when $n\in\mathbb N$, $x(n)$ might not be an integer value. Hence, we take the ceiling of this function to arrive at the size of $A$ that we specified in Equation \ref{eq:numagentssmi}.
    Now, we can calculate $\gamma_*$ explicitly from $\vert A\vert$, i.e., 
    \begin{align*}
        \gamma_*&=2\left(\left\lfloor\vert A\vert^{1/2}\right\rfloor+1\right)\\
        &= 2\left(\left\lfloor\sqrt{\left\lceil \frac{8}{3}\left(1728n^2+43n+24\sqrt{6}\sqrt{864n^4+43n^3}\right) \right\rceil}\right\rfloor+1\right),
    \end{align*} 
    and thus 
    \begin{align*}
        \vert A_t\vert &= 48\cdot n\cdot\gamma_*\\
        &= 96 n\left(\left\lfloor\sqrt{\left\lceil \frac{8}{3}\left(1728n^2+43n+24\sqrt{6}\sqrt{864n^4+43n^3}\right) \right\rceil}\right\rfloor+1\right).
    \end{align*}
    Now, it remains to specify the size of $A_r$, which we can do as follows:    
    \begin{align*}
        \vert A_r\vert&=\vert A\vert-(\vert A_v\vert+\vert A_c\vert+\vert A_t\vert) \\
        &= \left\lceil \frac{8}{3}\left(1728n^2+43n+24\sqrt{6}\sqrt{864n^4+43n^3}\right) \right\rceil \\
        &\;\;\;\;\;\;- \left(\frac{56}{3}n + 96 n\left(\left\lfloor\sqrt{\left\lceil \frac{8}{3}\left(1728n^2+43n+24\sqrt{6}\sqrt{864n^4+43n^3}\right) \right\rceil}\right\rfloor+1\right)\right).
    \end{align*} 
    Notice that all components in the equation except of $-\frac{56}{3}n$ are integers and that, by the structure of {\sc (2,2)-e3-sat} instances, it must be that $\frac{2}{3}n\in \mathbb N$. Thus, because $m=\frac{4}{3}n$ must also be a positive integer, we can assume that $n\in3\mathbb N$, in which case $\vert A_r\vert$ is an integer value. We need to verify that $\vert A_r\vert$ is non-negative.

    \begin{claim}
        \label{claim:arnonneg}
        For any integer $n\geq 1$, $\vert A_r\vert> 0$.
    \end{claim}
    \begin{proof}[Proof of Claim \ref{claim:arnonneg}]
    \renewcommand{\qedsymbol}{$\blacksquare$}

    First, we define 
    \begin{align*}
    R(n) &=24\sqrt{6}\sqrt{864n^4+43n^3},\\
    a(n) &=1728n^2+43n\text{, and}\\
    C(n)&=\left\lceil\frac{8}{3}\left(a(n)+R(n)\right)\right\rceil.
    \end{align*}
    Notice that 
    \begin{align*}
    R(n)^2 &= 2985984n^4+148608n^3 \text{, and}\\
    a(n)^2&=2985984n^4+148608n^3+1849n^2.
    \end{align*}
    Therefore, 
    $$R(n)^2=a(n)^2-1849n^2,$$
    and thus, by $n\geq 1$,
    \begin{align*}
        (a(n)-1)^2 &=a(n)^2-2a(n)+1 \\
        &= 2985984n^4+148608n^3-1607n^2 -86n+1\\
        &= R(n)^2-1607n^2 -86n+1\\
        &<R(n)^2\\
        &< a(n)^2,
    \end{align*}
    so
    $$a(n)-1<R(n)<a(n).$$
    Clearly, there must exist some $0<\varepsilon<1$ such that 
    $$ R(n)=a(n)-1+\varepsilon.$$
    Thus, substituting this expression for $R(n)$ in $C(n)$, we get
    \begin{align*}
    C(n) &= \left\lceil\frac{8}{3}\left(2a(n)+\varepsilon-1\right)\right\rceil\\
    &= \left\lceil
    9216n^2+229n+\frac{n}{3}+\frac{8}{3}\varepsilon-\frac{8}{3}\right\rceil.
    \end{align*}
    We define 
    \begin{align*}
    M(n)&=9216n^2+229n \text{, and}\\
    k(n)&=\left\lceil \frac{1}{3}(n+8\varepsilon-8)\right\rceil.
    \end{align*}
    Clearly $M(n)\in\mathbb N$ because $n$ is a positive integer, so we can take $M(n)$ out of the ceiling in $C(n)$ and simply state
    $$C(n)=M(n)+k(n).$$
    Notice that, for $n\geq 1$ and $0<\varepsilon<1$, we have that 
    $$-2\leq k(n)\leq \left\lceil\frac{n}{3}\right\rceil.$$
    Also, observe that 
    \begin{align*}
    (96n+1)^2 &= 9216n^2+192n+1 \text{, and}\\
    (96n+2)^2 &= 9216n^2+384n+4.
    \end{align*}
    Hence, as $n\geq 1$,
    $$(96n+1)^2 < C(n) < (96n+2)^2 \text{, and}$$
    $$\left\lfloor\sqrt{C(n)}\right\rfloor=96n+1.$$
    Thus, finally, 
    \begin{align*}
    \vert A_r\vert &= C(n)-\frac{56}{3}n-96n\left\lfloor\sqrt{C(n)}\right\rfloor-96n\\
    &= M(n)+k(n)-\frac{56}{3}n-96n(96n+1)-96n\\
    &= 9216n^2+229n+k(n)-\frac{56}{3}n-9216n^2-96n-96n\\
    &= \frac{55}{3}n+k(n)\\
    &>18n-2 \text{, by $k(n)\geq-2$}\\
    &>0 \text{, by $n\geq 1$}.
    \end{align*}    
    \end{proof}
    
    For each variable $V_i\in V$, we construct the variable gadget $G_{V_i}$ as in the proof of Theorem \ref{thm:1-Maxsmi}, and, for each clause $C_j\in C$, we construct the clause gadget $G_{C_j}$ as in the proof of Theorem \ref{thm:1-Maxsmi}.

    Now, furthermore, for each $x_i^r$ agent in each variable gadget $G_{V_i}$, we construct $\gamma_*$ many \textbf{connector gadgets} $G_{T_{i,\gamma}^r}$ (where $1\leq \gamma\leq \gamma_*$) consisting of the twelve agents $t_{i,\gamma}^{r,1}, t_{i,\gamma}^{r,2},\dots,t_{i,\gamma}^{r,12}$, with preferences as follows:
    \[
    \begin{minipage}{0.45\textwidth}
    \begin{align*}
    t_{i,\gamma}^{r,1} &: t_{i,\gamma}^{r,2} \; c(t_{i,\gamma}^{r,1}) \; t_{i,\gamma}^{r,12} \\
    t_{i,\gamma}^{r,2} &: t_{i,\gamma}^{r,3} \; t_{i,\gamma}^{r,1} \\
    t_{i,\gamma}^{r,3} &: t_{i,\gamma}^{r,4} \; t_{i,\gamma}^{r,2} \\
    t_{i,\gamma}^{r,4} &: t_{i,\gamma}^{r,5} \; t_{i,\gamma}^{r,3} \\
    t_{i,\gamma}^{r,5} &: t_{i,\gamma}^{r,6} \; t_{i,\gamma}^{r,4} \\
    t_{i,\gamma}^{r,6} &: t_{i,\gamma}^{r,5} \; t_{i,\gamma}^{r,7}
    \end{align*}
    \end{minipage}
    \hfill
    \begin{minipage}{0.45\textwidth}
    \begin{align*}
    t_{i,\gamma}^{r,7} &: t_{i,\gamma}^{r,6} \; x_{i,\gamma}^r \; t_{i,\gamma}^{r,8} \\
    t_{i,\gamma}^{r,8} &: t_{i,\gamma}^{r,7} \; t_{i,\gamma}^{r,9} \\
    t_{i,\gamma}^{r,9} &: t_{i,\gamma}^{r,8} \; t_{i,\gamma}^{r,10} \\
    t_{i,\gamma}^{r,10} &: t_{i,\gamma}^{r,9} \; t_{i,\gamma}^{r,11} \\
    t_{i,\gamma}^{r,11} &: t_{i,\gamma}^{r,10} \; t_{i,\gamma}^{r,12} \\
    t_{i,\gamma}^{r,12} &: t_{i,\gamma}^{r,11} \; t_{i,\gamma}^{r,1}
    \end{align*}
    \end{minipage}
    \]
    where the $c(t_{i,\gamma}^{r,1})$ entry will be specified later. The construction is illustrated in Figure \ref{fig:connectorsmi}. $G_{T_{i,\gamma}^r}$ only admits two perfect matchings: 
    \begin{align*}
        M^1_{i,\gamma,r}&=\{\{t_{i,\gamma}^{r,1},t_{i,\gamma}^{r,2}\},\{t_{i,\gamma}^{r,3},t_{i,\gamma}^{r,4}\},\{t_{i,\gamma}^{r,5},t_{i,\gamma}^{r,6}\},\{t_{i,\gamma}^{r,7},t_{i,\gamma}^{r,8}\},\{t_{i,\gamma}^{r,9},t_{i,\gamma}^{r,10}\},\{t_{i,\gamma}^{r,11},t_{i,\gamma}^{r,12}\}\} \text{ and}\\
        M^2_{i,\gamma,r}&=\{\{t_{i,\gamma}^{r,1},t_{i,\gamma}^{r,12}\},\{t_{i,\gamma}^{r,2},t_{i,\gamma}^{r,3}\},\{t_{i,\gamma}^{r,4},t_{i,\gamma}^{r,5}\},\{t_{i,\gamma}^{r,6},t_{i,\gamma}^{r,7}\},\{t_{i,\gamma}^{r,8},t_{i,\gamma}^{r,9}\},\{t_{i,\gamma}^{r,10},t_{i,\gamma}^{r,11}\}\},
    \end{align*}
    where $M^1_{i,\gamma,r}$ admits no blocking pairs within the gadget and $M^2_{i,\gamma,r}$ admits the blocking pair $\{t_{i,\gamma}^{r,5},t_{i,\gamma}^{r,6}\}$ within the gadget (not counting blocking pairs that may also involve $c(t_{i,\gamma}^{r,1})$ or $x_i^r$ entries).

    \begin{figure}[!tbh]
        \centering
        \includegraphics[width=11cm]{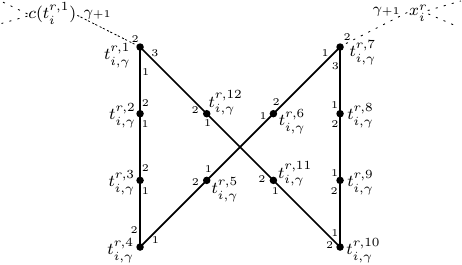}
        \caption{Illustration of the connector gadget construction for $G_{T_{i,\gamma}^r}$}
        \label{fig:connectorsmi}
    \end{figure}
    
    We now connect the gadgets as follows. For each vertex gadget $G_{V_i}$, we replace the entry $c(x_i^r)$ in the preference list of $x_i^r$ (from the previous construction in the proof of Theorem \ref{thm:1-Maxsmi}) with $\gamma_*$ many entries $t_\gamma(x_i^r)$ (for $1\leq \gamma\leq \gamma_*$). Specifically, let $t_\gamma(x_i^r)=t_{i,\gamma}^{r,7}$. In the construction of $G_{T_{i,\gamma}^r}$, we already included $x_i^r$ as the second-ranked entry in the preference list of $t_{i,\gamma}^{r,7}$. Thus, the edge that previously connected the variable gadget to the clause gadget was replaced by edges that now connect the variable gadget to the $\gamma_*$ many connector gadgets.
    
    For the links between connector gadgets and clause gadgets, let $C_j^s$ correspond to the $s$th literal in $C_j\in C$, which is either an unnegated or a negated occurrence of some variable $V_i$. Replace each $x(c_j^s)$ entry from the previous construction in the proof of Theorem \ref{thm:1-Maxsmi} with $\gamma_*$ many entries $t_\gamma(c_j^s)$ (for $1\leq \gamma\leq \gamma_*$). Now, if processing $B$ from left to right, $C_j^s$ is the
    \begin{itemize}
        \item first unnegated occurrence of $V_i$, then let $t_\gamma(c_j^s)=t_{i,\gamma}^{1,1}$ (for all $1\leq \gamma\leq \gamma_*$) and let $c(t_{i,\gamma}^{1,1})=c_j^s$;
        \item second unnegated occurrence of $V_i$, then let $t_\gamma(c_j^s)=t_{i,\gamma}^{2,1}$ (for all $1\leq \gamma\leq \gamma_*$) and let $c(t_{i,\gamma}^{2,1})=c_j^s$;
        \item first negated occurrence of $V_i$, then let $t_\gamma(c_j^s)=t_{i,\gamma}^{3,1}$ (for all $1\leq \gamma\leq \gamma_*$) and let $c(t_{i,\gamma}^{3,1})=c_j^s$;
        \item second negated occurrence of $V_i$, then let $t_\gamma(c_j^s)=t_{i,\gamma}^{4,1}$ (for all $1\leq \gamma\leq \gamma_*$) and let $c(t_{i,\gamma}^{4,1})=c_j^s$.
    \end{itemize} 

    Informally, we simply replace the communication edges between the vertex and the clause gadgets in the previous construction with $\gamma_*$ occurrences of a connector gadget copy in this construction. Lastly, we leave the preference lists of all agents in $A_r$ (if any exist) empty (essentially discarding them as they can never be matched). This is illustrated in Figure \ref{fig:gapgadgets}.

    \begin{figure}[!tbh]
        \centering
        \includegraphics[width=11cm]{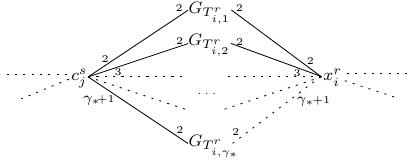}
        \caption{Illustration of the gap-introducing gadget connections}
        \label{fig:gapgadgets}
    \end{figure}

    Now, for the sake of argument below, we also construct the {\sc smi} instance $I=(A^I,\succ^I)$ from $B$ as in the proof of Theorem \ref{thm:1-Maxsmi}. From Claims \ref{claim:birosm}-\ref{claim:nocommblocking} we know that $B$ is satisfiable $\Leftrightarrow$ $I$ admits a perfect matching with at most $n+m$ blocking pairs $\Leftrightarrow$ $I$ admits a perfect matching in which no communication edge between clause and variable gadgets is blocking $\Leftrightarrow$ $I$ admits a perfect matching $M$ with $\max_{a_r\in A^I}\vert bp_{a_r}^I(M)\vert\leq 1$. We will now establish the following claims.

    \begin{claim}
        \label{claim:gapsmi1}
        If $I$ admits a perfect matching $M$ with $\max_{a_r\in A^I}\vert bp_{a_r}^I(M)\vert\leq 1$ then $J$ admits a maximum-cardinality matching $M'$ such that $\max_{a_r\in A}\vert bp_{a_r}^J(M')\vert\leq 1$.
    \end{claim}
    \begin{proof}[Proof of Claim \ref{claim:gapsmi1}]
    \renewcommand{\qedsymbol}{$\blacksquare$}
        Suppose that $I$ admits a perfect matching $M$ with $\max_{a_r\in A^I}\vert bp_{a_r}^I(M)\vert\leq 1$. To start, let $M'=M$. Now, for every communication edge $\{x_i^r,c_j^{s}\}$ in $I$, either $M(x_i^r)\succ_{x_i^r}^Ic_j^{s}$ or $M(c_j^{s})\succ_{c_j^{s}}^Ix_i^r$, otherwise they would be a blocking pair (and we know by Claim \ref{claim:nocommblocking} that this cannot be the case). If $M(x_i^r)\succ_{x_i^r}^Ic_j^{s}$, then, for every $1\leq \gamma\leq\gamma_*$, add $M_{i,\gamma,r}^1$ to $M'$, where $M_{i,\gamma,r}^1$ is as defined below the construction of the connector gadgets. Otherwise, add $M_{i,\gamma,r}^2$ (also defined below the connector gadget) to $M'$. Consider the final matching $M'$. Clearly, $M'$ is a maximum-cardinality matching of $J$ ($M$ is a perfect matching of $I$, and the only unmatched agents in $J$, if any, are the agents in $A_r$ who can never be matched due to empty preference lists). Furthermore, for every agent $t_{i,\gamma}^{r,1}$, either $t_{i,\gamma}^{r,2}=M'(t_{i,\gamma}^{r,1})\succ_{t_{i,\gamma}^{r,1}} c_j^{s}$ or $p_j^{s}=M'(c_j^{s})\succ_{c_j^{s}} t_{i,\gamma}^{r,1}$. Similarly, for every agent $t_{i,\gamma}^{r,7}$, either $M'(t_{i,\gamma}^{r,7})=t_{i,\gamma}^{r,6}\succ_{t_{i,\gamma}^{r,7}} x_i^{r}$ or $M'(x_i^{r})\succ_{x_i^{r}} t_{i,\gamma}^{r,7}$. Hence, no $t_{i,\gamma}^{r,1}$ and no $t_{i,\gamma}^{r,7}$ can block with an agent outside the connector gadget, and we already argued below the construction of the connector gadgets that no such agent blocks with agents within the connector gadget. Hence, $\max_{a_r\in A}\vert bp_{a_r}^J(M')\vert=\max_{a_r\in A^I}\vert bp_{a_r}^I(M)\vert\leq 1$ as required.
    \end{proof}

    \begin{claim}
        \label{claim:gapsmi2}
        If, for every perfect matching $M$ of $I$, $\max_{a_r\in A^I}\vert bp_{a_r}^I(M)\vert> 1$, then, for every maximum-cardinality $M'$ of $J$, $\max_{a_r\in A}\vert bp_{a_r}^J(M')\vert\geq \frac{\gamma_*}{2}$.
    \end{claim}
    \begin{proof}[Proof of Claim \ref{claim:gapsmi2}]
    \renewcommand{\qedsymbol}{$\blacksquare$}
        Clearly, every maximum-cardinality matching of $J$ consists of a perfect matching of $I$ together with either $M_{i,\gamma,r}^1$ or $M_{i,\gamma,r}^2$ for every $1\leq i\leq n$, $1\leq\gamma\leq \gamma_*$ and $1\leq r\leq 4$.
        
        We know from Claims \ref{claim:birosm}-\ref{claim:nocommblocking} that $I$ admits a perfect matching $M$ with $\max_{a_r\in A^I}\vert bp_{a_r}^I(M)\vert\leq 1$ if and only if $I$ admits a perfect matching in which no communication edge between clause and variable gadgets is blocking. Thus, because $\max_{a_r\in A^I}\vert bp_{a_r}^I(M)\vert> 1$ for all perfect matchings $M$ of $I$, then, for every such $M$, there must exist a communication edge that is blocking. Formally, for some communication edge $\{x_i^r,c_j^{s}\}$ (in $I$), both $c_j^{s}\succ_{x_i^r}^IM(x_i^r)$ and $x_i^r\succ_{c_j^{s}}^IM(c_j^{s})$. However, by construction of $\succ_{x_i^r}$ and $\succ_{c_j^{s}}$ (where $c_j^{s}$ and $x_i^r$ are replaced by $\gamma_*$ many copies of $t_{i,\gamma}^{r,7}$ and $t_{i,\gamma}^{r,1}$, respectively), we must also have that $t_{i,\gamma}^{r,7}\succ_{x_i^r}M'(x_i^r)$ and $t_{i,\gamma}^{r,1}\succ_{c_j^{s}}M'(c_j^{s})$ for any maximum-cardinality matching $M'$ of $J$ where $M\subset M'$, and for any $1\leq \gamma\leq\gamma_*$. Now, to minimise $\max_{a_i\in\{x_i^r,c_j^{s}\}}\vert bp_{a_i}^J(M')\vert$, we can do no better than to add $M_{i,\gamma,r}^1$ for $1\leq\gamma\leq\frac{\gamma_*}{2}$ and $M_{i,\gamma,r}^2$ for $\frac{\gamma_*}{2}<\gamma\leq\gamma_*$ (recall that $\gamma_*$ is a multiple of 2 by construction). However, this necessarily yields at least $\frac{\gamma_*}{2}$ blocking pairs for each of $x_i^r$ and $c_j^{s}$, and therefore $\max_{a_i\in\{x_i^r,c_j^{s}\}}\vert bp_{a_i}^J(M')\vert\geq\frac{\gamma_*}{2}$. Thus, for every maximum-cardinality $M'$ of $J$, $\max_{a_r\in A}\vert bp_{a_r}^J(M')\vert\geq \max_{a_i\in\{x_i^r,c_j^{s}\}}\vert bp_{a_i}^J(M')\vert\geq\frac{\gamma_*}{2}$ as required.
    \end{proof}

    Note that, because we assume that $n\geq 3$, $\frac{\gamma_*}{2}>290>1$, so we always introduce a gap. Our claims above establish that if $B$ is a yes-instance (i.e., it is satisfiable), then the optimal value of {\sc Minimax-AlmostStable-Max-smi} is at most 1. On the other hand, if $B$ is a no-instance (i.e., it is not satisfiable), then the optimal value of {\sc Minimax-AlmostStable-Max-smi} is at least 
    $$\frac{\gamma_*}{2}=\frac{2(\lfloor\vert A\vert^{1/2}\rfloor+1)}{2}=\left\lfloor\vert A\vert^{1/2}\right\rfloor+1> \vert A\vert^{1/2}.$$
    Thus, any approximation algorithm that returns a solution that is at most a square root of the total number of agents away from the optimal solution could decide the {\sf NP-complete} problem {\sc (2,2)-e3-sat}. This establishes the inapproximability result as required.
\end{proof}

\section{Exact Integer Linear Programs and Experiments}
\label{sec:ilp}

In the previous section, we explored approximation algorithms as a way to deal with the strong intractability results for minimax almost-stability. However, we also encountered a strong intractability frontier there when seeking approximation guarantees with respect to the number of agents, as seen in Theorem \ref{thm:inapproxsmi}. Our final contribution consists of two compact integer linear programming formulations (ILP) that solve our central optimisation problems of interest to optimality and some experimental results showing that the minimum maximum number of blocking pairs per agent can be expected to be small.

\subsection{ILP Formulations}

Let $I=(A,\succ)$ be a problem instance and consider the ILP below.
{\allowdisplaybreaks
\begin{align*} 
\min &\; r \\
\text{s.t.} \sum_{a_j\in A\setminus\{a_i\}}x_{ij} &\leq 1 & \forall a_i\in A \\ 
\sum_{a_j : a_j\text{ not on }\succ_i}x_{ij} &= 0 & \forall a_i\in A \\ 
x_{ij}&=x_{ji}   & \forall a_i,a_j\in A \\
b_{ij}&=b_{ji}   & \forall a_i,a_j\in A \\
\sum_{a_k:a_k\succeq_ia_j}x_{ik}+\sum_{a_k:a_k\succeq_ja_i}x_{jk} + b_{ij} &\geq 1 &\forall a_i,a_j\in A\\
\sum_{a_j\in A\setminus\{a_i\}}b_{ij} &\leq r & \forall a_i\in A\\
x_{ij}, b_{ij} &\in\{0,1\} &\forall a_i,a_j\in A\\
r &\in\mathbb Z^{\geq 0}
\end{align*}}

It is easy to see that this ILP corresponds to {\sc Minimax-AlmostStable-sri}: we introduce binary variables $x_{ij}\in\{0,1\}$ to indicate matches between pairs of agents $\{a_i,a_j\}\in A$, as well as binary variables $b_{ij}\in\{0,1\}$ to indicate blocking pairs between agents $a_i$ and $a_j$. We enforce that the $x_{ij}$ variables, collectively, represent a matching among acceptable agents in $I$ (allowing for the possibility of an agent being matched to themselves), and enforce the almost-stability constraint $\sum_{a_k:a_k\succeq_ia_j}x_{ik}+\sum_{a_k:a_k\succeq_ja_i}x_{jk} + b_{ij} \geq 1 $ to indicate that either there exists some agent $a_k$ who is at least as good as $a_j$ (according to $\succ_i$) that $a_i$ is matched to, or vice versa $a_j$ is matched to someone at least as good as $a_i$, or finally the agents are blocking (i.e., $b_{ij}=1$). Now the minimax constraint is straightforward: no agent finds themself in more than $r$ blocking pairs. Finally, we minimise $r$ subject to these constraints.

When requiring, furthermore, that solution corresponds to a maximum-cardinality matching (e.g., for {\sc Minimax-AlmostStable-Max-smi}), we can change the objective function to maximise (rather than minimise) the function 
$$\max \; \left((\vert A\vert+1)\cdot \sum_{a_i\in A}\sum_{a_j\in A\setminus \{a_i\}}x_{ij}\right)-r$$
instead. Clearly $r\leq \vert A\vert$, so maximising the matching size is always the primary objective. Subject to finding a maximum-cardinality matching, we then minimise $r$ as required. We note that this also works for {\sc Minimax-AlmostStable-Max-sri}.

\subsection{Experimental Results}

In this paper, we set out to answer how and when we can find matchings that distribute the amount of instability across the set of agents, minimising the concentration of instability on individual agents and thus avoiding unfairness and a greater risk of unravelling. We will now verify using randomly generated {\sc sri} and {\sc smi} instances whether these matchings satisfying this distributional objective are likely to exist in practice, and whether our ILP models can be used to find them in a reasonable amount of time.

For our random experiments, we generated instances with preference lists sampled uniformly at random, a common practice in algorithmic experiments in this area (e.g., see references \cite{delorme2019mathematical, mertens15,glitzner2025empirics}). All implementations were written in Python, and all computations were performed on the {\tt fataepyc} cluster. We used four compute nodes at a time, each equipped with dual AMD EPYC 7643 CPUs and 2TB RAM. All code is publicly available, see reference \cite{experimentsCode}. The experiments are seeded and can thus be easily replicated. All ILP models are solved using the PuLP package with the GurobiPy solver. 

In the random instance generation, we controlled the number of agents $n\in\{50,100,200\}$, the preference list length $l\in\{5,15,25\}$, and whether the instance is bipartite or not (i.e., {\sc sri} or {\sc smi} instances).\footnote{For non-bipartite instances, it cannot be avoided in some configurations that the preference list lengths are not exactly $l$. For bipartite instances, we forced one side of agents to have a preference list length $l$, but the preference list length of agents on the other side depends on the preferences that are randomly generated and might also not be exactly $l$. An alternative approach would have been to generate random regular bipartite graphs and then sample preferences over neighbours.} All results are averaged over 3000 instances per configuration. Table \ref{table:expresults} summarises our key findings, indicating the average matching size, ratio of stable matchings, and the average minimum maximum number of blocking pairs per agent (which we denote by a-mm-bp) across different settings and for different problems. Note that we did not include any results for the {\sc Minimax} ILP in the {\sc smi} setting, as any {\sc smi} instance admits a stable matching when not requiring a maximum-cardinality matching. 

\begin{table}[!htb]
    \footnotesize
    \centering
    \caption{Experimental results: Averages}
    \begin{tabular}{c c c c c c c c c c c}
        \toprule
        & & \multicolumn{6}{c}{\sc sri} & \multicolumn{3}{c}{\sc smi} \\
        \cmidrule(lr){3-8} \cmidrule(lr){9-11}
        & & \multicolumn{3}{c}{\sc Minimax-Max} & \multicolumn{3}{c}{\sc Minimax} & \multicolumn{3}{c}{\sc Minimax-Max} \\
        \cmidrule(lr){3-5} \cmidrule(lr){6-8} \cmidrule(lr){9-11}
        &  & size & stable & a-mm-bp & size & stable & a-mm-bp & size & stable & a-mm-bp \\
        \midrule
        \multirow{3}{*}{$n=50$} 
        & $l=5$ & 25.00 & 4.30\% & 0.96 & 23.48 & 77.37\% & 0.23 & 24.89 & 6.03\% & 1.01 \\
        & $l=15$ & 25.00 & 48.17\% & 0.52 & 24.76 & 60.63\% & 0.39 & 25.00 & 89.00\% & 0.11 \\
        & $l=25$ & 25.00 & 66.30\% & 0.34 & 24.98 & 67.17\% & 0.33 & 25.00 & 100.00\% & 0.00 \\
        \midrule
        \multirow{3}{*}{$n=100$}
        & $l=5$ & 50.00 & 0.07\% & 1.45 & 47.02 & 78.20\% & 0.22 & 49.74 & 0.20\% & 1.50 \\
        & $l=15$ & 50.00 & 13.20\% & 1.08 & 49.40 & 51.77\% & 0.48 & 50.00 & 21.07\% & 0.90 \\
        & $l=25$ & 50.00 & 18.27\% & 0.85 & 49.88 & 54.30\% & 0.46 & 50.00 & 21.37\% & 0.79 \\
        \midrule
        \multirow{3}{*}{$n=200$}
        & $l=5$ & 100.00 & 0.00\% & 1.68 & 94.03 & 77.97\% & 0.22 & 99.38 & 0.00\% & 1.80 \\
        & $l=15$ & 100.00 & 2.17\% & 1.84 & 98.78 & 47.77\% & 0.52 & 100.00 & 2.50\% & 1.76 \\
        & $l=25$ & 100.00 & 6.67\% & 1.61 & 99.66 & 41.47\% & 0.59 & 100.00 & 5.70\% & 1.61 \\
        \bottomrule
    \end{tabular}
    \label{table:expresults}
\end{table}

We find that, on average, all {\sc sri} instances admit perfect matchings, even when $l=5$. However, there is a significant trade-off in stability: only around $4.3\%$ of these matchings are stable, with the remaining matchings requiring some agent to be in a blocking pair. Notably, though, when $n=50$, then, in expectation, it is always possible to find a perfect matching in which no agent is in more than one blocking pair. The expected maximum number of blocking pairs per agent decreases as $l$ increases, which can be expected as the ratio of stable maximum-cardinality matchings greatly increases. For $n=100$, the ratio of stable maximum-cardinality matchings is much lower, starting at only $0.07\%$ when $l=5$ and going up to only $18.27\%$ when all agents rank roughly a quarter of the other agents. Still, on average, it is possible to find a maximum-cardinality matching where every agent is in fewer than two blocking pairs. When eliminating the maximum-cardinality requirement, the results greatly improve -- with only a minor trade-off in size, more than half of the instances become solvable across all settings, and the expected maximum number of blocking pairs per agent becomes much fewer than 1.

For the {\sc smi} instances, we also find that stable maximum-cardinality matchings are rare when preference lists are short -- only around $0.2\%$ of instances admit such matchings when $n=100$ and $l=5$, and none of the generated instances do when $n=200$ and $l=5$. However, the situation improves significantly when the preference list length increases. Notice the special case when $n=50$ and $l=25$: in the {\sc smi} setting, this means that each side has $\frac{n}{2}=25=l$ agents, so preference lists are complete and thus a stable perfect matching must always exist. We highlight that, as expected, the ratio of instances that admit stable maximum-cardinality matchings is higher for {\sc smi} than for {\sc sri}, but interestingly the average maximum number of blocking pairs per agent is larger when $l=5$; this is likely because the edges are more concentrated in {\sc smi} instances compared to {\sc sri} instances with the same number of agents and the same preference list length, so maximising for size might require a stronger sacrifice in stability.

Beyond average numbers, it might be interesting to know the maximum of the minimax number of blocking pairs per agent under the maximum-cardinality requirement. We show data for this (in the columns titled m-mm-bp), as well as the average time required to solve the ILP models to optimality, in Table \ref{table:expresults2}. It is notable that the encountered worst case instances still have a small minimum maximum number of blocking pairs, certainly much fewer than the theoretical worst case instances that we constructed in Section \ref{sec:structure}, and that when eliminating the maximum-cardinality requirement, none of the 27,000 {\sc sri} instances across the nine different settings required any agent to be in more than one blocking pair. The data in the table also shows that our ILP models can be solved very quickly in the parameter regimes that we considered. Naturally, the number of variables and constraints grows proportionally to the number of agents and the preference list lengths, and thus, the required time to solve the ILP models grows proportionally too. However, even when $l=25$, all sets of instances can be solved in fewer than 5 seconds on average.

\begin{table}[!htb]
    \footnotesize
    \centering
    \caption{Experimental results: worst case mm-bp and average timing}
    \begin{tabular}{c c c c c c c c}
        \toprule
        & & \multicolumn{4}{c}{\sc sri} & \multicolumn{2}{c}{\sc smi} \\
        \cmidrule(lr){3-6} \cmidrule(lr){7-8}
        & & \multicolumn{2}{c}{\sc Minimax-Max} & \multicolumn{2}{c}{\sc Minimax} & \multicolumn{2}{c}{\sc Minimax-Max} \\
        \cmidrule(lr){3-4} \cmidrule(lr){5-6} \cmidrule(lr){7-8}
        &  & m-mm-bp & avg. time (s) & m-mm-bp & avg. time & m-mm-bp & avg. time (s) \\
        \midrule
        \multirow{3}{*}{$n=50$} 
        & $l=5$ & 2 &  $<0.38$ & 1 & $<0.35$ & 3 & $<0.32$ \\
        & $l=15$ & 1 & $<0.38$ & 1 & $<0.35$ & 1 & $<0.32$ \\
        & $l=25$ & 1 & 0.38 & 1 & 0.35 & 0 & 0.32 \\
        \midrule
        \multirow{3}{*}{$n=100$}
        & $l=5$ & 3 & $<1.07$ & 1 & $<1.02$ & 4 & $<0.82$ \\
        & $l=15$ & 2 & $<1.07$ & 1 & $<1.02$ & 2 & $<0.82$ \\
        & $l=25$ & 2 & 1.07 & 1 & 1.02 & 2 & 0.82 \\
        \midrule
        \multirow{3}{*}{$n=200$}
        & $l=5$ & 4 & $<4.36$ & 1 & $<4.63$ & 5 & $<3.23$ \\
        & $l=15$ & 5 & $<4.36$ & 1 & $<4.63$ & 4 & $<3.23$ \\
        & $l=25$ & 5 & 4.36 & 1 & 4.63 & 4 & 3.23 \\
        \bottomrule
    \end{tabular}
    \label{table:expresults2}
\end{table}

Overall, we can draw two positive conclusions from these experiments: in practice, it is unlikely to encounter instances that require more than one blocking pair per agent when preference lists are sufficiently large, and our ILP formulations are a very viable option to solve our {\sf NP-hard} optimisation problems of interest in practice, despite the strong intractability results that we encountered throughout this paper.

\section{Conclusion}
\label{sec:conclusion}

We studied a family of natural optimisation problems arising from stable matching theory and initiated the study of a minimax notion of almost-stability. Through examples and experiments, we showed that minimax almost-stability is a viable alternative to previously studied notions of almost-stability, bringing its own unique advantages and challenges.

Across bipartite and general settings, our results, a subset of which are summarised in Table \ref{table:results}, reveal a sharp contrast between very strong intractability and tractable special cases, such as when preference lists are very short. In particular, we highlight that while the problems aiming to minimise either blocking pairs or blocking agents can be solved in exponential time in the optimum parameter value (i.e., the problems lie in the complexity class {\sf XP} with respect to this parameter), this is not the case for our new problem variants, which are {\sf NP-hard} even for a fixed optimum parameter value. This naturally raises the question of whether other parameters can restore computational tractability when kept small.

\begin{table*}[!tbh]
    \centering
    \footnotesize
    \setlength{\aboverulesep}{0pt}
    \setlength{\belowrulesep}{0pt}
    \renewcommand{\arraystretch}{1.1}
    \renewcommand\cellgape{\Gape[2pt]}  
    \caption{Overview of key complexity results. Here, $d$ denotes the maximum preference list length (also denoted by $d_{\max}$ in other parts of the paper), and CL indicates complete preference lists. Our contributions are coloured in \textcolor{blue}{blue}. The following symbols are used: ${\top}$ (result holds trivially), $^*$ (we newly introduced these problems).}
    \begin{tabular}{c | c | c}
        \toprule
        & {\sc sri} & {\sc Max-smi}  \\
        \midrule
        \makecell{{\sc MinBP}\\($\kappa=\min \vert bp(M)\vert$)} 
                   & \makecell{{\sf P} ($d\leq 2$) \cite{biro12}, {\sf XP} ($\kappa$) \cite{abraham06} \\ 
                   {\sf NP-h} ($d\leq 3$ and CL) \cite{abraham06,biro12}, {\sf W[1]-h} ($\kappa$) \cite{chen17}}
                   & \makecell{{\sf P} ($d\leq 2$ and CL) \cite{biro_sm_10}, {\sf XP} ($\kappa$) \cite{biro_sm_10}\\
                   {\sf NP-h} ($d\leq 3$) \cite{biro_sm_10}} \\
        \midrule
        \makecell{{\sc MinBA}\\($\kappa=\min \vert ba(M)\vert$)} 
                   & \makecell{{{\sf P} ($d\leq 2$) \cite{glitznermanlovepref}}, {\sf XP} ($\kappa$) \cite{chen17}\\
                   {\sf NP-h} ($d\leq 5$ \cite{chen17} {and CL \cite{glitznermanlovepref}}), {\sf W[1]-h} ($\kappa$) \cite{chen17}}
                   & \makecell{{\sf P} ($d\leq 2$ and CL) \cite{biro_sm_10}, {\sf XP} ($\kappa$) \cite{biro_sm_10}\\
                   {\sf NP-h} ($d\leq 3$) \cite{biro_sm_10}} \\
        \midrule
        \makecell{{\sc Minimax}$^*$\\($\kappa=\min\max \vert bp_i(M)\vert$)} 
                   & \textcolor{blue}{\makecell{{\sf P} ($d\leq 2$) [T\ref{thm:sriexact2}]\\
                   {\sf NP-h} ($\kappa=1$, $d\leq 10$ and CL) [T\ref{thm:1-Maxbounded}-\ref{thm:1-Maxcomplete}]}}
                   & \textcolor{blue}{\makecell{{\sf P} ($d\leq 2$ and CL) [T\ref{thm:smiexact},$\top$]\\
                   {\sf NP-h} ($\kappa=1$, $d\leq 3$) [T\ref{thm:1-Maxsmi}]}} \\
        \bottomrule
    \end{tabular}
    \label{table:results}
\end{table*}

Beyond their algorithmic significance, these findings also carry implications for multi-agent systems. In many practical applications where centralised coordination mechanisms are sought, solutions need to be robust to agent incentives to deviate, even when full stability is unattainable, and fairness considerations are crucial. We presented a new approach and characterised the computational possibilities and limits of achieving such robustness. Finally, we showed that although the intractability results hold even under strong restrictions, the ILP models that we proposed can be solved very efficiently in practice, thus posing an applicable tool to real-world settings. Simultaneously, our experiments showed that a balanced solution that distributes the burden of instability well is usually attainable in practice, often requiring no agent to experience justified envy towards more than one other agent. While we do not argue that minimax almost-stability is strictly preferable to other notions of almost-stability, we believe that it should be an essential consideration in the design of real-world mechanisms and could be combined, in a weighted or hierarchical fashion, with other objectives such as welfare guarantees, matching size guarantees, or alternative notions of almost-stability.

Many other interesting directions remain open. For instance: 
\begin{enumerate}
    \item Does there exist a polynomial-time approximation algorithm for {\sc Minimax-AlmostStable-sri} with constant- or logarithmic-factor guarantees? We conjecture that the answer is no, but proving lower bounds appears to be a challenge. We were unable to come up with a construction that introduces arbitrarily large gaps or a reduction that preserves sufficiently large gaps to rule out these kinds of approximation algorithms.
    \item We conjecture that {\sc Minimax-AlmostStable-sri} cannot be solved in polynomial time even when $d\leq 7$ (under standard complexity-theoretic assumptions), but that the optimal value of the problem under this restriction is at most 1. This would suggest a total-search type problem. What is the complexity of this problem under these restrictions?
\end{enumerate}

\paragraph{Software and Data Availability} 
The software to generate the same instances that we used in our experiments and replicate our results is openly available on Zenodo, see reference \cite{experimentsCode}.

\paragraph{Funding}
Frederik Glitzner is supported by a Minerva Scholarship from the School of Computing Science, University of Glasgow. 


\paragraph{Acknowledgements}
The authors would like to thank the anonymous AAMAS and MATCH-UP reviewers and attendees for their very helpful comments and suggestions on earlier versions of this paper.

\bibliographystyle{reference_format}
\bibliography{papers}

\end{document}